\documentclass[11pt]{article}
\usepackage{graphicx} 
\usepackage[normalem]{ulem}
\input{macros}
\RequirePackage[numbers]{natbib}

\usepackage{multirow} %
\usepackage{booktabs}  
\usepackage{siunitx} 
\usepackage{bm}

\newtheorem*{theorem*}{Theorem}
\newcommand{\deteqv}[1]{\tilde{#1}}
\DeclareMathOperator{\Tr}{Tr}
\DeclareMathOperator{\Id}{Id}

\iaincommentsfalse

\title{Method of Moments Estimation of High-Dimensional Covariance Using a Parametric Model }
\author{%
Iain M. Johnstone\thanks{Department of Statistics, Stanford University; email: \texttt{imj@stanford.edu}.}\and
	Yuchen Wu\thanks{School of Operations Research and Information Engineering, Cornell University; email: \texttt{yuchen.wu@cornell.edu}.}\and
 Ran Xie\thanks{Department of Statistics, Stanford University; email: \texttt{ranxie@stanford.edu}.}\and
}

\begin{document}
\maketitle

\begin{abstract}
    We propose method-of-moments estimators for the eigenvalues of variance component covariance matrices in multivariate mixed effects models. Assuming a parametric form for the eigenvalue distribution, we focus on the high-dimensional regime where the number of predictors is large and comparable to the number of realizations of each random effect. In this setting, we show that the empirical moments of sum-of-squares matrices (e.g., MANOVA estimators of the covariance matrices) can be closely approximated by deterministic functions of the underlying parameters. This relationship enables the construction of consistent and asymptotically normal estimators via moment matching. Our approach is motivated by applications in quantitative genetics, where estimating genetic covariance components across multiple phenotypic traits is of central interest. We implement our method in a new python
    package \texttt{mlmm-mom}, and demonstrate how our method adapts to several common experimental designs in this domain.
    \end{abstract}

\tableofcontents    
\section{Introduction}

Linear mixed models (LMMs) are foundational tools in quantitative genetics, providing a flexible framework for modeling genetic and environmental sources of variation in complex traits \cite{fisher1919xv,henderson1975best,henderson1984applications,lynch1998genetics}. In particular, multi-trait LMMs are increasingly relevant in modern applications where phenotypes are collected on a large scale over thousands of traits simultaneously \cite{lande1979quantitative,jiang1995multiple,blows2015phenome}.

In this paper, we study multi-trait LMMs in a high-dimensional regime,
where the number of traits $p$ may be comparable to or exceed the
number of individuals $n$.

\begin{model}
\iain{ (i) Consider a $k$-level multi-trait  linear mixed model of the form
\begin{align}\label{eq:general_model}
    Y = X\beta + \sum_{r=1}^k U_r \alpha_r, \quad \alpha_r \sim N(0, \Id_{I_r} \otimes \Sigma_r),
\end{align}
where $Y \in \mathbb{R}^{n \times p}$ is the observed trait matrix, $X
\in \mathbb{R}^{n \times I_0}$ is the fixed effects design matrix,
$\beta \in \mathbb{R}^{I_0 \times p}$ is the fixed effect coefficient
matrix, $U_r \in \mathbb{R}^{n \times I_r}$ are known design matrices,
and $\alpha_r \in \mathbb{R}^{I_r \times p}$ are independent Gaussian random
effect matrices.}[MLMM][Perhaps it is better to give the model a name
so that it can be referenced in theorems later.]
\end{model}

Given observations $Y$ and known design matrices $(X, \{U_r\}_{r =
  1}^k)$, our goal is to estimate the unknown covariance matrices
$\{\Sigma_r\}_{r = 1}^k \subseteq \mathbb{R}^{p \times p}$. 
In high-dimensional regimes, these covariance matrices are high-dimensional objects and cannot be consistently estimated without imposing additional structural assumptions. 
In this work, we propose a simple yet effective parameterization of these covariance matrices. 
Specifically, we assume that the eigenvalue distribution of each covariance matrix $\Sigma_r$ is parameterized by a finite-dimensional parameter vector $\theta_r \in \mathbb{R}^{d_r}$.
Then the original estimation problem reduces to estimating a combined parameter vector $\theta = (\theta_1,\ldots,\theta_k) \in \mathbb{R}^D$, where $D = \sum_r d_r$.

A primary motivation for this work is the estimation of the additive genetic covariance matrix, typically denoted by $G$. The $G$-matrix is a central object in quantitative genetics, as it characterizes the genetic architecture of multivariate traits and determines the rate and direction of phenotypic evolution \cite{lande1979quantitative, walsh2009abundant}. In the context of the multi-trait mixed model \eqref{eq:general_model}, the estimation of $G$ often reduces to estimating specific covariance components $\{\Sigma_r\}_{r=1}^k$.

As a motivating example, consider a full-sibling design, a classical breeding scheme where traits are measured on offspring from multiple families. In this design, individuals within the same family share a common genetic background, allowing the total phenotypic variance to be partitioned into between-family and within-family components \cite{falconer1996introduction}. Specifically, the traits of the $j$-th sibling in the $i$-th family are modeled as
\begin{align}\label{eq:full.sib.ij}
    y_{ij} = \mu + a_i + e_{ij}\in\mathbb{R}^p, \quad a_i \sim N(0, \Sigma_1), \quad e_{ij} \sim N(0, \Sigma_2).
\end{align}
Here, $\mu \in \mathbb{R}^p$ is the global mean vector, $a_i$ represents the family-level random effect shared by all siblings in family $i$, and $e_{ij}$ is the residual effect representing individual-specific variation. Under the full sibling model, the $G$-matrix is related to the between-family covariance $\Sigma_1$ via $G \approx 2\Sigma_1$ \cite{falconer1996introduction, searle2009variance}. Thus, the problem of estimating the spectral features of $G$ is equivalent to making inference on the eigenvalues of $\Sigma_1$.

Computational scalability and theoretical characterization remain the primary challenges in estimating genetic covariance matrices within high-dimensional regimes. The advent of high-throughput technologies, such as microarrays, has enabled the simultaneous measurement of thousands of quantitative phenotypes, offering unprecedented opportunities to explore the extent of pleiotropy and the effective dimensionality of evolutionary responses \cite{lockhart2000genomics, blows2015phenome, mcguigan2014nature}. However, traditional statistical frameworks are ill-equipped for such data. Restricted Maximum Likelihood (REML), while standard in quantitative genetics, is typically restricted to low-dimensional settings ($p \leq 20$) due to the computational burden of the underlying non-convex optimization. Furthermore, REML estimators often suffer from significant bias when estimating key spectral features, such as the leading eigenvalues or the null space, in high dimensions \cite{pavlyshyn2022}. While the Multivariate Analysis of Variance (MANOVA) provides a simpler alternative, it frequently yields estimates that are not positive semi-definite and exhibits eigenvalue spreading, leading to inconsistent estimates of the spectral distribution \cite{fan2017eigenvalue}.

To address these computational challenges and establish a rigorous
theoretical foundation for high-dimensional inference, we propose a
method-of-moments (MoM) framework. Our approach utilizes moments of
quadratic forms in $Y$ that are that are designed to cancel out the
fixed effects.

\begin{model}
\iain{  (ii) Let  $B \in \mathbb{R}^{n \times n}$ be a non-zero positive
  semi-definite matrix that satisfies $BX = 0$. We study quadratic
  forms and empirical moments
  \begin{equation*}
    S_n = p^{-1} Y^T B Y, \qquad
    \hat{\mu}_{nl} = p^{-1} \Tr(S_n^l)
  \end{equation*}
Let $\hat{\mu}_n = (\hat{\mu}_{n1}, \ldots, \hat{\mu}_{nD})$ be the
first $D$ empirical moments of $S_n$.}[MLMMctd][Again highlight the
main definitions for future reference]
\end{model}

We show that in a high-dimensional regime where $n, p, I_1, \cdots, I_k \to \infty$, there exists an explicit deterministic function $\mathcal{F}_n$ such that
\begin{align*}
    \hat{\mu}_n=\mathcal{F}_n(\theta)+o_{a.s.}(1),
\end{align*}
which naturally leads to the following MoM estimator
\begin{align*}
    \hat{\theta}_n=\mathcal{F}_n^{-1}(\hat{\mu}_n).
\end{align*}
Under mild  assumptions, we show that the proposed MoM estimator is consistent and asymptotically normal, providing a principled approach to statistical inference for the spectral parameters of high-dimensional covariance components in multi-trait mixed linear models.

\subsection{Motivation from quantitative genetics}\label{sec:motivation}
In evolutionary genetics, 
a central object of interest is the additive genetic covariance matrix  $G$, 
which measures how traits genetically covary with each other. 
Because selection acts on whole organisms rather than on individual traits in isolation, the multivariate structure of genetic variation critically shapes evolutionary response. In particular, multivariate trait combinations with low or no genetic variation are expected to evolve slowly or not at all. 

Consider a vector of $p$ quantitative traits $z\in\mathbb{R}^p$. Let $P\in\mathbb{R}^{p\times p}$ denote the total phenotypic covariance matrix, which can be partitioned into several constituent components, including additive genetic, residual non-additive genetic and environmental effects \cite{falconer1996introduction}. 
If an episode of selection induces a change $s\in\mathbb{R}^p$ in the mean of $z$ in a population, denoted by $\bar{z}\in\mathbb{R}^p$, then the change inherited by the next generation is predicted by the multivariate breeder's equation \cite{lande1979quantitative}
\begin{align*}
    \Delta \bar{z} = GP^{-1}s.
\end{align*}

In the univariate case with $p=1$, the ratio $G/P$ expresses the extent to which phenotypes are determined by genes transmitted from the parents. This corresponds to the classical notion of heritability, and is of the greatest importance in breeding programs \cite{falconer1996introduction}. Writing $\beta=P^{-1}s$, commonly referred to as the vector of selection gradients \cite{lande1979quantitative}, this relationship can be expressed more compactly as $ \Delta \bar{z} = G\beta.$

Reliable estimation of the spectral structure of $G$ is therefore central to understanding evolutionary dynamics. To illustrate this, consider the spectral decomposition $G = \sum_i \lambda_i g_i g_i^{\top}$, which yields the representation \cite{walsh2009abundant}
\begin{align*}
    \Delta \bar{z} = \sum_{i}\lambda_ig_ig_i^T\beta,
\end{align*}
where $\lambda_i$ and $g_i$ denote the eigenvalues and corresponding eigenvectors of $G$, respectively. Closely related quantities of interest include the genetic variance along an arbitrary linear combination of traits $b \in \mathbb{R}^p$ \cite{lin1977heritability}, given by
\begin{align*}
    \sigma^2_g(b)=b^TGb.
\end{align*}
After normalization, this leads to the notion of evolvability along the direction of $b$ \cite{hansen2008measuring}, defined as 
\begin{align*}
    e(b)=\frac{b^TGb}{b^Tb}.
\end{align*}

Another quantity of fundamental interest is the rank of $G$, or equivalently the dimension of its null space, which encodes absolute biological constraints. A zero eigenvalue indicates the absence of additive genetic variance and hence evolutionary response in a particular direction of the trait space \cite{lande1979quantitative, walsh2009abundant}.
More generally, the rank of $G$ provides a natural criterion for assessing the extent of pleiotropy, the influence of a single gene on multiple traits. As a milder alternative to rank,  Kirkpatrick \cite{kirkpatrick2009patterns} introduced the effective number of dimensions, 
\begin{align*}
    n_D=\sum_{i=1}^n\frac{\lambda_i}{\lambda_n},
\end{align*}
where \( \lambda_1 \) is the largest eigenvalue of \( G \). This quantity measures the degree of ill-conditioning of \( G \) and can be interpreted as the reciprocal of the proportion of total genetic variance explained by largest eigenvalue.

In summary, estimating the additive genetic covariance matrix $G$ is central to quantitative genetics, and developing reliable and efficient methods to estimate $G$ from data remains a fundamental statistical challenge.

\subsection{Related Work}
The computational and statistical challenges of Restricted Maximum Likelihood (REML) have motivated several algorithmic developments for specific experimental designs. For balanced designs, Calvin and Dykstra \cite{calvin1991maximum} introduced an iterative algorithm based on convex duality that achieves rapid convergence in low-dimensional settings. More recently, Pavlyshyn et al. \cite{pavlyshyn2022} demonstrated that this approach can render REML feasible for $p=1000$ in the special case of balanced half-sib designs. Calvin and Dykstra \cite{CalvinDykstraREML1995} also extended their framework to include parametric covariance models for certain nested balanced designs. However, whether these algorithmic improvements can be generalized to high-dimensional regimes with more complex, unbalanced designs remains an open question.

Our work is also related to the literature on moment-based estimation for high-dimensional sample covariance matrices. In particular, Bai et al. \cite{baiyaochen2010} proposed a method-of-moments approach for Wishart-type sample covariance matrices where the population covariance has a fixed number of distinct eigenvalues, establishing the consistency and asymptotic normality of the resulting estimators. This setting can be viewed as a restricted instance of our framework with a single level of random effects ($k = 1$) and a specific choice of design matrix. Our framework extends this approach to multi-level nested random effects models ($k > 1$), where the combinatorial analysis of the moments becomes significantly more complex due to the interaction between multiple high-dimensional covariance components. Additional details are in the Supplementary Materials Sections
\ref{sec:matrix-normal}--\ref{app:half.sib}.

\subsection{Outline of paper}\label{sec:outline}
Section~\ref{sec:main} formalizes the model setup and states the main results. In Section~\ref{sec:main.results}, we state the proposed MoM estimators and establish their asymptotic properties under the assumption that $\Sigma_r$, $r=1,\ldots,k$, are simultaneously diagonalizable. When this assumption fails, we propose an empirical fix in Section~\ref{subsec:empirical.fix}. Section~\ref{sec:sequential.nested} specializes the construction of the estimators to nested models, yielding a computationally more efficient implementation, and verifies the regularity conditions required for consistency and asymptotic normality under three parametric eigenvalue decay models for the $\Sigma_r$'s. In Section~\ref{subsec:exp}, we present simulation studies under commonly used quantitative genetic designs to assess the finite-sample performance of the proposed estimators. Section~\ref{sec:free.prob} introduces preliminary results from free probability theory and random matrix theory. Section~\ref{sec:proof.main} contains the proofs of the main results.

\textit{Notations.} Unless otherwise stated, $\|\cdot\|$ denotes the
operator norm of a matrix.
The Jacobian matrix of $\mathcal{F}: \Theta \subset \R^D \to \R^D$ at $\theta$ is
denoted by $J \mathcal{F}(\theta)$.  
The notation $X_n(\omega)$ is used, when helpful, to highlight the
dependence of a random variable $X_n$ on the sample outcome $\omega$.
The indices $p, I_1, \ldots, I_k$ and matrices $Y, X, U_r, \Sigma_r,
B, F_r, \Lambda_r$, etc. are understood to depend on the sample size
$n$, while constants $c, C, D, L$ (not necessarily the same at each
appearance) do not.

\section{Model and main results}\label{sec:main}
Throughout the paper, we work under the following assumptions:
\begin{assum}\label{ass:minimum}
There exist constants $c,C,L> 0$ such that:
\begin{enumerate}
    \item As $n, p, I_1, \dots, I_k \to \infty$, the ratios $p/n$ and $I_r/n$ for each $r = 1, \dots, k$ satisfy $c < p/n < C$ and $c < I_r/n < C$.
    \item The matrices $B$, $\Sigma_r$, and $U_r$ have bounded operator norms, specifically $\|B\| < L$, $\|\Sigma_r\| < L$, and $\|U_r\| < L$ for all $r = 1, \dots, k$.
 \end{enumerate}
\end{assum}
Unless otherwise stated, $\|\cdot\|$ denotes the operator norm of a matrix. We further assume that the covariance matrices $\Sigma_r$ exhibit a well-behaved limiting spectral structure:
\begin{assum}\label{ass:param}
    \begin{enumerate}
        \item Assume that $\Sigma_1,\ldots,\Sigma_k$ are simultaneously diagonalizable by an orthonormal matrix $V$. Let $\Lambda_1, \ldots, \Lambda_k$ denote the resulting diagonal matrices, where $\Lambda_r = V^T \Sigma_r V$ for each $r = 1, \ldots, k$. 
 
        \item For each $r$, let $\Lambda_r =
            \diag(\lambda_{j,r})$. There exists a 
          piecewise continuous $g_r: [0,1] \to [0,\infty)$ such that
          $\lambda_{\lfloor px\rfloor,r} \to g_r(x)$ at continuity
          points of $g_r$. 
    \end{enumerate}
\end{assum}

In our simulations, we take $\lambda_{j,r} = g_r(j/p)$.

An immediate criterion for verifying simultaneous diagonalizability is provided by the following lemma:
\begin{lem}\label{lem:simu.diag}
    Let $\mathcal{G}$ be a family of real symmetric matrices. There
exists an orthonormal matrix $V$ such that $VAV^T$ is diagonal for all $A\in\mathcal{G}$ if and only if $AB=BA$ for all $A,B\in\mathcal{G}$.
\end{lem}
\begin{proof}
    This is a direct corollary of Theorem 4.1.6 in \cite{Horn_Johnson_1985}. 
\end{proof}

\begin{assum} \label{ass:theta}
    Suppose each function $g_r(\cdot)$ is parameterized by $\theta_r \in
    {\Theta_r \subset} \mathbb{R}^{d_r}$. Let $\theta := (\theta_1, \dots,
    \theta_k) \in 
    {\Theta = \prod_r \Theta_r \subset}
    \mathbb{R}^D$ denote the joint parameter vector, where $D =
    \sum_{r=1}^k d_r$ is the total number of parameters to be
    estimated. 
\end{assum}

\subsection{Main results}\label{sec:main.results}

To derive the deterministic equivalents for these moments and
characterize the mapping between the parameters $\theta$ and the
moments of $S_n$, we employ combinatorial tools from free probability
theory.
Some definitions and needed results are recalled in Section
  ~\ref{sec:free.prob}.
Let $\text{NC}_2(2l)$ denote the set of non-crossing pair partitions
of $[2l] = \{1, \dots, 2l\}$, and for each partition $\pi \in
\text{NC}_2(2l)$,
let $K(\pi) \in \text{NC}_2(2l)$ denote its
Kreweras complement 
\cite{bose2021random}. Each block of $K(\pi)$ consists entirely of
either odd or even indices; we call such blocks odd or even.
Let $K(\pi)[1]$ denote the collection of
odd blocks and $K(\pi)[2]$ the set of even ones.
The essential deterministic approximation 
for the moments of $S_n$, valid asymptotically, is given in the next
lemma, which is proved in Section \ref{sec:proof.main}.

\begin{lem}\label{lem:map.asym} 
    \begin{itemize}
        \item[(a)]  Under Model MLMM and Assumption \ref{ass:minimum}, 
        define $F_i=BU_iU_i^T$, $i=1,\ldots,k$,  and let $\varphi=p^{-1}\Tr $. Then as $n\rightarrow\infty$,
        \begin{align*}
            \hat{\mu}_{nl}
            =&\sum_{r_1,\ldots,r_l=1}^k
            \sum_{\pi\in\mathrm{NC}_2(2l)}
            \prod_{V\in K(\pi)[1]}
            \varphi\!\left({\textstyle\prod}_{v\in V} F_{r_{(i+1)/2}}\right)
            \prod_{V'\in K(\pi)[2]}
            \varphi\!\left({\textstyle\prod}_{v\in V'} \Sigma_{r_{v/2+1}}\right)
            + o_{a.s.}(1),
        \end{align*}
        where $r_{l+1}:=r_1$.
        \item[(b)] If in addition  Assumption \ref{ass:param} holds,
          let $\phi=\int_0^1 dx$, then the previous display holds with
          \begin{equation*}
         \varphi\!\left({\textstyle\prod}_{v\in V'}
           \Sigma_{r_{v/2+1}}\right)
         \qquad \text{replaced with} \qquad
         \phi(\textstyle{\prod}_{v\in V'} g_{r_{v/2+1}}).
          \end{equation*}
    \end{itemize}
\end{lem}

In the formulas, note that if $v \in [2l]$ is odd, then $(v+1)/2$ is an
integer in $[l]$, while if $v$ is even, then $v/2 + 1 \in [l]$,
remembering that $l + 1$ is interpreted as $1$.

Part (b) follows directly from (a) under Assumption
  \ref{ass:param}, for then Riemann approximation implies that for
 $R \subset [k]$, we have
\begin{equation}
  \label{eq:riemman-app}
  \varphi \Big( \textstyle{\prod}_{r \in R} \ \Sigma_{r} \Big)
  = \varphi \Big( \textstyle{\prod}_{r \in R} \ \Lambda_{r} \Big)
   = \phi \Big( \prod_{r \in R} \ g_{r} \Big) + o(1).
\end{equation}

\medskip
\textit{Examples.} For $l=1$, $\text{NC}_2(2)$ has a single pair
partition $\{\{1,2\}\}$, which is non-crossing. We write it $12$ for
short. The Kreweras complement $K(\pi)$ has two singleton blocks
$\{\{1\}, \{2\}\}$, written $K(\pi) = 1|2$ for short.
Thus $K(\pi)[1] = 1$ and $K(\pi)[2] = 2$, and
Lemma \ref{lem:map.asym}(b) yields the following deterministic equivalent:
\begin{align*}
    \hat{\mu}_{n1} =\sum_{r_1=1}^k \varphi(F_{r_1}) \phi(g_{r_1}) + o_{a.s.}(1).
\end{align*}

For $l=2$, there are two pair partitions: $\text{NC}_2(4) = \{ \pi_1 =
12|34, \pi_2 = 14|23 \}$ and, from Section \ref{sec:free.prob},
$K(\pi_1) = 1|24|3$, and $K(\pi_2) =
13|2|4$. Lemma \ref{lem:map.asym}(b) gives
\begin{equation*}
  \hat{\mu}_{n2} =\sum_{r_1, r_2 =1}^k \varphi(F_{r_1})
  \varphi(F_{r_2}) \phi(g_{r_1} g_{r_2})
     + \varphi(F_{r_1} F_{r_2}) \phi(g_{r_1})
     \phi(g_{r_2}) + o_{a.s.}(1).
\end{equation*}

The map $\mathcal{F}_n:\mathbb{R}^D\rightarrow \mathbb{R}^D$ described below model MLMM in the introduction has components $\mathcal{F}_{n,l}$, $l\in[D]$. Under Assumptions \ref{ass:minimum} and \ref{ass:param}, Lemma \ref{lem:map.asym}(b) describes $\mathcal{F}_{n,l}$ as the composition $\Psi_l\circ\Phi_l$ of a parameter embedding map $x=\Phi_l(\theta)$ and a moment spreading map $\Psi_l(x)$. We use notation $\bm{g}^{\bm{l}}$ for products $g_1^{l_1}\cdots g_k^{l_k}$, indexed by multi-indices $\bm{l}=(l_1,\ldots,l_k)$ with non-negative integer components $l_r$ and degree $|\bm{l}|=l_1+\ldots+l_k$. We write $\mathcal{L}_l[k]$ for the set of multi-indices with $k$ components and degree at most $l$. The right side of the display in Lemma \ref{lem:map.asym}(b)
is a polynomial $\Psi_l(x)$ in variables $x=(x^{\bm{l}})$, $x^{\bm{l}} = \phi(\bm{g}^{\bm{l}})$ for $\bm{l}\in \mathcal{L}_l[k]$, so that
\begin{align*}
    \hat{\mu}_{nl}=\Psi_{l}(\{\phi(\bm{g}^{\bm{l}}):\bm{l}\in \mathcal{L}_l[k]\})+o_{a.s.}(1),
\end{align*}
and the embedding map $\Phi_l$ maps $\theta$ to $\{\phi(\bm{g}^{\bm{l}}):\bm{l}\in \mathcal{L}_l[k]\}$.

Denote the first $D$ moments of $S_n$ by $\hat{\mu}_n=(\hat{\mu}_{n1},\ldots,\hat{\mu}_{nD})$, then 
\begin{align}\label{eq:Fn}
    \hat{\mu}_n = \mathcal{F}_n(\theta)+o_{a.s.}(1),
\end{align}
with $\mathcal{F}_n:\mathbb{R}^D\rightarrow\mathbb{R}^D$.  Finally, the method of moments estimator $\hat{\theta}_n$ of $\theta$ is defined as the solution to the moments equation
\begin{align}\label{eq:mom.construction}
    \hat{\mu}_n=\mathcal{F}_n(\hat{\theta}_n).
\end{align}

To study the asymptotic behavior of the MoM estimators $\hat{\theta}_n$, we first analyze the empirical moments $\hat{\mu}_n$, from which the behavior of $\hat{\theta}_n$ can be derived via their relationship in (\ref{eq:mom.construction}). We begin by introducing a deterministic equivalent distribution $\tilde{F}_n$ for the empirical spectral distribution $F^{S_n}$ of $S_n$.  Here, deterministic equivalence is in the sense that almost surely $F^{S_n}(x) - \tilde{F}_n(x) \to 0$ at each $x\in\mathbb{R}$.
This equivalent law $\tilde{F}_n$ is constructed from model quantities $B$, $U_r$, $\Sigma_r$, $r=1,\ldots,k$,  as described in detail in Section \ref{subsec:deteqv}. From $\tilde{F}_n$, we define a sequence of deterministic equivalent moments $\deteqv{\mu}_n$ such that as  $n\rightarrow\infty$,
$$
\hat{\mu}_n - \deteqv{\mu}_n\rightarrow0
$$
almost surely. Moreover, as we detail in Section \ref{subsec:clt}, Theorem \ref{thm:clt}, a central limit theorem holds for $\hat{\mu}_n$: there exist  $R_n\in\mathbb{R}^D$ and $P_n\in\mathbb{R}^{D\times D}$ such that 
\begin{align*}
    P_n^{-1/2}\left(p(\hat{\mu}_n-\deteqv{\mu}_n) - R_n\right)\rightarrow N(0,\Id _D)
\end{align*}
in distribution. We are now ready to present our main result.

\begin{theorem}[Main result]\label{thm:main}
    Under Model MLMM and Assumptions
      \ref{ass:minimum}--\ref{ass:theta}, with true value $\theta$,
    further assume that   
    \begin{enumerate}
        \item[(i)] $B$ is PSD, $BU_r\neq 0$ for each $r=1,\ldots,k$, and $BX=0$,
        \item[(ii)] $\Phi$ is twice continuously differentiable on
            $\Theta$,
        \item[(iii)] there exist $\rho>0, c>0$ and $n_0$ such that
          uniformly in $\theta' \in B(\theta,\rho) \subset \Theta$ and $n
          \geq n_0$,
          \begin{equation*}
            \big|\text{det}(J\mathcal{F}_n(\theta'))\big| \geq c > 0.
          \end{equation*}
    \end{enumerate}
    then
    \begin{enumerate}
           \item almost surely, for $n \geq n_1(\omega)$, there exists a
             unique solution $\hat{\theta}_n(\omega)$ to
             (\ref{eq:mom.construction}) that is consistent,
             $\hat{\theta}_n \to \theta$. 
        \item for $D_n:=(J\mathcal{F}_n(\theta))^{-1}$, as $n\rightarrow \infty$,
        \begin{align*}
            (D_nP_nD_n^T)^{-1/2}\left[p\left(\hat{\theta}_n-\theta\right)-Q_n\right]\convd N(0,\Id)
        \end{align*}
        for some $P_n\in\mathbb{R}^{D\times D}$ and $Q_n=D_nR_n+p(\mathcal{F}_n^{-1}(\deteqv{\mu}_n)-\theta)\in\mathbb{R}^{D}$, with $P_n$, $R_n$ specified in Theorem \ref{thm:clt} and the construction of $\deteqv{\mu}_n$ specified in Section \ref{subsec:deteqv}. 
        \end{enumerate}
    \end{theorem}

\begin{rem} \label{eq:plug-in.sd.bias}
For a given MoM estimator $\hat{\theta}_n$, we obtain plug-in estimates of the asymptotic covariance and bias terms. Let $\hat{D}_n=(J\mathcal{F}_n(\hat{\theta}_n))^{-1}$. As in Section \ref{subsec:deteqv}, write $B=\Gamma\Gamma^T$. Then $P_n$, $R_n$, and $\deteqv{\mu}_n$ depend on $\Gamma_r:=\Gamma^TU_rU_r^T\Gamma$ and $\Lambda_r$, $r=1,\ldots,k$. Given $\hat{\theta}_n$, define $\hat{\Lambda}_r=\text{diag}\bigl(\hat{g}_r(i/p;\hat{\theta}_n)\bigr)$
and compute $\hat{P}_n$, $\hat{R}_n$, and $\hat{\deteqv{\mu}}_n$ by replacing $\Lambda_r$ with $\hat{\Lambda}_r$ in these expressions. Finally, set $\hat{Q}_n=\hat{D}_n\hat{R}_n+p\bigl(\mathcal{F}_n^{-1}(\hat{\deteqv{\mu}}_n)-\hat{\theta}_n\bigr)$.
\end{rem}

\begin{rem}\label{rem:nested}
    In practice, solving the full system of $\sum_{r=1}^kd_r$ equations can be computationally intensive. However, for certain classes of mixed-effects models, such as nested models where $\text{col}(U_1)\subset\ldots\subset\text{col}(U_k)$, with $\mathrm{col}(U_r)$ denoting the column space of $U_r$, we can simplify the process. Specifically, we can construct matrices $B_1,\ldots,B_k$ such that each $B_r$ cancels out $U_1,\ldots,U_{r-1}$. This allows us to sequentially solve for each set of parameters $\theta_r$ using the previously estimated parameters $\hat{\theta}_{r+1},\ldots,\hat{\theta}_{k}$. As a result, instead of solving the entire system at once, we solve a smaller system of $d_r$ equations at each stage. Section \ref{sec:sequential.nested} provides a detailed example.
\end{rem}

Here we introduce three representative parameterizations for the limiting eigenvalue distributions $g_r(x)$, $r=1,\ldots,k$. For each model, we verify in Section~\ref{sec:sequential.nested} that conditions (ii) and (iii) of Theorem~\ref{thm:main} are satisfied.  

\begin{enumerate}[leftmargin=*, align=left, labelsep=0.6em]
    \item[\textbf{Model 1.}] Step function with $s$ mass points ($d_r=2s-1$). The eigenvalues are parameterized by a set of parameters $\theta$, which lies in the parameter space 
    $$\Theta:=\left\{\theta = (\tau_{1},\ldots,\tau_{s},\rho_{1},\ldots,\rho_{s-1}):\tau_1>\ldots>\tau_s\geq 0, \rho_j>0,{\textstyle\sum}_{j=1}^{s-1}\rho_j<1\right\},$$ 
    with $g_r(x) = \tau_1\mathbbm{1}[0\leq x\leq \rho_1]+{\textstyle\sum}_{j=2}^{s} \tau_j\mathbbm{1}[\sum_{i=1}^{j-1}\rho_i<x\leq \sum_{i=1}^{j}\rho_i]$, where $\rho_s:=1-\sum_{j=1}^{s-1}\rho_j$.
    \item[\textbf{Model 1N.}] Step function with null space ($d_r = 2s-2$). As for
    Model 1, but with $\tau_s \equiv 0$.
    \item[\textbf{Model 2.}] Polynomial decay ($d_r=2$). The eigenvalues are parameterized by $\theta\in\Theta:=\{\theta=(\tau,\rho):\tau>0,\rho>0\}$, with $g_{r}(x) = \tau(1-\rho x)^s_+$ for a given integer $s>0$. When $\rho>1$, the model has a nontrivial null space, taking up a fraction of $1-\rho^{-1}$.  
    \item[\textbf{Model 3.}] Exponential decay ($d_r=2$). The eigenvalues are parameterized by $\theta\in\Theta:=\{\theta=(\tau,\rho):\tau,\rho>0\}$, with $g_r(x) = \tau\exp(-\rho x)$.
\end{enumerate}

Model 1 corresponds to the parametric model studied in \cite{baiyaochen2010}, which considers a single level of randomness, i.e., $k=1$. While \cite{baiyaochen2010} does not require the eigenvalues $\{\tau_i\}$ to be ordered, it is often most natural in practice to parameterize the eigenvalue decay, with $\tau_1>\ldots>\tau_s\geq 0$. Model 3, which assumes exponentially decaying eigenvalues, is motivated by empirical observations of genetic covariance structures \cite{kirkpatrick1992measuring}.  Across all three models, there is a parameter that represents the value of the largest eigenvalue of $\Sigma_r$, but this will typically not be appropriate for modeling large
separated eigenvalues (`spikes') when they are present, see Section
\ref{subsec:empirical.fix} below. 
 In both Model 1 and Model 2, the proportion of the null space can be expressed explicitly in terms of model parameters. Specifically, in Model 1N, setting $\tau_s=0$, setting $\tau_s=0$ yields a null space fraction of $\rho_s=1-\sum_{j=1}^{s-1}\rho_j$. In Model 2, this quantity is directly given by $1-\rho^{-1}$ when $\rho>1$. In practical applications, it is often difficult to distinguish between small eigenvalues and zero, leading to the notion of a "nearly-null" space in quantitative genetics. In Model 3, due to exponential decay, the eigenvalues rapidly approach zero for sufficiently large values of $\rho$, effectively producing a "nearly-null" space.

\subsection{Empirical fix}\label{subsec:empirical.fix}
    When it is unclear whether $\Sigma_1,\ldots,\Sigma_k$ are
    simultaneously diagonalizable, we estimate the eigenvectors of
    each $\Sigma_r$ separately. Let $\Sigma_r=V_r\Lambda_rV_r^T$ be
    the spectral decomposition of $\Sigma_r$, where $\Lambda_r$ is
    diagonal with entries in nonincreasing order. Let
    $\hat{\Sigma}_r=Y^TM_rY$ denote the MANOVA estimator of
    $\Sigma_r$, for an appropriate choice of symmetric matrices
    $M_1,\ldots,M_k$ as in \cite{searle2009variance}. Write the
    spectral decomposition of each $\hat{\Sigma}_r$ as
    $\hat{\Sigma}_r=\hat{V}_r\hat{\Lambda}_r\hat{V}_r^T$, where
    $\hat{\Lambda}_r$ contains the eigenvalues in nonincreasing
    order. We use $\hat{V}_r$ as an empirical estimate of $V_r$.
    MANOVA eigenvectors are generally not consistent
      in high dimensions and may exhibit eigenvector aliasing
      \cite{fan2018spiked}. Nevertheless, in our setting they still
      yield method-of-moments estimators with good empirical accuracy,
      as shown in Section \ref{subsec:exp}. A speculative, but
      plausible explanation is that our estimator aggregates moment information and does not rely on recovering any single eigenvector accurately.

\paragraph*{Empirical eigenvectors}
Motivated by the expansion given in Lemma \ref{lem:map.asym}(a)---which does \textit{not} require simultaneous diagonalization---,
  we construct the mapping $\Psi_n$ similarly as in
  Section~\ref{sec:main} and adapt $\Phi$ to incorporate empirically
  estimated eigenvectors. We use notation $\Sigma_{\bm{w}}$ for
  products $\Sigma_{r_1}\Sigma_{r_2}\cdots\Sigma_{r_s}$, indexed by
  words $\bm{w}=r_1r_2\ldots r_s$ in letters $r_\nu \in [k]$
  of length $s$. We write $\mathcal{W}_l[k]$ for the collection of words of length $s\leq l$. In this case the display in Lemma \ref{lem:map.asym}(a) can be written as a polynomial $\Psi_l=\Psi_{\mathcal{W},l}$, such that
\begin{align*}
    \hat{\mu}_{nl}=\Psi_{l}(\{\varphi(\Sigma_{\bm{w}}):\bm{w}\in \mathcal{W}_l[k]\})+o_{a.s.}(1).
\end{align*}and the embedding map $\Phi=\Phi_{\mathcal{W}}$ maps $\theta$ to $\{\varphi(\Sigma_{\bm{w}}):\bm{w}\in \mathcal{W}_l[k]\}$ for each $l=1,\ldots,D$
For a given parameter vector $\theta$, define surrogate diagonal matrices
\[
    \Lambda_r'(\theta)=\mathrm{diag}\bigl(g_r(i/p;\theta_r)\bigr),
\]
where $g_r$ is the limiting function from Assumption~\ref{ass:param}. Using $\Lambda_r'(\theta)$ and $\hat{V}_r$, define $\hat{\Phi}_{\mathcal{W},l}$ from $\theta$ to the collection of moments
\begin{align*}
    \left\{\varphi\left({\textstyle\prod}_{r\in\bm{w}} \hat{V}_r\Lambda_r'\hat{V}_r^T\right):\bm{w}\in \mathcal{W}_l[k]\right\}.
\end{align*}
It is easy to see that $\hat{\Phi}_{\mathcal{W},l}$ is random, and is polynomial in the diagonal entries of $\Lambda_r'(\theta)$.
Let $\hat{\mathcal{F}}_n=\Psi_{\mathcal{W}}\circ\hat{\Phi}_{\mathcal{W}}$, and obtain the method of moments estimator $\hat{\theta}_n$ as the solution to
\begin{align}\label{eq:emp.fix.main}
    \hat{\mu}_n=\hat{\mathcal{F}}_n(\hat{\theta}_n).
\end{align}
Unlike the setting of Theorem~\ref{thm:main}, the mapping $\hat{\Phi}$, and hence $\hat{\mathcal{F}}_n$, is now random, as it depends on the data through the estimated eigenvectors $\hat{V}_r$. Consequently, the asymptotic results established in Theorem~\ref{thm:main} do not directly apply in this setting.

In contrast, when the oracle eigenvector matrices $V_1,\ldots,V_k$ are known, we may define $\Phi_{\mathcal{W}}$ as the deterministic mapping from $\theta$ to
\begin{align*}
    \left\{\varphi\left({\textstyle\prod}_{r\in\bm{w}} {V}_r\Lambda_r'{V}_r^T\right):\bm{w}\in \mathcal{W}_l[k]\right\}, \quad l=1,\ldots,D.
\end{align*}  
When $\Lambda_r'=\Lambda_r$ for all $r$, then $\Phi_{\mathcal{W}}$ simply maps $\theta$ to $\{\varphi(\Sigma_{\bm{w}}):\bm{w}\in \mathcal{W}_l[k]\}$ for each $l$.
Letting $\mathcal{F}_n = \Psi_{\mathcal{W}} \circ \Phi_{\mathcal{W}}$, the corresponding method of moments estimator $\hat{\theta}_n$ is defined as the solution to
\begin{align}\label{eq:emp.fix.oracle.main}
    \hat{\mu}_n=\mathcal{F}_n(\hat{\theta}_n).
\end{align}
This formulation may be viewed as a generalization of the simultaneous diagonalizability assumption in Assumption~\ref{ass:param}. Indeed, Assumption~\ref{ass:param} corresponds to the special case where $V_1 = \cdots = V_k$ is known, in which case
\begin{align*}
    \varphi\left({\textstyle\prod}_{r\in\bm{w}} \Sigma_r\right)=\phi\left({\textstyle\prod}_{r\in\bm{w}} g_r\right)+o_n(1).
\end{align*}
In this oracle setting, consistency and asymptotic normality follow by arguments analogous to those used in the proof of Theorem~\ref{thm:main}, under appropriate design assumptions and construction of the matrix $B$. One concrete example is the class of nested models without an intercept. Details are deferred to Supplementary Appendix~\ref{app:oracle.clt}, and a simulated example under the full-sibling model is presented in Section~\ref{subsec:exp}.

\paragraph*{Accommodation of spikes}
The empirical-eigenvector construction above also allows a simple extension to spiked spectra. For a single-spike illustration, let $\hat{s}_r$ be an estimator of the leading spike in component $r$, and define
\[
    \bar{\Lambda}_r'(\theta)=\mathrm{diag}\bigl(\hat{s}_r,\ g_r(i/(p-1);\theta_r)\bigr).
\]
Replacing $\Lambda_r'(\theta)$ by $\bar{\Lambda}_r'(\theta)$ while retaining the same empirical eigenvectors $\hat{V}_r$, we use the corresponding moments
\begin{align*}
    \left\{\varphi\!\left({\textstyle\prod}_{r\in\bm{w}} \hat{V}_r\bar{\Lambda}_r'(\theta)\hat{V}_r^T\right):\bm{w}\in \mathcal{W}_l[k]\right\}.
\end{align*}
The framework is compatible with any spike estimation procedure, with a simple default being the largest MANOVA eigenvalue.  We apply this extension to a real dataset on rice gene expression in Section~\ref{sec:real.data}.

\subsection{Specialization to nested models}\label{sec:sequential.nested} 
A mixed effects model is nested if the design matrices satisfy
\begin{align}
&I_0\leq I_1\leq ...\leq I_k\leq n,\label{eq:nested.cond.1}\\
&\text{col}(X)\subset\text{col}(U_1)\subset \ldots \subset\text{col}(U_k),\label{eq:nested.cond.2}
\end{align}
where $\text{col}$ denotes the column space. To simplify discussion of the sequential estimation method, we
assume that there 
exist positive semi-definite matrices $B_1,\ldots,B_k$ such
that $B_rX=0$ for all 
$r$, and such that the following conditions hold (Section
\ref{sec:construct.B.nested} has a construction).
\begin{assum}\label{assum:nested.indep}
    \begin{enumerate}
        \item[(i)] $B_rU_r\neq0$, $B_rU_q=0$, for all $r=1,\ldots,k$, $q<r$.
        \item[(ii)] $B_rU_tU_t^TB_s=0$ for all $r=1,\ldots,k$, $s>r$, $t\geq s$.
    \end{enumerate}
\end{assum}

For $r=1,\ldots,k$, define quadratic form matrices
\begin{align*}
    S_{nr}=\frac{1}{p}Y^TB_rY.
\end{align*}

By Assumption~\ref{assum:nested.indep}(i), each $S_{nr}$ only depends
on the random effects from levels $r,\ldots,k$:
\begin{align*}
    S_{nr}=\frac{1}{p}\left(\sum_{t=r}^k\alpha_t^TU_t^T\right)B_r\left(\sum_{t=r}^kU_t\alpha_t\right).
\end{align*}
Equivalently, $S_{nr}$ is the quadratic-form matrix associated with the
reduced model
\begin{align*}
    Y^{(r)}=U_r\alpha_r+\ldots+U_k\alpha_k,\quad S_{nr}=\frac{1}{p}{Y^{(r)}}^TB_rY^{(r)}.
\end{align*}
Set $D_r=d_r+\ldots+d_k$. Applying Lemma \ref{lem:map.asym}(b)
  to the above model for the purpose of estimating $\theta_r$ based on
  $\hat{\mu}_{n}^{(r)}=(p^{-1}\Tr S_{nr},\ldots,p^{-1}\Tr
  S_{nr}^{d_r})$ yields the approximation
 $\hat{\mu}_{n}^{(r)} = \mathcal{F}_{n}^{(r)}(\theta_r; \theta_{r+1}, \dots, \theta_k)+o_{a.s.}(1)$ where the map $\mathcal{F}_n^{(r)}:\mathbb{R}^{D_r}\rightarrow\mathbb{R}^{d_r}$has components $\mathcal{F}_{n,l}^{(r)}=\Psi_l^{(r)}\circ\Phi_l^{(r)}$ for $l\in[d_r]$. 

Then the nested configuration enables a sequential estimation procedure: we first estimate $\theta_k$ from $\hat{\mu}_{n}^{(k)} = \mathcal{F}_{n}^{(k)}(\hat{\theta}_k)$, and then recursively estimate $\theta_r$ for $r = k-1, \dots, 1$ by solving the $d_r$
estimating equations
$$
\hat{\mu}_{n}^{(r)} = \mathcal{F}_{n}^{(r)}(\theta; \hat{\theta}_{r+1}, \dots, \hat{\theta}_k),
$$
for $\hat{\theta}_r$, where $\hat{\theta}_{r+1}, \dots, \hat{\theta}_k$ are the estimates obtained in the preceding steps. Finally, we assemble the sequential MoM estimator
\begin{align*}
    \hat{\theta}_n'=(\hat{\theta}_1,\ldots,\hat{\theta}_k).
\end{align*}

Analogous to the general setting, the empirical moments $\hat{\mu}_{n}^{(r)}$ satisfy a central limit theorem when centered at the deterministic equivalents $\deteqv{\mu}_{n}^{(r)}$ defined in Section \ref{subsec:deteqv}. Specifically, Theorem \ref{thm:clt} implies the existence of covariance matrices $P_{nr} \in \mathbb{R}^{d_r \times d_r}$ and bias vectors $R_{nr} \in \mathbb{R}^{d_r}$ such that
\begin{align*}
    {P_{nr}}^{-1/2}\left(p(\hat{\mu}_{n}^{(r)}-\deteqv{\mu}_{n}^{(r)}) - R_{nr}\right)\convd N(0,\Id _{d_r}).
\end{align*}
The following lemma shows that Assumption~\ref{assum:nested.indep}(ii) makes the 
moment blocks independent.
\begin{lem}\label{lem:nested.indep}
    Under Model MLMM and Assumption~\ref{assum:nested.indep}, the matrices
    $S_{n1},\ldots,S_{nk}$ are mutually independent. Consequently, the moment
    vectors $\hat\mu_n^{(1)},\ldots,\hat\mu_n^{(k)}$ are mutually independent.
\end{lem}
\begin{proof}
    The proof is given in Section \ref{subsec:proof.nested.indep}. 
\end{proof}
Since the moment vectors are mutually independent for each $n$, the marginal
CLTs imply the joint CLT by the Cramér--Wold device. Let $\hat{\mu}_n = (\hat{\mu}_{n}^{(1)}, \dots, \hat{\mu}_{n}^{(k)})$ and $\deteqv{\mu}_n=(\deteqv{\mu}_{n}^{(1)},\ldots,\deteqv{\mu}_{n}^{(k)})$, then 
\begin{align}\label{eq:nested.clt.mu}
    {P_{n}}^{-1/2}\left(p(\hat{\mu}_{n}-\deteqv{\mu}_{n}) - R_{n}\right)\convd N(0,\Id _{D}),
\end{align} 
where
\begin{align*}
    R_n = \begin{bmatrix} R_{n1}^T &\ldots & R_{nk}^T \end{bmatrix}^T \in \mathbb{R}^D, \quad
    P_n = \text{diag}(P_{n1}, \dots, P_{nk}) \in \mathbb{R}^{D \times D}.
    \end{align*}
Notably, while the moment vectors are mutually independent, the resulting MoM estimators are generally not. Indeed, the estimators are obtained sequentially, with each $\hat{\theta}_r$ depending on the preceding estimates of $\hat{\theta}_{r+1},\ldots,\hat{\theta}_k$. The following lemma formalizes the asymptotic behavior of the MoM estimators.

\begin{lem}\label{lem:nested}
    \begin{enumerate}
        \item[(a)] Under Assumptions \ref{ass:minimum}, \ref{ass:param} and \ref{assum:nested.indep}(i), further assume that 
        \begin{enumerate}
        \item[(i)] $\Phi^{(1)},\ldots,\Phi^{(k)}$ are twice continuously differentiable,
        \item[(ii)]for $n$ sufficiently large, there exists $\rho>0$, such that for any $\theta'\in B(\theta,\rho)$, the absolute value of the determinant $$\big|\text{det}(J\mathcal{F}_n^{(r)}(\theta_r';\theta_{r+1}',\ldots,\theta_k'))\big|$$ is bounded below by a constant independent of $n$, for all $r=1,\ldots,k$, 
    \end{enumerate}
    then for $n$ sufficiently large, almost surely there exists a unique moment estimator $\hat{\theta}_n'$ such that $\hat{\theta}_n'\to\theta$.
    \item[(b)] If additionally Assumption \ref{assum:nested.indep}(ii) holds, then as $n\rightarrow \infty$, the estimator satisfies
    \begin{align*}
        (D_nP_nD_n^T)^{-1/2}\left[p\left(\hat{\theta}_n'-\theta\right)-Q_n\right]\convd N(0,\Id_D),
    \end{align*}
    where $D_n=(J\mathcal{F}_n(\theta))^{-1}$ 
    and $Q_n = D_n R_n + p(\mathcal{F}_n^{-1}(\deteqv{\mu}_n) - \theta)$. Here $P_n$, $R_n$ and $\deteqv{\mu}_n$ are defined in \eqref{eq:nested.clt.mu}. 
    \end{enumerate}
\end{lem}
\begin{proof}
    The proof is given in Section \ref{subsec:proof.nested}. 
\end{proof}

To facilitate the application of Lemma \ref{lem:nested}, we next verify that conditions (i) and (ii) are satisfied when $g_1,\ldots,g_k$ are modeled by one of the decay models introduced in Section \ref{sec:main}.

\begin{lem}\label{lem:param.bdd.jacobian}
Under Assumptions \ref{ass:minimum} and \ref{ass:param}, if 
\begin{enumerate} 
    \item there exists $c>0$ such that $\varphi({F}_r)\geq c$, for all $r=1,\ldots,k$ and sufficiently large $n$,
    \item the eigenvalue decay of $\Sigma_1,\ldots,\Sigma_k$ asymptotically follow any of the three parametrization models introduced in Section \ref{sec:main},
\end{enumerate}
then there exists a parameter space $\Theta$ and a constant $c' > 0$ such that for all $\theta \in \Theta$, $$|\text{det}(J{\mathcal{F}}_n^{(r)})(\theta_r;\theta_{r+1},\ldots,\theta_k)|>c'$$ for each $r = 1, \dots, k$ and sufficiently large $n$.

\end{lem}
\begin{proof}
    The proof is given in Supplementary Appendix \ref{app:proof.nested}.
\end{proof}

\subsubsection{Constructing $B_1,\ldots,B_k$}\label{sec:construct.B.nested}
Let $\Pi_U$ denote the orthogonal projection $U(U^TU)^{-1}U^T$, and
 for each $r \in [0,k]$, with the convention $U_0 = X$, let
$\mathcal{U}_r = \operatorname{col}(U_r)$. In the nested model
(\ref{eq:nested.cond.2}), write $\mathcal{U}_r = \mathcal{U}_{r-1}
\oplus \mathcal{V}_{r-1}$.
\paragraph*{Balanced designs}
We call the nested design balanced if $U_r^TU_r\propto\Id$ for each
$r$.
For $r \in [1,k]$, let $$B_r = \Pi_{U_r} - \Pi_{U_{r-1}}$$ be orthogonal
projection onto $\mathcal{V}_{r-1}$.
Let  
The nested structure \eqref{eq:nested.cond.2} implies $B_r X = 0$ for all $r$, and $B_r U_s = 0$ for all $s < r$. Since $(U_r^TU_r)^{-1}\propto \Id$ and thus $U_rU_r^T\propto U_r(U_r^TU_r)^{-1}U_r^T$, we trivially have that $B_rU_tU_t^TB_s=0$ for all $r=1,\ldots,k$, $s>r$, $t\geq s$.

\paragraph*{Unbalanced designs} The unbalanced case is generally more delicate. A simple exception is the
two-level full-sibling design, where $U_2=\Id$. Thus the canonical increment projections $B_1, B_2$ constructed in the balanced case still satisfies Assumption \ref{assum:nested.indep}(ii), with $B_1U_2U_2^TB_2=B_1B_2=0$. 

In the remainder of this section, we focus on unbalanced designs with
$k\geq 3$. We construct the matrices $B_1,\ldots,B_k$ recursively from
the top level down, with $B_r=Q_rQ_r^T$ for some matrix $Q_r$
with orthonormal columns. 
For each $r$, let $\mathcal{U}_r = \operatorname{col}(U_r)$ and
$E_r$ be a matrix whose columns form an
orthonormal basis of the increment space
$\mathcal{V}_r$.
We set $Q_k=E_k$ and 
$$
B_k=Q_kQ_k^T= \Pi_{U_k} - \Pi_{U_{k-1}}.
$$
Having constructed
$Q_{r+1},\ldots,Q_k$, form
\[
    M_r=
    \begin{pmatrix}
        Q_s^TU_tU_t^TE_r
    \end{pmatrix}_{s>r,\ t\ge s},
\]
where the blocks are stacked vertically over all pairs $(s,t)$ with $s>r$ and
$t\ge s$. Let $N_r$ have orthonormal columns spanning $\ker(M_r)$, and set
\[
    Q_r=E_rN_r,\quad B_r=Q_rQ_r^T.
\]
Then $Q_r$ spans the largest subspace of
$\mathcal{V}_{r-1}$, 
conditional on the higher-level choices, satisfying
\[
    Q_s^TU_tU_t^TQ_r=0
    \qquad (s>r,\ t\ge s).
\]
If the null space of $M_r$ is trivial, then the construction yields
$Q_r=0$,
meaning that no nonzero level-$r$ direction satisfies the imposed independence
constraints given the higher-level choices.

\section{Experiments}\label{subsec:exp} 
We evaluate the finite-sample performance of the proposed MoM method through simulation studies and a real-data analysis under the full-sibling design \eqref{eq:full.sib.ij}. Additional experiments are presented in the supplementary material. In Supplementary Appendix \ref{app:half.sib}, we consider a full-sib half-sib design involving seven unknown parameters. 

Consider a design with $I$ families of sizes $J = (J_1, \dots, J_I)^T \in \mathbb{R}^I$, totaling $n = \sum_{i=1}^I J_i$ individuals. Let $Y\in\mathbb{R}^{n\times p}$ denote the matrix obtained by stacking the trait measurements $y_{ij}$ by row. The model can then be compactly expressed as 
\begin{align}
    Y=X\beta+U_1\alpha_1+U_2\alpha_2,\label{eq:fullsib.exp1}
\end{align}
where $X\in\mathbb{R}^{n\times I_0}$ is the fixed effects design matrix, $\beta\in\mathbb{R}^{I_0\times p}$ is the fixed effect, and $U_2=\Id_n$ is the membership matrix that assigns each individual to the corresponding individual effect. {The block-diagonal matrix $U_1=\text{diag}(\mathbbm{1}_{J_1},\ldots,\mathbbm{1}_{J_I})\in\mathbb{R}^{n\times I}$  assigns
each sibling to the corresponding family, where each column $i$ specifies the $J_i$ siblings belonging to family $i$.} The matrices $\alpha_1 \in \mathbb{R}^{I \times p}$ and $\alpha_2 \in \mathbb{R}^{n \times p}$ contain the stacked family-level and individual-level random effects, respectively. Assume that $X$ has full column rank. When $\col(X)\subset\col(U_1)$, the model satisfies the nested structure defined by conditions \eqref{eq:nested.cond.1} and \eqref{eq:nested.cond.2}. 

Following the construction proposed in Section \ref{sec:sequential.nested}, we define the projection matrices
\begin{align*}
    B_1=\begin{cases}
        U_1(U_1^TU_1)^{-1}U_1^T,&\text{if }X=0\\
        U_1(U_1^TU_1)^{-1}U_1^T-X(X^TX)^{-1}X^T,&\text{if }X\neq0
    \end{cases},~ B_2=\Id_n-U_1(U_1^TU_1)^{-1}U_1^T.
\end{align*}
{Throughout the simulation study, we set $X=0$, which simplifies $B_1$ to $B_1=U_1(U_1^TU_1)^{-1}U_1^T$, corresponding to a model without fixed effects. This choice permits theoretical characterization of asymptotic normality under model misspecification in Scenarios 3 and 5. In the real data application of Section~\ref{sec:real.data}, we take $X=\mathbbm{1}_n$ to account for the population mean.} The corresponding sum-of-squares matrices are
\begin{align}
    S_{n1}&=\frac{1}{p}Y^TB_1Y=\begin{cases}
        \frac{1}{p}\sum_{i=1}^IJ_i\Bar{Y}_i\Bar{Y}_i^T,&~\text{if }X=0\label{eq:Sn1}\\
        \frac{1}{p}\sum_{i=1}^IJ_i(\Bar{Y}_i-\Bar{Y})(\Bar{Y}_i-\Bar{Y})^T,&~\text{if }X=\mathbbm{1}_n
    \end{cases}\\
    S_{n2}&=\frac{1}{p}Y^TB_2Y= \frac{1}{p}\sum_{i=1}^I\sum_{j=1}^{J_i}(y_{ij}-\Bar{Y}_i)(y_{ij}-\Bar{Y}_i)^T,\nonumber
\end{align} 
where $\Bar{Y}_i$ denotes the mean trait vector within the $i$-th family, and $\Bar{Y}$ denotes the mean trait vector across all samples. 
The standard MANOVA estimators for $\Sigma_1$ and $\Sigma_2$ are given by \cite[Chapter 3.6]{searle2009variance}
\begin{align}
    \hat{\Sigma}_1&=\frac{p}{L}\left(S_{n1}-\frac{I-I_0}{n-I}S_{n2}\right),\quad \hat{\Sigma}_2=\frac{p}{n-I}S_{n2},\label{eq:manova.fullsib}
\end{align}
which satisfy $\E[\hat\Sigma_1]=\Sigma_1$ and $\E[\hat\Sigma_2]=\Sigma_2$, where
\begin{align*}
    L=\begin{cases}
        n &\text{if }X=0,\\
        n-\Tr (U_1^TX(X^TX)^{-1}X^TU_1)=n-\frac{1}{n}\sum_{i=1}^IJ_i^2 &\text{if }X=\mathbbm{1}_n.
    \end{cases}
\end{align*}

\subsection{Simulation study under a simple three-parameter model}\label{subsec:sim.3param}
Assume that $\Sigma_1$ and $\Sigma_2$ take the following form:
    \begin{align*}
	\Sigma_1 = V_1\left[ \begin{array}{cc}
		\tau_1 \Id_{\lceil p \rho \rceil } & 0 \\
		0 &0
	\end{array} \right]V_1^T, \quad \Sigma_2=V_2\left[ \begin{array}{cc}
		\tau_2 \Id_{\lceil p {\rho_0} \rceil } & 0 \\
		0 &0
	\end{array} \right]V_2^T,
    \end{align*}
where $\tau_1,\tau_2>0$ and $\rho\in(0,1)$ are unknown parameters, and $\rho_0\in(0,1)$ is a known constant. 
{This parametric structure allows us to illustrate the application of Lemma~\ref{lem:map.asym} for $l=1,2,3$, and yields closed-form MoM estimators when estimating $\theta_2=(\tau_2)$ and $\theta_1=(\tau_1,\rho)$ sequentially following the procedure outlined in Section~\ref{sec:sequential.nested}.} 
For convenience of illustration, we further assume that $\rho<\rho_0$. If $\Sigma_1$ and $\Sigma_2$ are simultaneously diagonalizable with $V_1=V_2$, then Assumption~\ref{ass:param} holds with
\begin{align*}
	g_1(x)=\tau_1\mathbbm{1}[0\leq x\leq\rho],\quad g_2(x)=\tau_2\mathbbm{1}[0\leq x\leq {\rho_0}].
\end{align*}

We consider five scenarios that differ in eigenvector structure and estimation strategy. When $\Sigma_1$ and $\Sigma_2$ are simultaneously diagonalizable, we may assume without loss of generality that $V_1=V_2=\Id_p$. Under this setting, we compare joint estimation of $\theta=(\tau_1,\rho,\tau_2)$ as in Section~\ref{sec:main.results} with sequential estimation as in Section~\ref{sec:sequential.nested}, in which $\theta_2$ is estimated first, followed by $\theta_1$ conditional on $\hat{\theta}_2$.

We also examine the robustness to misspecification by considering settings where simultaneous diagonalizability fails. {Specifically, we take $V_2=\Id_p$ and $V_1=R(\omega)$, where $R(\omega)$ is a structured orthogonal matrix that applies rotations of angle $\omega$ to each coordinate pair $(i,p+1-i)$. That is, the entries in positions $(i,i)$, $(i,p+1-i)$, $(p+1-i,i)$, $(p+1-i,p+1-i)$ form the $2\times 2$ rotation block 
\begin{align*}
    \left[ \begin{array}{cc}
        \cos(\omega) & -\sin(\omega) \\
        \sin(\omega) & \cos(\omega)
    \end{array} \right]
\end{align*}
When $\omega=0$, $V_1=\Id_p$, and $\Sigma_1$ reduces to a diagonal matrix with eigenvalues in descending order. When $\omega=\pi/2$, the eigenvalues are arranged in ascending order. This is the most adversarial scenario, where $\Sigma_1$ and $\Sigma_2$ are both diagonal with eigenvalues sorted in opposite order.}
We then compare joint estimation under the incorrect simultaneous diagonalizability assumption with the empirical-fix procedure of Section~\ref{subsec:empirical.fix}, which estimates the eigenvector matrices via MANOVA prior to parameter estimation. {As a theoretical benchmark, we further include an oracle setting in which the true eigenvector matrices $V_1$ and $V_2$ are assumed known.} The five scenarios are summarized below: 
\begin{enumerate}[leftmargin=*, align=left, labelsep=0.6em]
    \item[{Scenario 1.}] $V_1=V_2=\Id_p$; sequential estimation of $\theta_2=(\tau_2)$ and $\theta_1=(\tau_1,\rho)$.
    \item[{Scenario 2.}] $V_1=V_2=\Id_p$; joint estimation of $\theta=(\tau_1,\rho,\tau_2)$.
    \item[{Scenario 3.}] $V_1=R(\omega)$, $V_2=\Id_p$; joint estimation without simultaneous diagonalizability. 
    \item[{Scenario 4.}] $V_1=R(\omega)$, $V_2=\Id_p$; empirical-fix procedure.
    \item[{Scenario 5.}] {$V_1=R(\omega)$, $V_2=\Id_p$; oracle-fix procedure.}
\end{enumerate}

\subsubsection{Construction of MoM estimators under Scenario 1}
We exploit the nested structure to construct a sequential estimation procedure in which $\theta_2$ is estimated first, followed by estimation of $\theta_1$ conditional on the plug-in estimate $\hat{\theta}_2$. This sequential approach yields closed-form MoM estimators under the full sibling model. 
We first construct the method of moments estimator $\hat{\theta}_2$. The matrix $S_{n2}$ can be alternatively interpreted as a sum of squares matrix constructed from the linear model
\begin{align*}
    Y^{(2)}=U_2\alpha_2,\quad S_{n2}=\frac{1}{p}{Y^{(2)}}^TB_2Y^{(2)}.
\end{align*}
Let $\hat{\mu}_{n}^{(2)}=\frac{1}{p}\Tr S_{n2}$ denote the first empirical moment of $S_{n2}$. From Example \ref{eg:kre}, we have that $\text{NC}_2(2)$ contains a single non-crossing pair partition $\pi=\{\{1,2\}\}$ with Kreweras complement $K(\pi)=\{\{1\},\{2\}\}$, so that $K(\pi)[1]=\{1\}$, $K(\pi)[2]=\{2\}$. Applying Lemma \ref{lem:map.asym}, we obtain the asymptotic expansion
\begin{align*}
    \hat{\mu}_{n}^{(2)} =_{a.s.} \varphi(B_2U_2U_2^T)\phi(g_2),
\end{align*}
where $a=_{a.s.}b\Leftrightarrow a=b+o_{a.s.}(1)$. Above, the second factor is given by
\begin{align*}
    \phi(g_2)=\int_0^{\rho_0}\tau_2dx=\rho_0\tau_2.
\end{align*}
Define the mapping ${\Phi}^{(2)}(x)=\rho_0x$, which maps $\tau_2$ to $\phi(g_2)$, and set ${\Psi}_n^{(2)}(x)=\varphi(B_2U_2U_2^T)x$. Let ${\mathcal{F}}_n^{(2)}={\Psi}_n^{(2)}\circ{\Phi}^{(2)}$. The method of moments estimator $\hat{\theta}_2=\hat{\tau}_2$ is then the solution to the equation
\begin{align*}
   \hat{\mu}_{n}^{(2)}={\mathcal{F}}_n^{(2)}(\hat{\theta}_2)
   \Leftrightarrow \hat{\mu}_{n}^{(2)} = \varphi(B_2U_2U_2^T)\rho_0\hat{\theta}_2.
\end{align*}
Since $B_2$ is a projection matrix with $\rank(B_2) = n - I$, we have $\varphi(B_2 U_2 U_2^T) = (n - I)/p$. This yields the closed-form estimator
\begin{align}\label{eq:hat.tau.2.main}
   \hat{\tau}_2= \hat{\theta}_2 = \frac{1}{\rho_0}\frac{p}{n-I}\hat{\mu}_{n}^{(2)}=\frac{1}{(n-I)\rho_0}\Tr S_{n2}.
\end{align}

Next, we derive the method-of-moments estimator for $\theta_1$. 
Consider the empirical moments $\hat{\mu}_{n}^{(1)}=\left(\frac{1}{p}\Tr S_{n1},\frac{1}{p}\Tr S_{n1}^2\right)$. 
Let
\begin{align*}
    F_1&=B_1U_1U_1^T,\quad F_2=B_1U_2U_2^T.
\end{align*}
By the enumeration of $\text{NC}_2(2)$ and $\text{NC}_2(4)$ with the corresponding Kreweras complements listed in Example \ref{eg:kre}, we
apply Lemma \ref{lem:map.asym} and obtain the asymptotic expansions
\begin{align}
    \begin{cases}
    \hat{\mu}_{n1}^{(1)}=_{a.s}&\sum_{r_1=1}^2\varphi(F_{r_1})\phi(g_{r_1})\\
    \hat{\mu}_{n2}^{(1)}=_{a.s}&\sum_{r_1=1}^2\sum_{r_2=1}^2\left[\varphi(F_{r_1})\varphi(F_{r_2})\phi(g_{r_1}g_{r_2})+\varphi(F_{r_1}F_{r_2})\phi(g_{r_1})\phi(g_{r_2})\right].
    \end{cases}\label{eq:hat.mu.12.main}
\end{align}
For each $\theta_1=(\tau_1,\rho)$, let ${\Phi}^{(1)}:\mathbb{R}^2\rightarrow\mathbb{R}^3$ be
\begin{align*}
    {\Phi}^{(1)}(\theta_1)=\left(\rho\tau_1,\hat{\tau}_2\rho\tau_1,\rho\tau_1^2\right),
\end{align*}
which maps $\theta_1$ to $\{\phi(g_1),\phi(\hat{g}_2g_1),\phi(g_1^2)\}$. For each $x=(x_1,x_2,x_3)\in\mathbb{R}^3$, define ${\Psi}_n^{(1)}=({\Psi}_{n1}^{(1)},{\Psi}_{n2}^{(1)})$, where 
\begin{align*}
    {\Psi}_{n1}^{(1)}(x)=&\varphi(F_{1})x_1+\varphi(F_{2})\phi(\hat{g}_{2}),\\
    {\Psi}_{n2}^{(1)}(x)=&\varphi(F_{1}^2)x_1^2+2\varphi(F_{1}F_{2})\phi(\hat{g}_{2})x_1+2\varphi(F_{1})\varphi(F_{2})x_2+\varphi(F_{1})^2x_3\\
    &+\varphi(F_{2})^2\phi(\hat{g}_{2}^2)+\varphi(F_{2}^2)\phi(\hat{g}_{2})^2.\nonumber
\end{align*}
Define ${\mathcal{F}}_n^{(1)}={\Psi}_n^{(1)}\circ{\Phi}^{(1)}$, then with $\hat{\theta}_2$ estimated, we solve for the method of moments estimators $\hat{\theta}_1=(\hat{\tau}_1,\hat{\rho})$ from the system of equations
\begin{align*}
    \hat{\mu}_{n}^{(1)}={\mathcal{F}}_n^{(1)}(\hat{\theta}_1).
\end{align*}
Consequently, together with \eqref{eq:hat.tau.2.main}, we have the closed form solutions 
\begin{align*}
    \hat{\tau}_2 &= \frac{1}{\rho_0}\frac{p}{n-I}\hat{\mu}_{n}^{(2)},\quad \hat{\rho}=\frac{\hat{\mu}_{n1}^{(1)}-\varphi(F_{2})\rho_0\hat{\tau}_2}{\varphi(F_{1})\hat{\tau}_1},\\
        \hat{\tau}_1&=\frac{\hat{\mu}_{n2}^{(1)}-\varphi(F_{2})^2\rho_0\hat{\tau}_2^2-\varphi(F_{2}^2)\rho_0^2\hat{\tau}_2^2}{\varphi(F_{1})(\hat{\mu}_{n1}^{(1)}-\varphi(F_{2})\rho_0\hat{\tau}_2)}-\frac{\varphi(F_{1}^2)\left(\hat{\mu}_{n1}^{(1)}-\varphi(F_{2})\rho_0\hat{\tau}_2\right)}{\varphi(F_{1})^3}-2\hat{\tau}_2\frac{\rho_0+\varphi(F_{2})}{\varphi(F_{1})}.
\end{align*}

\subsubsection{Construction of MoM estimators under Scenarios 2 and 3}
In contrast to the sequential approach of Scenario 1, we now estimate the full parameter vector $\theta=(\tau_1,\rho,\tau_2)$ jointly using the between-group sum-of-squares matrix $S_{n1}$ defined in \eqref{eq:Sn1}. 
The set of non-crossing partitions $\text{NC}_2(6)$ contains five elements with corresponding Kreweras complements:
\begin{enumerate}
    \item $\pi_1=12|34|56$, $K(\pi_1)=1|3|5|246$,
    \item $\pi_2=12|36|45$, $K(\pi_2)=1|35|26|4$,
    \item $\pi_3=14|23|56$, $K(\pi_3)=13|2|5|46$,
    \item $\pi_4=16|23|45$, $K(\pi_4)=135|2|4|6$,
    \item $\pi_5=16|25|34$, $K(\pi_5)=15|3|24|6$.
\end{enumerate}
Let $\hat{\mu}_n=(p^{-1}\Tr S_{n1},\,p^{-1}\Tr S_{n1}^2,\,p^{-1}\Tr S_{n1}^3)$ denote the first three empirical moments of $S_{n1}$.
In addition to the asymptotic expansions derived in (\ref{eq:hat.mu.12.main}) for $\hat{\mu}_{n1}^{(1)}$ and $\hat{\mu}_{n2}^{(1)}$, we
apply Lemma \ref{lem:map.asym} to the third moment to obtain the asymptotic expansions
\begin{align*}
    \begin{cases}
    \hat{\mu}_{n1}=_{a.s}&\sum_{r_1=1}^2\varphi(F_{r_1})\phi(g_{r_1})\\
    \hat{\mu}_{n2}=_{a.s}&\sum_{r_1=1}^2\sum_{r_2=1}^2\left[\varphi(F_{r_1})\varphi(F_{r_2})\phi(g_{r_1}g_{r_2})+\varphi(F_{r_1}F_{r_2})\phi(g_{r_1})\phi(g_{r_2})\right]\\
    \hat{\mu}_{n3}=_{a.s}&\sum_{r_1=1}^2\sum_{r_2=1}^2\sum_{r_3=1}^2\left[\varphi(F_{r_1})\varphi(F_{r_2})\varphi(F_{r_3})\phi(g_{r_1}g_{r_2}g_{r_3})\right.\nonumber\\
    &\left.+\varphi(F_{r_1}F_{r_2}F_{r_3})\phi(g_{r_1})\phi(g_{r_2})\phi(g_{r_3})+3\varphi(F_{r_1})\varphi(F_{r_2}F_{r_3})\phi(g_{r_1}g_{r_2})\phi(g_{r_3})\right].
    \end{cases}
\end{align*}
For each $\theta=(\tau_1,\rho,\tau_2)$, let ${\Phi}:\mathbb{R}^3\rightarrow\mathbb{R}^9$ be
    \begin{align*}
  {\Phi}(\theta)=\left(\rho\tau_1,\rho_0\tau_2,\rho\tau_1^2,\rho\tau_1\tau_2,\rho_0\tau_2,\rho\tau_1^3,\rho\tau_1^2\tau_2,\rho\tau_1\tau_2^2,\rho_0\tau_2^3\right),
    \end{align*}
which maps $\theta$ to $$\{\phi(g_1),\phi(g_2),\phi(g_1^2),\phi(g_1g_2),\phi(g_2^2),\phi(g_1^3),\phi(g_1^2g_2),\phi(g_1g_2^2),\phi(g_2^3)\}.$$ 
Let $\Psi:\mathbb{R}^9\rightarrow\mathbb{R}^3$ be defined according to the asymptotic expansions above, and define ${\mathcal{F}}_n={\Psi}_n\circ{\Phi}$. The MoM estimator $\hat{\theta}=(\hat{\tau}_1,\hat{\rho},\hat{\tau}_2)$ is then the solution to the system of equations
\begin{align*}
    \hat{\mu}_{n}={\mathcal{F}}_n(\hat{\theta}).
\end{align*}
Under Scenario 3, simultaneous diagonalizability fails, rendering the asymptotic expansions invalid. Nevertheless, we apply the MoM construction above under the false assumption that simultaneous diagonalizability holds.

\subsubsection{Construction of MoM estimators under Scenarios 4 and 5}
When $\Sigma_1$ and $\Sigma_2$ are not simultaneously diagonalizable, we employ the empirical fix procedure of Section~\ref{subsec:empirical.fix}. We first estimate the eigenvector matrices using the MANOVA estimators defined in \eqref{eq:manova.fullsib}. Under the choice $X=0$, these estimators simplify to
\begin{align*}
    \hat{\Sigma}_1=\frac{p}{n}\left(S_{n1}-\frac{I}{n-I}S_{n2}\right),\quad \hat{\Sigma}_2=\frac{p}{n-I}S_{n2}.
    \end{align*}
Let $\hat{\Sigma}_r=\hat{V}_r\hat{\Lambda}_r\hat{V}_r^T$, $r=1,2$, denote the spectral decomposition of the MANOVA estimator.

For each $\theta=(\tau_1,\rho,\tau_2)$, we construct surrogate covariance matrices $\Sigma_1'(\theta)$ and $\Sigma_2'(\theta)$ by replacing the true eigenvalues with parametric forms while retaining the empirical eigenvectors:
\begin{align*}
	\Sigma_1' = \hat{V}_1\Lambda_1'(\theta)\hat{V}_1^T,\quad\Lambda_1'(\theta) = \left[ \begin{array}{cc}
		\tau_1 \Id_{\lceil p \rho \rceil } & 0 \\
		0 &0
	\end{array} \right], \\
    \Sigma_2' = \hat{V}_2\Lambda_2'(\theta)\hat{V}_2^T,\quad\Lambda_2'(\theta)=\left[ \begin{array}{cc}
		\tau_2 \Id_{\lceil p {\rho_0} \rceil } & 0 \\
		0 &0
	\end{array} \right].
    \end{align*}
Define $\hat{\Phi}:\mathbb{R}^3\rightarrow\mathbb{R}^9$ as the mapping from $\theta$ to the collection of moments 
\begin{align*}
    \{\varphi(\Sigma_1'),\varphi(\Sigma_2'),\varphi(\Sigma_1'^2),\varphi(\Sigma_1'\Sigma_2'),\varphi(\Sigma_2'^2),\varphi(\Sigma_1'^3),\varphi(\Sigma_1'^2\Sigma_2'),\varphi(\Sigma_1'\Sigma_2'^2),\varphi(\Sigma_2'^3)\}.
\end{align*}
Let $\hat{\mathcal{F}}_n=\Psi\circ\hat{\Phi}$, where $\Psi$ is the same mapping from Scenarios 2 and 3. The method of moments estimator $\hat{\theta}_n$ is then the solution to
\begin{align}
    \hat{\mu}_n=\hat{\mathcal{F}}_n(\hat{\theta}_n).
\end{align}
For the oracle setting in Scenario 5, we assume that the eigenvector matrices $V_1$ and $V_2$ are known. In this case, we construct the surrogate matrices as 
$$
\Sigma_1'=V_1\Lambda_1'(\theta)V_1^T,\quad \Sigma_2'=V_2\Lambda_2'(\theta)V_2^T,
$$
and follow the same procedure.

\subsubsection{Simulations}\label{subsec:simu.main}
We conduct a simulated experiment involving $I=1000$ families, where for each individual we measure $p=500$ phenotypic traits. The covariance matrices $\Sigma_1$ and $\Sigma_2$ are each assumed to have two distinct eigenvalues: one nonzero eigenvalue and one zero eigenvalue. For $\Sigma_1$, the nonzero eigenvalue is $\tau_1=1$ taking up a proportion of $\rho=0.5$. For $\Sigma_2$, the nonzero eigenvalue is $\tau_2=0.5$ taking up a known proportion of $\rho_0=0.8$. In Scenarios 3, 4, and 5, we set $V_1=R(\omega)$, where $\omega=\pi/4$.
The family sizes $\{J_i\}_{i=1}^I$ are modeled as independent and identically distributed random variables, each taking the value 1 or 2 with equal probability, i.e., $J_i\sim \text{Unif}\{1,2\}$. 

For each scenario, we conduct 1000 independent experiments to obtain MoM estimators $\{\hat{\theta}^{(s)}=(\hat{\tau}_1^{(s)},\hat{\rho}^{(s)},\hat{\tau}_2^{(s)})\}_{s=1}^{1000}$. Let $\bar{\hat{\theta}}=1000^{-1}\sum_{s=1}^{1000}\hat{\theta}^{(s)}$ denote the sample mean. We compute the empirical bias $\bar{\hat{\theta}}-\theta$ and the sample standard deviation across all replicates. Figure~\ref{fig:hist.mom.est.main} displays histograms of the MoM estimators for each scenario. Under Scenarios 1, 2, 4, and 5, the estimators are centered around the true parameter values, whereas under Scenario 3, the estimators exhibit systematic bias.

For Scenarios 1--3, we numerically evaluate the theoretical bias and standard deviations given in Theorem~\ref{thm:main} and Lemma~\ref{lem:nested} using the procedure described in Section~\ref{subsec:clt}. Scenario 4, which employs estimated eigenvectors, does not admit asymptotic guarantees. For the oracle setting in Scenario 5, where the eigenvector matrices $V_1$ and $V_2$ are assumed known, consistency and asymptotic normality follow by analogous arguments, with minor modifications to account for $\Sigma_1$ and $\Sigma_2$ not being simultaneously diagonalizable. Additional details on evaluating the theoretical biases and standard deviations under each scenario are provided in Appendix~\ref{app:experiments}. Table~\ref{tab:sd_bias_main} shows that the empirical bias and standard deviations align closely with their theoretical counterparts for Scenarios 1, 2, 3, and 5, confirming the validity of our asymptotic theory, even under model misspecification. In Appendix \ref{subsec:corr}, we also report the empirical and theoretical correlation matrix between the MoM estimators.

To assess the variability of the empirical bias and standard deviation, we partition the 1000 replicates into 20 groups of 50 replicates each. Within each group $i$, we compute the empirical bias and standard deviation, denoted $\widehat{\text{bias}}^{(i)}(\hat{\theta})$ and $\widehat{\text{sd}}^{(i)}(\hat{\theta})$, respectively. We then calculate the mean and standard deviation of these 20 estimates and construct 99\% confidence intervals of the form $[\bar{x} - 2.576s,\, \bar{x} + 2.576s]$, where $\bar{x}$ denotes the sample mean and $s$ denotes the sample standard deviation across the 20 groups. These intervals quantify the sampling variability of the empirical bias and standard deviation estimates.  Figure~\ref{fig:normal.approx.main} displays these confidence intervals, with the theoretical biases and standard deviations indicated in red. From Figure~\ref{fig:normal.approx.main}, we observe that for Scenarios 1, 2, 3, and 5, the theoretical bias and standard deviation of each MoM estimator lie within the 99\% confidence intervals constructed from the empirical estimates.

 A notable feature of the sequential estimation procedure in Scenario 1 is that the MoM estimator for $\theta_2=(\tau_2)$, which governs the eigenvalue decay of $\Sigma_2$, exhibits markedly lower bias and variance than those for $\theta_1=(\tau_1,\rho)$, which governs $\Sigma_1$. This phenomenon can be attributed to two factors: (i) estimating $\theta_1$ requires deconvolving the combined effects of both $\Sigma_1$ and $\Sigma_2$ from the between-group sum-of-squares matrix $S_{n1}$, whereas $\theta_2$ is estimated directly from the within-group matrix $S_{n2}$, which depends solely on $\Sigma_2$; and (ii) the second-level random effect $\alpha_2\in\mathbb{R}^{n\times p}$ involves substantially more observations than the first-level random effect $\alpha_1\in\mathbb{R}^{I\times p}$, yielding more accurate estimation of $\theta_2$.

 Comparing Scenarios~1 and~2, Figure~\ref{fig:hist.mom.est.main} and Table~\ref{tab:sd_bias_main} show that sequential estimation substantially improves accuracy for $\rho$ and $\tau_2$. In particular, under joint estimation (Scenario~2), both the bias and the standard deviation of $\hat{\tau}_2$ are nearly an order of magnitude larger than under sequential estimation (Scenario~1), and the corresponding quantities for $\hat{\rho}$ are more than three times larger. This pattern is consistent with the differences between the procedures: (i) the CLT for sequential estimation combines moment information from both $S_{n1}$ and $S_{n2}$ which are independent, while the joint procedure relies on moments of $S_{n1}$ alone; and (ii) in the sequential procedure $\tau_2$ is estimated independently from $S_{n2}$, whereas the joint procedure  estimates all parameters simultaneously from $S_{n1}$ and therefore faces a harder deconvolution problem. The estimator $\hat{\tau}_1$ exhibits comparable performance across the two scenarios. The estimator $\hat{\tau}_1$ exhibits comparable performance under both scenarios. Figure~\ref{fig:normal.approx.main} further illustrates that the empirical bias and standard deviation of $\hat{\rho}$ and $\hat{\tau}_2$ exhibit substantially greater variability under Scenario 2 than under Scenario 1. Finally, we compare the empirical fix procedure under Scenario 4 with the oracle fix procedure under Scenario 5. The empirical biases and standard deviations under Scenario 4 exhibit greater variability, with $\hat{\tau}_2$ showing the most pronounced increase, reflecting the additional uncertainty introduced by estimating the eigenvector matrices.

\begin{figure}[H]
    \centering
    \begin{subfigure}[t]{0.48\textwidth}
        \centering
        \includegraphics[trim={0cm 0cm 0cm 0cm},clip,width=\linewidth]{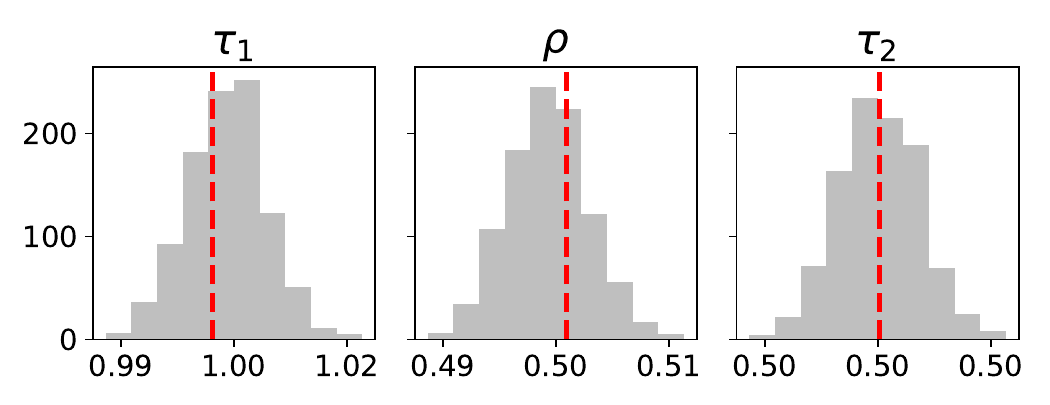}
        \caption{Scenario 1: Correct, sequential estimation.}
    \end{subfigure}\\
    \begin{subfigure}[t]{0.48\textwidth}
        \centering
        \includegraphics[trim={0cm 0cm 0cm 0cm},clip,width=\linewidth]{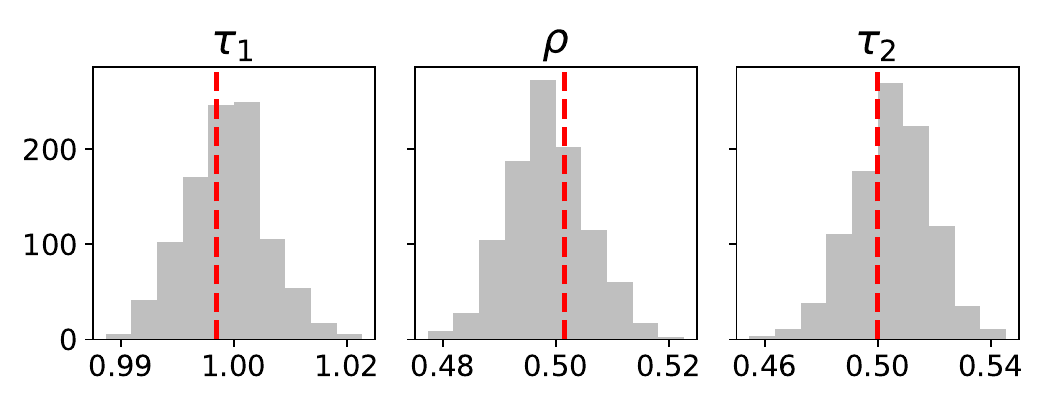}
        \caption{Scenario 2: Correct, joint estimation.}
    \end{subfigure}
    \begin{subfigure}[t]{0.48\textwidth}
        \centering
        \includegraphics[trim={0cm 0cm 0cm 0cm},clip,width=\linewidth]{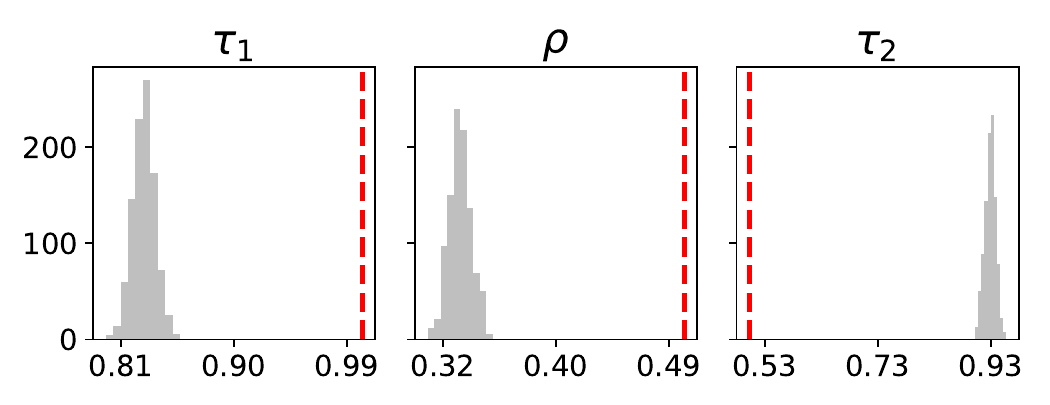}
        \caption{Scenario 3: Misspecified, joint estimation.}
    \end{subfigure}
    \begin{subfigure}[t]{0.48\textwidth}
        \centering
        \includegraphics[trim={0cm 0cm 0cm 0cm},clip,width=\linewidth]{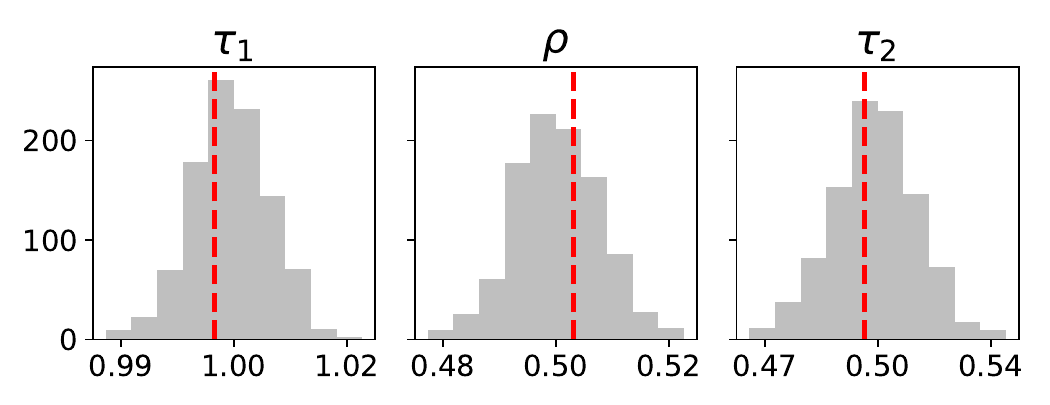}
        \caption{Scenario 4: Misspecified, empirical fix.}
    \end{subfigure}
    \begin{subfigure}[t]{0.48\textwidth}
        \centering
        \includegraphics[trim={0cm 0cm 0cm 0cm},clip,width=\linewidth]{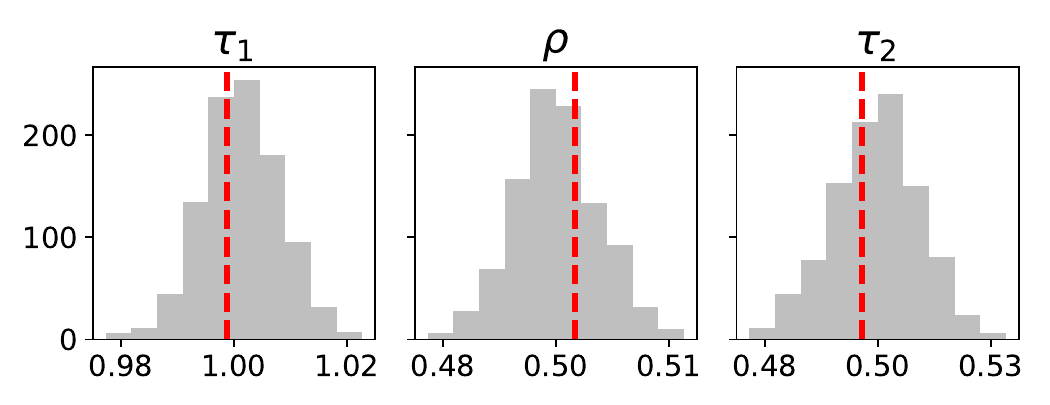}
        \caption{Scenario 5: Misspecified, oracle fix.}
    \end{subfigure}
    \caption{Histogram of MoM estimators, with true parameter value marked in red.}
    \label{fig:hist.mom.est.main}
\end{figure}

\begin{table}[H]
    \centering
     \caption{Empirical and theoretical biases and standard deviations of the MoM estimators.}
     \begin{tabular}{c c c c c @{\,} c c c}
        & \multicolumn{3}{c}{Scenario 1} && \multicolumn{3}{c}{Scenario 2} \\
        \cline{2-4} \cline{6-8}
              & $\tau_1$& $\rho$  &$\tau_2$ && $\tau_1$ & $\rho$ &$\tau_2$\\
               \hline
             {bias (emp.)}  & 0.0020& -0.0010& 0.0000& & 0.0012& -0.0031& 0.0047\\
             {bias (theo.)}  &0.0020& -0.0010& 0.0000&&0.0009& -0.0030& 0.0048 \\
             {sd (emp.)}&0.0056& 0.0022& 0.0016&& 0.0049& 0.0073& 0.0136\\
             {sd (theo.)}&0.0055& 0.0022& 0.0015&& 0.0049& 0.0072& 0.0134
          \end{tabular}

    \begin{tabular}{c c c c c @{\,} c c c c @{\,} c c c}
  & \multicolumn{3}{c}{Scenario 3} && \multicolumn{3}{c}{Scenario 4}&& \multicolumn{3}{c}{Scenario 5} \\
  \cline{2-4} \cline{6-8} \cline{10-12}
        & $\tau_1$& $\rho$  &$\tau_2$ && $\tau_1$ & $\rho$ &$\tau_2$&& $\tau_1$ & $\rho$ &$\tau_2$\\
         \hline
        bias (emp.)  &-0.176& -0.1693& 0.4283& &0.0021& -0.0028& 0.0033 &&0.0023& -0.0027& 0.0029 \\
        bias (theo.)& -0.1761& -0.1693& 0.4284&&n.a.&n.a.&n.a.&&0.0023& -0.0027& 0.0030\\
       sd (emp.) &0.0091& 0.0084& 0.0096&&0.0051& 0.0078& 0.0136&&0.0064& 0.0066& 0.009\\
       sd (theo.) &0.0086& 0.0082& 0.0093&&n.a.&n.a.&n.a.&&0.0061& 0.0064& 0.0087
    \end{tabular}
  \label{tab:sd_bias_main}
\end{table}

\begin{figure}[H]
    \centering
    \begin{subfigure}[t]{0.49\textwidth}
        \centering
        \includegraphics[trim={0cm 0cm 0cm 0cm},clip,width=\linewidth]{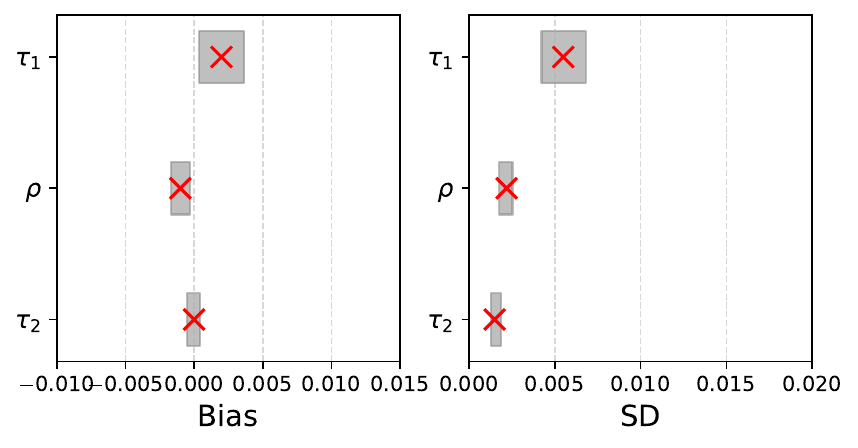}
        \caption{Scenario 1: Correct, sequential estimation.}
    \end{subfigure}\\
    \begin{subfigure}[t]{0.48\textwidth}
        \centering
        \includegraphics[trim={0cm 0cm 0cm 0cm},clip,width=\linewidth]{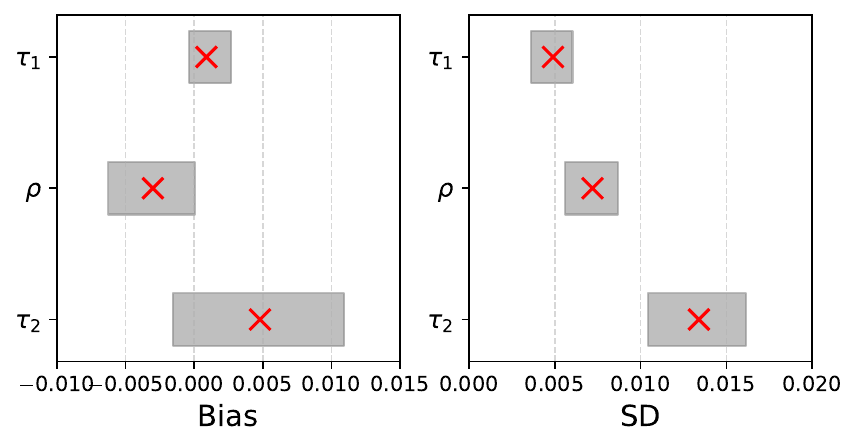}
        \caption{Scenario 2: Correct, joint estimation.}
    \end{subfigure}
    \begin{subfigure}[t]{0.48\textwidth}
        \centering
        \includegraphics[trim={0cm 0cm 0cm 0cm},clip,width=\linewidth]{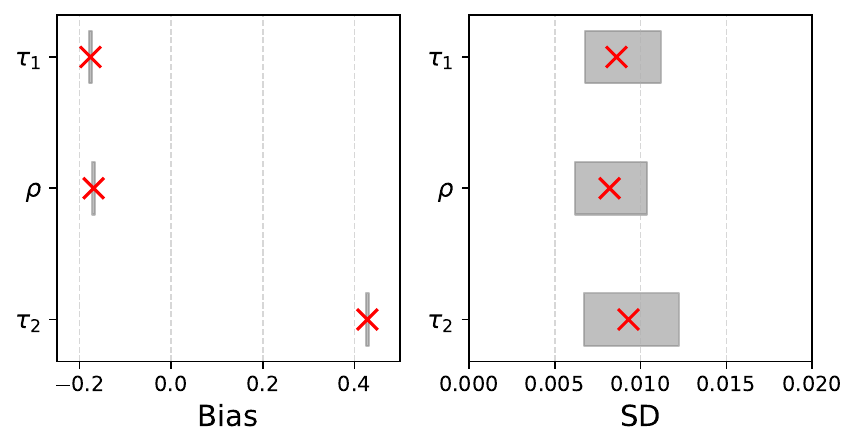}
        \caption{Scenario 3: Misspecified, joint estimation.}
    \end{subfigure}
    \begin{subfigure}[t]{0.48\textwidth}
        \centering
        \includegraphics[trim={0cm 0cm 0cm 0cm},clip,width=\linewidth]{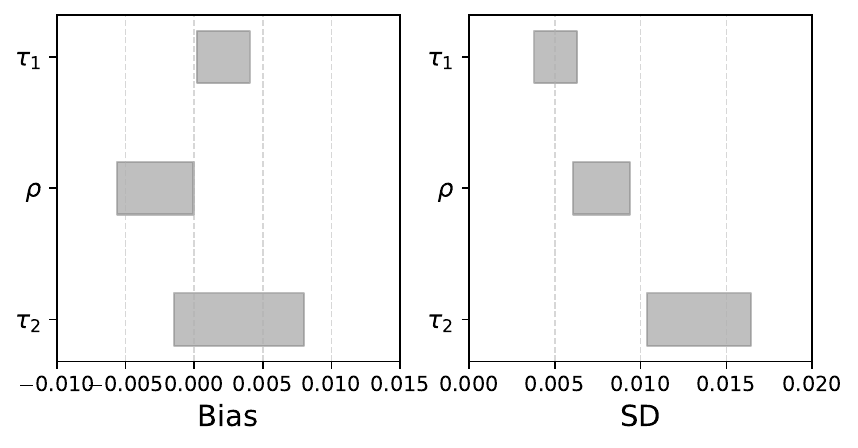}
        \caption{Scenario 4: Misspecified, empirical fix.}
    \end{subfigure}%
    \begin{subfigure}[t]{0.48\textwidth}
        \centering
        \includegraphics[trim={0cm 0cm 0cm 0cm},clip,width=\linewidth]{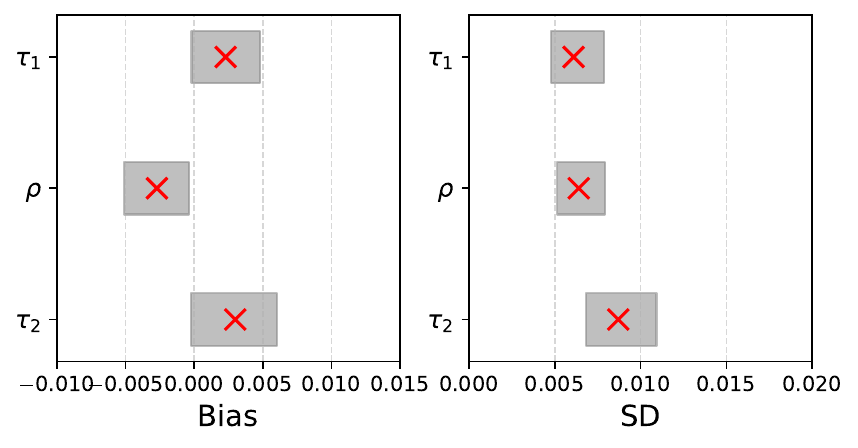}
        \caption{Scenario 5: Misspecified, oracle fix.}
    \end{subfigure}%
    \caption{Normal approximation 99\% confidence intervals for the bias and standard deviation of method of moments estimators, computed using 20 empirical estimates, each based on 50 MoM estimators. The theoretical bias and standard deviations of the MoM estimators are marked in red.}
    \label{fig:normal.approx.main}
\end{figure}
\subsection{Real data analysis}\label{sec:real.data}

We analyze gene expression data from $I=136$ accessions of the indica subspecies of \textit{Oryza sativa} assayed under drought conditions \cite{rice20}. Each accession has $J=3$ replicate measurements, giving $n=408$ observations in total. We restrict attention to the $p=100$ most heritable gene expression traits. For the $j$-th replicate of the $i$-th accession, we model the observed trait vector $y_{ij} \in \mathbb{R}^p$ as
\begin{align*}
    y_{ij} = \mu + \alpha_i + \epsilon_{ij}\in\mathbb{R}^p, \quad \alpha_i\sim N(0,{\Sigma_1}), \epsilon_{ij} \sim N(0,{\Sigma_2}),
\end{align*}
where $\mu \in \mathbb{R}^p$ is the population mean,$\alpha_i$ represents the accession-level random effect, and $\epsilon_{ij}$ is the residual error. This model can then be expressed in the form \eqref{eq:fullsib.exp1}
with $X=\mathbbm{1}_n$, $\beta=\mu^T$, $U_1=\mathbbm{1}_{J}\otimes \mathbbm{1}_I \in\mathbb{R}^{n\times I}$ and $U_2=\Id_n$.

We focus on estimating $\Sigma_1$, which is proportional to the additive genetic covariance matrix $G$ \cite{falconer1996introduction}. Under the balanced full-sib design, the REML estimator admits the closed-form Amemiya representation \cite{Amemiya01051985}:
\begin{align*}
    \hat\Sigma_{1,R}= \frac{p}{n(J-1)}S_{n2}^{1/2}H\left(\frac{n-I}{I-1}S_{n2}^{-1/2}S_{n1}S_{n2}^{-1/2}\right)S_{n2}^{1/2},
\end{align*}
where $H(S)$ maps a nonnegative definite matrix $S$ with spectral decomposition $\sum_1^p\lambda_ju_ju_j^T$ into $\sum_1^p(\lambda_j-1)_+u_ju_j^T$.  The MANOVA estimator $\hat{\Sigma}_{1,M}$, computed from \eqref{eq:manova.fullsib}, simplifies to
\begin{align*}
    \hat\Sigma_{1,M}= \frac{p}{n-J}S_{n1}-\frac{p}{n(J-1)}S_{n2}.
\end{align*}
Since $n(J-1)(I-1)=(n-J)(n-I)$, REML coincides with MANOVA whenever all eigenvalues of $\frac{n-I}{I-1}S_{n1}S_{n2}^{-1}$ are at least one. Figure \ref{fig:rice_M_R_hist} shows histograms of the REML and MANOVA eigenvalue estimates for $\Sigma_1$. We observe that both exhibit a pronounced leading spike.
\begin{figure}
    \centering
    \includegraphics[trim={0cm 0.22cm 0cm 0.cm},clip,width=0.8\linewidth]{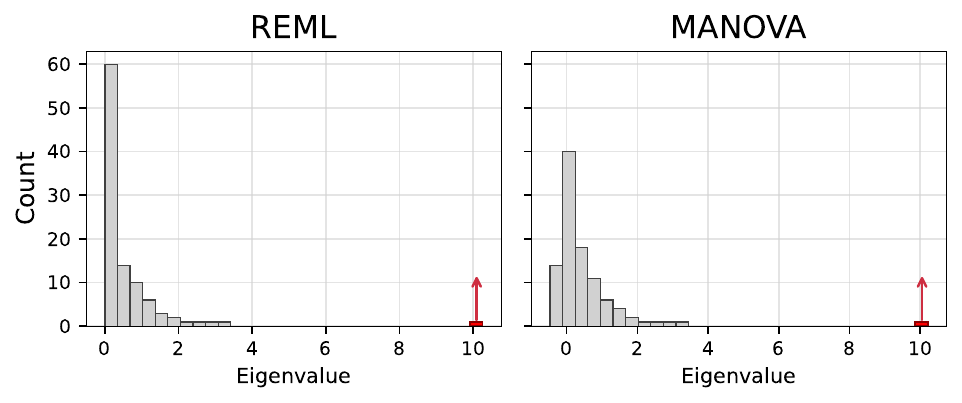}
    \caption{Histogram of the REML and MANOVA estimates for $\Sigma_1$ on the rice data.}
    \label{fig:rice_M_R_hist}
\end{figure}

We next compute MoM estimators for the eigenvalues of $\Sigma_1$ under the quadratic model
\begin{align*}
    g_r(x)=\tau_r(1-\rho_rx)_+^2,\quad r=1,2,
\end{align*}
and the exponential model
\begin{align*}
    g_r(x)=\tau_{r}\exp(-\rho_{r}x),\quad r=1,2.
\end{align*}
To accommodate spikes, we use the $\text{\rm MoM-EEVS}$ procedure in Section \ref{subsec:empirical.fix}, combined with the sequential procedure in Section~\ref{sec:sequential.nested} for computational simplicity.

Figure~\ref{fig:rice} compares the estimated eigenvalues of $\Sigma_1$ across the four methods. REML and MANOVA are broadly similar, except that MANOVA yields some negative eigenvalues whereas REML is constrained to be positive semidefinite. Relative to REML and MANOVA, the MoM estimators shrink the spectrum toward its center, increasing smaller eigenvalues and decreasing larger ones, thereby reducing overall eigenvalue dispersion. REML and the quadratic MoM fit both imply a nontrivial estimated null space: the estimated null-space fraction is $0.30$ under REML and $0.14$ under quadratic MoM. Taken together, these results suggest that REML may overstate both upper-tail eigenvalues and null-space mass, which would inflate inferred high evolvabilities and evolutionary constraints.


\begin{figure}
    \centering
    \begin{subfigure}[t]{\textwidth}
        \centering
        \includegraphics[trim={0cm 0.2cm 0cm 0.2cm},clip,width=\linewidth]{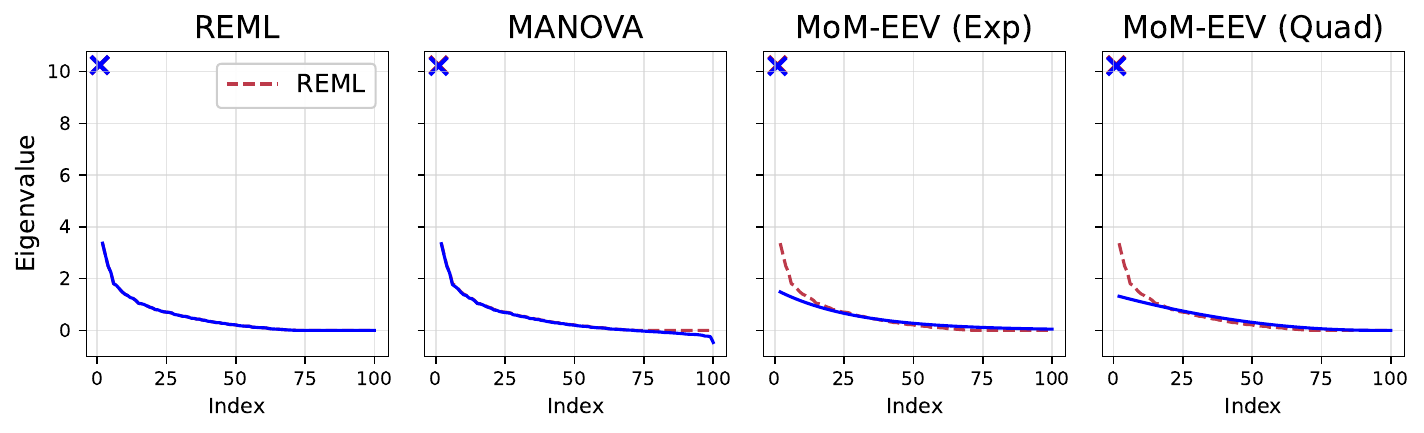}
        \caption{Comparison of the estimated eigenvalues of $\Sigma_1$.}
    \end{subfigure}
    \begin{subfigure}[t]{\textwidth}
        \centering
        \includegraphics[trim={0cm 0.2cm 0cm 0.2cm},clip,width=\linewidth]{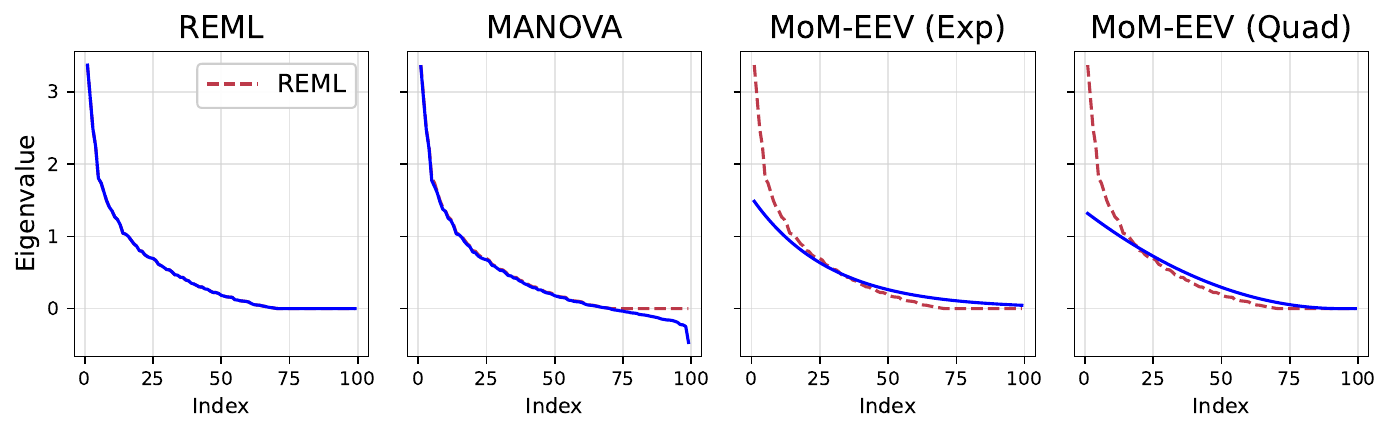}
        \caption{Same as (a), excluding the largest estimated eigenvalue.}
    \end{subfigure}
    \caption{Comparison of the four estimators for $\Sigma_1$ on the rice data.}
    \label{fig:rice}
\end{figure}

\section{Background}\label{sec:free.prob}
\subsection{Free probability}
In this section, we review some useful definitions from free probability theory, referencing \cite{bose2021random,nica2006lectures}, and establish some preliminary results that will be used in the proof of our main result. 
\begin{definition}[Paraphrasing Definition 2.5.2 and 2.5.4 in \cite{bose2021random}]
Let $\mathcal{A}$ over $\mathbb{C}$ be a unital $*$-algebra with unity $1_{\mathcal{A}}$ {\cite[Definition~2.5.3]{bose2021random}}. Let $\varphi:\mathcal{A}\rightarrow\mathbb{C}$ be a linear functional which satisfies 
\begin{align*}
   \varphi(1_{\mathcal{A}})=1,\quad \varphi(a^*a)\geq 0 \text{ for all }a\in\mathcal{A}.
\end{align*}
Then $(\mathcal{A},\varphi)$ is called a \textit{non-commutative $*$-probability space} (NCP) and $\varphi$ is called a \textit{state}. A state $\varphi$ is said to be \textit{tracial} if 
\begin{align*}
    \varphi(ab)=\varphi(ba),\text{ for all }a,b\in\mathcal{A}.
\end{align*}
\end{definition}
\begin{example}
For any positive $m$, consider the $*$-algebra $\mathcal{M}_m(\mathbb{C})=\mathbb{C}^{m\times m}$, with the $*$-operation given by the conjugate transpose. For $A\in\mathcal{M}_m(\mathbb{C})$, let $\varphi(A)=m^{-1}\Tr A$. Then $\varphi$ is tracial and $(\mathcal{M}_m(\mathbb{C}),m^{-1}\Tr)$ is an NCP.  
\end{example}
\begin{example}
    In the above example, allowing matrices in $\mathcal{M}_m(\mathbb{C})$ to be random gives the NCP
    $(\mathcal{M}_m(\mathbb{C}),m^{-1}\mathbb{E}\operatorname{Tr})$.
\end{example}
\begin{definition}[Non-crossing partitions, Paraphrasing Definition 2.1.1 in \cite{bose2021random}]
    Let $[n]=\{1,\ldots,n\}$, and $\pi=\{V_1,\ldots,V_r\}$ be a partition of $[n]$ into $r$ blocks. If there exists $p_1<q_1<p_2<q_2$ in $[n]$ such that $p_1$ and $p_2$ are in the same block, $q_1$ and $q_2$ are in the same block, but $\{p_1,q_1,p_2,q_2\}$ are not in the block, then $\pi$ is called crossing. If $\pi$ is not crossing, then it is called non-crossing. Denote all non-crossing partitions of $[n]$ by $\text{NC}(n)$. For any $\pi,\sigma\in\text{NC}(n)$, we say that $\pi\leq \sigma$ if every block of $\pi$ is contained in some block of $\sigma$.
\end{definition}
As discussed in Section 3.8 in \cite{bose2021random}, we can
  alternatively view each non-crossing partition $\pi\in\text{NC}(n)$
  as a permutation of $[n]$ as follows: within each block,
  $i_1 < i_2 < \cdots < i_t$ say,
  consider the cyclic permutation which shifts
  every element to the next right element of the block, and the last
  element of the block is shifted to the first element of the
  block.
  We write, for example, $\pi(i_1) = i_2$, and $\pi(i_t) =
    i_1$.
  This gives the permutation on $[n]$ whose cycles are the permutations within each block of $\pi$. We next define the Kreweras complement of a non-crossing partition and show in Lemma~\ref{lem:kreweras} how to compute it from this permutation representation. For notational convenience, we use $\pi$ to denote both the partition and its associated permutation.

\begin{definition}[Kreweras complement, Definition 3.8.1 in \cite{bose2021random}] 
   {Let $\pi\in\text{NC}(n)$ be a non-crossing partition of $[n]$.  Introduce additional elements $[\bar{n}]=\{\bar{1},\ldots,\bar{n}\}$ with the ordering $\bar{1}<\ldots<\bar{n}$, and define the interlaced ordering $1<\bar{1}<2<\bar{2}<\ldots<n<\bar{n}$. The Kreweras complement of $\pi$, denoted $K(\pi)$, is defined to be the biggest element among all non-crossing partitions $\sigma$ of $[\bar{n}]$ such that $\pi\cup\sigma$ is a non-crossing partition of the $2n$ elements $\{1,\Bar{1},\ldots,n,\Bar{n}\}$.}
\end{definition}

\begin{lem}[Kreweras complement of non-crossing pair partitions, Lemma 2.3.1 and Lemma 3.8.1 in \cite{bose2021random}]\label{lem:kreweras}
    Denote the set of non-crossing pair partitions of $[2l]$ by $\text{NC}_2(2l)$, then $|\text{NC}_2(2l)|=\frac{(2l)!}{(l+1)!l!}$. For any $\pi\in\text{NC}_2(2l)$, we have $K(\pi)=\pi\gamma_{2l}$ 
    , where $\gamma_{2l}$ is the cyclic permutation $1\rightarrow 2\rightarrow\ldots\rightarrow 2l\rightarrow 1$, with $|K(\pi)|=l+1$. {Moreover, $K$ is a bijection on $\text{NC}_2(2l)$, with $\pi=K(\pi)\gamma_{2l}^{2l-1}$.} 
\end{lem}
\begin{proof}
    {Note that $\gamma_{2l}^{2l}$ is the identity mapping. As a result, given that $K(\pi)=\pi\gamma_{2l}$, it directly follows that $\pi=K(\pi)\gamma_{2l}^{2l-1}$. The remaining claims are established in Lemma 2.3.1 and Lemma 3.8.1 in \cite{bose2021random}.}
\end{proof}

\begin{example}\label{eg:kre}
For $l=1$, $\text{NC}_2(2)=\{\pi\}$, where
\begin{align*}
    \pi=12,~ K(\pi)=1|2.
\end{align*}
    For $l=2$, $\text{NC}_2(2l)=\{\pi_1,\pi_2\}$, where 
    \begin{enumerate}
        \item $\pi_1=12|34$, $K(\pi_1)=1|24|3$,
        \item $\pi_2=14|23$, $K(\pi_2)=13|2|4$.
    \end{enumerate}
\end{example}

\begin{definition}[Moments, Section 2.6 in \cite{bose2021random}] \label{def:moments} Let $(\mathcal{A},\varphi)$ be an NCP. For $a\in \mathcal{A}$, the numbers $\{\varphi(\prod_{i=1}^na^{\epsilon_i})|\epsilon_i\in\{1,*\},n\geq 1\}$ are called moments of $a$. Define a sequence of multi-linear functionals (that is, linear in each coordinate) $(\varphi_n)_{n\geq 1}$ on $\mathcal{A}^n$ via
\begin{align*}
    \varphi_n(a_1,\ldots,a_n):=\varphi(a_1\ldots a_n)
\end{align*}
for each $a_1,\ldots,a_n\in\mathcal{A}$. If $\pi=\{V_1,\ldots,V_r\}\in \text{NC}(n)$, then define
\begin{align*}
    \varphi_{\pi}:=\varphi(V_1)[a_1,\ldots,a_n]\ldots\varphi(V_r)[a_1,\ldots,a_n],
\end{align*}
where $\varphi(V)[a_1,\ldots,a_n]:=\iain{\varphi(a_{i_1}\ldots a_{i_s})}[subscripts]$ for $V=\{i_1,\ldots,i_s\}$ with $i_1<\ldots<i_s$.
\end{definition}
\begin{definition}[Free cumulants, Definition 2.6.2 in \cite{bose2021random}] Let $(\mathcal{A},\varphi)$ be an NCP. The joint free cumulant of order $n$ is defined on $\mathcal{A}^n$ by
\begin{align*}
    \kappa_n(a_1,\ldots,a_n)=\sum_{\sigma\in \text{NC}(n)}\varphi_{\sigma}[a_1,\ldots,a_n]\mu[\sigma,\mathbf{1}_n],\quad (a_1,\ldots,a_n)\in\mathcal{A}^n,
\end{align*}
where $\mu$ is the Möbius function of $\text{NC}(n)$ {\cite[Definition 1.4.2]{bose2021random}}. 
\end{definition}

\begin{definition}[Semicircular variables, Example 2.6.1 and Definition 3.6.2 and in \cite{bose2021random}]
{A variable $s$ on a $*$-probability space $(\mathcal{A},\varphi)$ is said to be semicircular if it is self-adjoint, its second free cumulant $\kappa_2(s)=1$, and all other free cumulants are 0.}
\end{definition}
\begin{definition}[Circular variables, Definition 11.22 in \cite{nica_speicher_2006}]\label{def:circular} An element $c$ of the form $c=\frac{1}{\sqrt{2}}(s_1+is_2)$ -- where $s_1$ and $s_2$ are two freely independent semicircular elements of
variance 1 -- is called a circular element. A circular element satisfies $\kappa_2(c,c^*)=\kappa_2(c^*,c)=1$ and all other free cumulants of $\{c,c^*\}$ are 0.
\end{definition}

\begin{definition}[Free independence, Definition 3.1.1 in \cite{bose2021random}]  Let $(\mathcal{A},\varphi)$ be an NCP. Then the $*$-sub-algebras $(\mathcal{A}_i)_{i\in I}$ of $\mathcal{A}$ are said to be free if, for all $n\geq 2$, and all $a_1,\ldots,a_n$ from $(\mathcal{A}_i)_{i\in I}$, $\kappa_n(a_1,\ldots,a_n)=0$ whenever at least two of the $a_i$ are from different $\mathcal{A}_i$. In particular, any collection of variables is said to be free if the sub-algebras
generated by these variables are free.
    
\end{definition}

\begin{theorem}[Theorem 3.2.1 in \cite{bose2021random}]\label{thm:free.prod}
    Let $(\mathcal{A}_i,\varphi^{(i)})_{i\in I}$ be a family of
    $*$-probability spaces. Then there exists a $*$-probability space $(\mathcal{A},\varphi)$, called the
    $*$-free product of $(\mathcal{A}_i,\varphi^{(i)})_{i\in I}$, such that there is a copy of $\mathcal{A}_i,i\in I$ which
    are freely independent in $(\mathcal{A},\varphi)$ and $\varphi$ restricted to $\mathcal{A}_i$ equals $\varphi^{(i)}$ for all $i$.
\end{theorem}

Building upon the above concepts and results, we introduce two lemmas which serve as the building blocks of our main result.

\begin{lem}[Lemma 3.9.1 in \cite{bose2021random}]\label{lem:kreweras.moments}
Let $(\mathcal{A},\varphi)$ be an NCP. Suppose $\{a_i,1\leq i\leq n\}$ and $\{b_i,1\leq i\leq n\}$ are free. Then
\begin{align*}
    \varphi(a_1b_1\ldots a_nb_n) = \sum_{\pi\in\text{NC}(n)} \kappa_{\pi} [a_1,\ldots,a_n]\varphi_{K(\pi)}[b_1,\ldots,b_n].
\end{align*}
\end{lem}

\begin{lem}\label{lem:free.equiv}
    Let $\{C_m^{(\nu)}\}_{1\leq \nu\leq k}$ be $m\times m$ independent
    matrices with i.i.d.  Gaussian entries with mean 0 and variance
    $m^{-1}$. Let $\{D_m^{(\nu)}\}_{1\leq \nu \leq \ell }$ be $m\times m$ deterministic matrices with $\lVert D_m^{(\nu)}\rVert\leq C$ for some constant $C$. Then there exists a sequence of NCP $(\mathcal{A}_m,\varphi_m)$, such that 
    there exist free circular elements $c_1\ldots,c_k\in\mathcal{A}_m$ and elements $d_m^{(1)},\ldots,d_{m}^{(\ell)}\in\mathcal{A}_m$ with the same {moments} 
    as  $\{D_m^{(\nu)}\}_{1\leq \nu \leq \ell }$
    w.r.t. $\frac{1}{m}\Tr $, {where $\{c_1,\ldots,c_k\}$ are also
      free of $\{d_m^{(1)},\ldots,d_m^{(\ell)}\}$}. Moreover, for all
    polynomials $\Pi$, as $m \to \infty$ it holds that
    \begin{align}\label{eq:conv.moments.deteqv}
        \left|\frac{1}{m}\mathbb{E}\Tr
      \Pi(\{{C}_m^{(\nu)},D_m^{(\nu')}\}_{1\leq \nu \leq k, 1\leq \nu' \leq
      \ell}) - \varphi_m(\Pi(\{c_\nu,d_m^{(\nu')}\}_{1\leq \nu \leq k,
      1\leq \nu' \leq \ell}))\right|\rightarrow 0.
    \end{align}
\end{lem}
\begin{proof}
    We use Theorem \ref{thm:free.prod} to prove the existence of the sequence $(\mathcal{A}_m,\varphi_m)$, and adapt the proof of Theorem 11.2.1 in \cite{bose2021random} to establish \eqref{eq:conv.moments.deteqv}. Details are given in Supplementary Appendix \ref{app:proof.free.equiv}.
\end{proof}

\subsection{Deterministic equivalent spectral distribution}\label{subsec:deteqv}
We now describe the deterministic equivalent law
$\deteqv{F}_n$ of the empirical spectral distribution $F^{S_n}$. The
resulting deterministic equivalent moments  
\begin{align}\label{eq:deteqv.mu}
    \deteqv{\mu}_n =\left(\int xd\deteqv{F}_n(x),\ldots,\int x^Dd\deteqv{F}_n(x)\right).
\end{align}
serve as the centering sequence for the central limit theorem established in Section~\ref{subsec:clt}. 

Since $B$ is PSD, write it as $B = \Gamma \Gamma^T$. For each $r=1,\ldots,k$, define 
\begin{align}\label{eq:Gamma_r.main}
    \Gamma_r = \Gamma^T U_r U_r^T \Gamma.
\end{align}
 Under the simultaneous diagonalizability assumption in
 Assumption~\ref{ass:param}, we have $\Sigma_r=V\Lambda_rV^T$ for some
 orthonormal matrix $V\in\mathbb{R}^{p\times p}$. Define
 $m_n(z)=p^{-1}\Tr(S_n-z\Id_p)^{-1}=\int (x-z)^{-1}dF^{S_n}(x)$ as the
 Stieltjes transform of $F^{S_n}$.
Appendix \ref{app:deteqv.mmts}  establishes the
 following result (by reducing to the special case of Theorem 4.1 in
 \cite{fan2017eigenvalue} studied in \cite{xie2024cltlinearspectralstatistics}.)
\begin{lem}\label{lem:deteqv.mmts}
    Under Assumptions \ref{ass:minimum} and \ref{ass:param}, for each
    $z\in\mathbb{C}^+$, there exists unique z-dependent values
    $\deteqv{g}_i^{(r)}(z)$ with
    $z \deteqv{g}_1^{(r)}(z) \in\overline{\mathbb{C}^+}$
    and
     $ \deteqv{g}_2^{(r)}(z) \in\mathbb{C}^+\cup\{0\}$,
    such that:
\begin{align*}
&\begin{cases}
    z\deteqv{g}_1^{(r)}(z)&=-\frac{1}{p}\Tr \left(({\sum_{s=1}^k}\deteqv{g}_2^{(s)}(z)\Lambda_s+\Id )^{-1}\Lambda_r\right),r=1,\ldots,k\\
    z\deteqv{g}_2^{(r)}(z)&=-\frac{1}{p}\Tr \left(({\sum_{s=1}^k}\deteqv{g}_1^{(s)}(z)\Gamma_s+\Id )^{-1}\Gamma_r\right),r=1,\ldots,k.
    \end{cases}
\end{align*}
The function $\tilde{m}_n:\mathbb{C}^+\rightarrow\mathbb{C}^+$ given by
\begin{equation*}
  z\deteqv{m}_n(z)=-(1-n/p)-\frac{1}{p}\Tr \left({\textstyle\sum}_{r=1}^k\deteqv{g}_1^{(r)}(z)\Gamma_r+\Id  \right)^{-1}
\end{equation*}
defines the Stieltjes transform of a probability measure $\tilde{F}_n$ on $\mathbb{R}$ such that such that 
\begin{align}\label{eq:deteqv.Fn}
    F^{S_n}(x)-\deteqv{F}_n(x)\rightarrow 0
\end{align}
almost surely at each $x\in\mathbb{R}$. Moreover,
    \begin{align*}
    &m_n(z)-\deteqv{m}_n(z)\rightarrow0,\quad \hat{\mu}_n-\deteqv{\mu}_n\rightarrow0,
    \end{align*}
    pointwise almost surely.
\end{lem}

\subsection{Central Limit Theorem of the empirical moments $\hat{\mu}_n$}\label{subsec:clt}
To derive the central limit theorem for the empirical moments $\hat{\mu}_n$, we invoke Theorem 2.3 of \cite{xie2024cltlinearspectralstatistics}, which establishes a CLT for linear spectral statistics of high-dimensional covariance-type matrices. Let $\lambda_{j,r}$ denote the $j$-th eigenvalue in the diagonal matrix $\Lambda_r$, and define the auxiliary matrix
\begin{align}\label{eq:Sn_prime}
    S_n'=\frac{1}{p}\sum_{j=1}^pT_j^{1/2}x_jx_j^TT_j^{1/2},\quad
  T_j=\sum_{r=1}^k\lambda_{j,r}\Gamma_r, \quad x_j\sim N(0,\Id_n),
\end{align}
where $\Gamma_r$ is defined in (\ref{eq:Gamma_r.main}) and the vectors $\{x_j\}$ are independent. Define the constants
\begin{align*}
    s_{\Gamma}:= \lVert B\rVert^{1/2} \lim\sup_{n,r}\lVert U_r\rVert, \quad s_{\Lambda}:=\lim\sup_{n,r}\lVert \Sigma_r\rVert^{1/2} ,
\end{align*} 
which satisfy $\lim\sup_{n,r}\lVert \Gamma_r\rVert^{1/2}\leq s_{\Gamma}$,  $\lim\sup_{n,r}\lVert \Lambda_r\rVert^{1/2}\leq s_{\Lambda}$.  This matrix has the same structural form as the class of matrices studied in \cite{xie2024cltlinearspectralstatistics}.
Applying Theorem 2.3 of \cite{xie2024cltlinearspectralstatistics} to $S_n'$ yields the following result:
 \begin{theorem}\label{thm:clt}
       Under Assumptions \ref{ass:minimum} and \ref{ass:param}(1), as $n \to \infty$, let $\hat\mu_n$ be the first $D$ empirical moments of $S_n$, and $\deteqv{\mu}_n$ be the deterministic equivalent moments given in \eqref{eq:deteqv.mu}, then 
\begin{align*}
    P_n^{-1/2}\left(p(\hat{\mu}_n-\deteqv{\mu}_n) - R_n\right)
\end{align*}
converges weakly to $N(0,\Id _D)$, where the bias vector $R_n \in \mathbb{R}^D$ and covariance matrix $P_n \in \mathbb{R}^{D \times D}$ have entries
\begin{align}
    R_{n}[\ell]&=-\frac{1}{2\pi i}\oint z^{\ell}\rho_n(z)dz,\quad P_n[\ell,m]=-\frac{1}{4\pi^2}\oint\oint z_1^{\ell}z_2^m\sigma_n^2(z_1,z_2)dz_1dz_2.\label{eq:main.cov}
\end{align}
The functions $\rho_{n}$ and $\sigma_{n}^{2}$ are defined in
Lemma~2.4 of \cite{xie2024cltlinearspectralstatistics},
applied to the matrix $S_{n}'$ of \eqref{eq:Sn_prime}; they depend on the matrices $\{\Lambda_{r},\Gamma_{r}\}_{r=1}^{k}$, with explicit formulas recalled in Appendix~\ref{sec:asym.bias.cov}. The contours are closed, positively oriented curves in the complex plane, each enclosing the interval $[0,k^2(1+\sqrt{C})^2s_{\Lambda}^2 s_{\Gamma}^2]$, with the two contours in \eqref{eq:main.cov} taken to be non-overlapping.
\end{theorem}

The deterministic equivalent Stieltjes transform $\deteqv{m}_n(z)$ and the functions $(\deteqv{g}_1^{(r)},\deteqv{g}_2^{(r)}), r=1,\ldots,k$ can be numerically evaluated following the procedure described in Theorem 1.5 of \cite{fan2017eigenvalue}. Additionally, the contour integrals defining $\deteqv{\mu}_n$, $P_n$, and $R_n$ can be approximated using the trapezoidal rule. Further details are provided in Section 2.4 of \cite{xie2024cltlinearspectralstatistics}.

\section{Proof of main results}\label{sec:proof.main}

\subsection{Proof of Lemma \ref{lem:map.asym}(a)}

First, we establish several useful properties of the Kreweras complement.

\begin{lem} \label{lem:even.odd}
    Let $\pi \in \text{NC}_2(2l)$ be a non-crossing pair partition of $\{1, 2, \dots, 2l\}$, and $K(\pi)$ be its Kreweras complement. Then the following properties hold:
    \begin{enumerate}
    \item[(a)]  Each pair in $\pi$ has one even and one odd
      element.
      We denote the odd element paired with $2 i$ by $2
      i' - 1$. Thus $\pi$ determines a bijection $i \to i'$
      on $[l]$.
    \item[(b)]  Every block in the Kreweras complement $K(\pi)$
      is entirely even (consists only of even elements) or entirely odd.
    \item[(c)] If $K(\pi)(i) = j$, viewing $K(\pi)$ as a permutation of
      $[2l]$, then $\{ i+1, j \}$ is a pair in
      $\pi$.
      If $\{2i_1,\ldots,2i_t\}$, $i_1<\ldots<i_t$ is an even block in
      $K(\pi)$, then for the pairs in $\pi$ we have
      \begin{equation*}
    i_2' = i_1 + 1, \quad \ldots, \quad  i_t' = i_{t-1}+1, \quad  i_1' = i_t + 1.
  \end{equation*}
  If $\{2i_1-1,\ldots,2i_t-1\}$, $i_1<\ldots<i_t$ is an odd block in
      $K(\pi)$, then
      \begin{equation*}
    i_1' = i_2, \quad \ldots, \quad  i_{t-1}' = i_t, \quad  i_t' = i_1.
  \end{equation*}
    \end{enumerate}
    \end{lem}
    \begin{proof}
     (a) Consider placing the integers ${1, 2, \dots, 2l}$ around a circle in clockwise order. Suppose, for the sake of contradiction, that there exists a pair $(i, j) \in \pi$ with $i < j$, such that both $i$ and $j$ are of the same parity—that is, both even or both odd. Then the number of elements lying strictly between $i$ and $j$ along the circular order (in either the clockwise or counter-clockwise direction) must be odd. This implies that the edge connecting $i$ and $j$ would necessarily intersect with at least one other edge formed by the remaining pairs. This introduces a crossing, contradicting the assumption that $\pi$ is non-crossing. It follows that each pair must consist of one index from an even position and one from an odd position.
     
    \vspace{0.3cm}
    
    \noindent (b) Recall from Lemma~\ref{lem:kreweras} that the Kreweras complement is given by $K(\pi)=\pi\gamma_{2l}$, where $\gamma_{2l}$ is the cyclic permutation $1\rightarrow 2\rightarrow\ldots\rightarrow 2l\rightarrow 1$. For any even index $2i$, we have $\gamma_{2l}(2i)=2i+1$. From the arguments above, since $\pi$ pairs even indices with odd indices, we have $\pi(2i+1)=2j$ for some $j$. Therefore, $\pi\gamma_{2l}(2i)=2j$. That is, the Kreweras complement $K(\pi)$ maps the even index $2i$ to another even index $2j$. Similarly, for any even index $2i+1$, we have $\gamma_{2l}(2i+1)=2i+2$, and $\pi$ maps the even index $2i+2$ to an odd index $2j+1$. As a result, in the Kreweras complement $K(\pi)$, the index $2i+1$ is connected to another odd index $2j+1$. It follows that every block in the Kreweras complement $K(\pi)$ consists entirely of either even or odd indices, completing the proof.
     
    \vspace{0.3cm}
    
    \noindent (c) By Lemma \ref{lem:kreweras}, the non-crossing pair
    partition $\pi$ can be recovered from its Kreweras complement
    $K(\pi)$ with $\pi=K(\pi)\gamma_{2l}^{2l-1}$, where
    $\gamma_{2l}^{2l-1}$ is the cyclic permutation $1\rightarrow
    2l\rightarrow2l-1\rightarrow\ldots\rightarrow 2\rightarrow 1$.
Applying the recovery formula, we have
   \begin{equation*}
     \pi(i+1) = K(\pi) (\gamma_{2l}^{2l-1}(i+1)) = K(\pi)(i) = j.
   \end{equation*}
Applying this to the even block, if $K(2i) = 2j$, then $\{ 2i + 1,
2j \}$ is a pair of $\pi$ and so $j' = i+1$, yielding the first
display.
For the odd block, if $K(2i-1) = 2j-1$, then $\{ 2i, 2j-1\}$ is a pair
of $\pi$ and $i' = j$, which gives the second display.

    \end{proof}

\begin{proof}
For each $l$, set $X=S_n^l$. By Cauchy--Schwarz for the trace inner product,
\begin{align*}
    \E [|\hat{\mu}_{n}|^2]= \frac{1}{p^2}\E[|\Tr X|^2]\leq \frac{1}{p}\E[\Tr(X^2)]=\E[\hat{\mu}_{n,2l}].
\end{align*}
Using the asymptotic approximation in \eqref{eq:lem.kreweras.Sigma}, we obtain
\begin{align*}
    \sup_n\E[|\hat{\mu}_{nl}|^2]<\infty.
\end{align*}
Hence $\{\hat{\mu}_{n,l}\}_n$ is uniformly integrable. Since
$\hat{\mu}_{n,l}-\tilde{\mu}_{n,l}\to 0$ almost surely by Lemma~\ref{lem:deteqv.mmts},
we conclude that 
\begin{equation}  \label{eq:concentrat}
    \hat{\mu}_n = \E\hat{\mu}_n+o_{a.s.}(1).
\end{equation}
It thus remains to establish the mapping for $\mathbb{E}[\hat{\mu}_{n,l}]$.

Write $\alpha_r$ in model (\ref{eq:general_model}) as $\alpha_r\deq G_r\Sigma_r^{1/2}$, where $G_r$ is a $I_r\times p$ matrix with i.i.d. standard Gaussian entries, then we can also express $S_n$ as
\begin{align*}
    S_n\deq \frac{1}{p}\sum_{r=1}^k\sum_{s=1}^k\Sigma_r^{1/2}G_r^TF_{rs}G_s\Sigma_s^{1/2},\quad F_{rs}:=U_r^TBU_s.
\end{align*}
Note here the matrices $\{G_r\}_r$, $\{F_{rs}\}_{rs}$ are rectangular. Let $m=\max\{I_1,\ldots,I_k,p\}$. For convenience, define
\begin{align*}
\Tilde{S}_n&=\frac{1}{m}\sum_{r=1}^k\sum_{s=1}^k\Tilde{\Sigma_r}^{1/2}\tilde{G_r}^T\Tilde{F}_{rs}\Tilde{G}_s\Tilde{\Sigma}_s^{1/2},
\end{align*}
where $\Tilde{G}_r$ are independent $m\times m$ matrices with i.i.d. standard Gaussian entries, and 
\begin{align}\label{eq:tilde.U.F.Sigma} 
    \Tilde{\Sigma}_{r}=\begin{bmatrix}
        \Sigma_r &0\\
        0&0
    \end{bmatrix}\in\mathbb{R}^{m\times m},\quad \Tilde{F}_{rs}=\Tilde{U}_{r}^TB\Tilde{U}_{s}\in\mathbb{R}^{m\times m},\quad  \Tilde{U}_{r}=\begin{bmatrix}
        U_r&0
    \end{bmatrix}\in\mathbb{R}^{n\times m}.
\end{align}

It is easy to see that
\begin{align*}
\Tilde{S}_n=\frac{p}{m}
\begin{bmatrix}
\Id _p\\
0_{m-p,p}
\end{bmatrix}S_n\begin{bmatrix}
\Id _p &0_{p,m-p}
\end{bmatrix}.
\end{align*}
As a result, for any positive integer $l>0$, we have
\begin{align*}
    \Tr S_n^l
    &=\left(\frac{m}{p}\right)^l\Tr \Tilde{S}_n^l.
\end{align*}

Define $\tilde{C}_r:=m^{-1/2}\tilde{G}_r$. By Lemma
\ref{lem:free.equiv}, there exists a sequence of non-commutative
probability spaces $(\mathcal{A}_n,\varphi_n)$ containing elements
$\{b_{rs}\},\{d_{rs}\}_{r,s=1}^k$ and free circular elements
$\{c_1,\ldots,c_k\}$ such that the joint distribution of
$\{b_{rs}\},\{d_{rs}\}_{r,s=1}^k$ with respect to $\varphi_n$
 \iain{asymptotically}  coincides with that of
$\{\tilde{F}_{rs}\},\{\tilde{\Sigma}_{s}^{1/2}\tilde{\Sigma}_{r}^{1/2}\}_{r,s=1}^k$
with respect to $\bar{\varphi} = m^{-1} \E \Tr$, and
 \begin{align}
 &\left(\frac{p}{m}\right)^lm^{-1}\mathbb{E}\Tr
   S_n^l=m^{-1}\mathbb{E}\Tr\Tilde{S}_n^l
   = \bar{\varphi}(\Tilde{S}_n^l)  \label{eq:mean-1} \\
     =
&\sum_{r_1,\ldots,r_l=1}^k\sum_{s_1,\ldots,s_l=1}^k
  \bar{\varphi}\left(\tilde{C}_{r_1}^T\Tilde{F}_{r_1s_1}\tilde{C}_{s_1}\Tilde{\Sigma}_{s_1}^{1/2}...\Tilde{\Sigma}_{r_l}^{1/2}\tilde{C}_{r_l}^T\Tilde{F}_{r_ls_l}\tilde{C}_{s_l}\Tilde{\Sigma}_{s_l}^{1/2}\Tilde{\Sigma}_{r_1}^{1/2}
  \right)  \notag \\
     =&\sum_{r_1,\ldots,r_l=1}^k\sum_{s_1,\ldots,s_l=1}^k\varphi_n\left(c_{r_1}^*b_{r_1s_1}c_{s_1}d_{r_2s_1}...c_{r_l}^*b_{r_ls_l}c_{s_l}d_{r_1s_l}\right)+o_n(1). \notag 
 \end{align}
At the second line, the factor $\Tilde{\Sigma}_{r_1}^{1/2}$ was moved
from front to end using the tracial property of $\bar{\varphi}$. 
By freeness of $\{c_r\}$ with $\{b_{rs},d_{rs}\}$, applying Lemma \ref{lem:kreweras.moments}, we have that
\begin{align}
    &\sum_{r_1,\ldots, r_l}\sum_{s_1,\ldots,s_l}\varphi_n\left(c_{r_1}^*b_{r_1s_1}c_{s_1}d_{r_2s_1}...c_{r_l}^*b_{r_ls_l}c_{s_l}d_{r_1s_l}\right) \nonumber\\
    =&\sum_{r_1,\ldots, r_l}\sum_{s_1,\ldots,s_l}\sum_{\pi\in\text{NC}(2l)}\kappa_\pi[c_{r_1}^*,c_{s_1},...,c_{r_l}^*,c_{s_l}]\varphi_{n,K(\pi)}\left[b_{r_1s_1},d_{r_2s_1},...,b_{r_ls_l},d_{r_1s_l}\right].\label{eq:intermediate}
\end{align}
For free standard circular variables $c_1,...,c_k$, by freeness and Definition \ref{def:circular}, any free cumulant of order $q\neq 2$ vanishes. Furthermore,
\begin{align*}
    \kappa_2(c_r^{\epsilon_1},c_s^{\epsilon_2})=1
\end{align*}
if and only if $r=s$, $\epsilon_1=1$, $\epsilon_2=*$, or $r=s$,
$\epsilon_1=*$, $\epsilon_2=1$. As a result,
$\kappa_\pi[c_{r_1}^*,c_{s_1},...,c_{r_l}^*,c_{s_l}]\neq 0$ 
only if $\pi$ is a pair-partition. Importantly, by Lemma
\ref{lem:even.odd} (a),
  in any such pair partition $\pi$, each pair
  has one even and one odd element, and the partition may be written
  $\pi = \{ (2i, 2i'-1)\}_{i \in [l]}$. Consequently
  $\kappa_\pi[c_{r_1}^*,c_{s_1},...,c_{r_l}^*,c_{s_l}]
  = \prod_i \kappa_2(c_{s_i}, c_{r_{i'}}^*)$
  vanishes unless $s_i = r_{i'}$ for each $i \in [l]$.
In other words,
\begin{align*}
    (\ref{eq:intermediate})&=
   \sum_{r_1}...\sum_{r_l}\sum_{\pi\in\text{NC}_2(2l)}
            \varphi_{n,K(\pi)} \left[b_{r_1 r_{1'}},d_{r_2r_{1'}},...,
                             b_{r_lr_{l'}},d_{r_1r_{l'}}\right].
\end{align*}

Recall that
$\{b_{rs}\},\{d_{rs}\}_{r,s=1}^k$ w.r.t. $\varphi_n$ follow the same
distribution as
$\{\Tilde{F}_{rs}\},\{\Tilde{\Sigma}_{s}^{1/2}\Tilde{\Sigma}_{r}^{1/2}\}$
w.r.t. $\bar{\varphi}$.

To avoid treble subscripts below, we temporarily write
\begin{equation*}
  \mathsf{s}(i) = \Tilde{\Sigma}_{r_i}^{1/2}, \quad
  \mathsf{u}(i) = \Tilde{U}_{r_i}, \quad
  \mathsf{f}(i,j) = \Tilde{F}_{r_i r_j}, \quad
  \mathsf{f}(i) = B \mathsf{u}(i) \mathsf{u}(i)^T.
\end{equation*}
Collecting the terms, we have
\begin{equation}  \label{eq:mean-2}
  \bar{\varphi}(\Tilde{S}_n^l)
 =\sum_{r_1}...\sum_{r_l}\sum_{\pi\in\text{NC}_2(2l)}
  \bar{\varphi}_{K(\pi)}\left[ \mathsf{f}(1,1'),
  \mathsf{s}(1') \mathsf{s}(2), \ldots, \mathsf{f}(l,l'),
  \mathsf{s}(l') \mathsf{s}(1)
  \right]+o_n(1).
\end{equation}

Setting $a_{2i -1} = \mathsf{f}(i,i')$ and $a_{2i} = \mathsf{s}(i')
\mathsf{s}(i+1)$, the summands above have the form
\begin{equation} \label{eq:mean-3}
  \bar{\varphi}_{K(\pi)} [a_1, \ldots, a_{2l}]
    = { \prod_{V \in K(\pi)}} \bar{\varphi}(V)
      [a_1, \ldots, a_{2l}],
\end{equation}
after applying Definition \ref{def:moments} for moments.
By Lemma \ref{lem:even.odd} (b), each block $V$ in $K(\pi)$ is
  even or odd. 
  In particular, an even block consists of elements from  $\{
  \mathsf{s}(i') \mathsf{s}(i+1) \}_{i \in [l]}$ while an odd block
  chooses from $ \{ \mathsf{f}(i,i') \}_{i \in [l]}$.

Let $V = \{ 2i_1,\ldots,2i_t \}$ with $i_1<\ldots<i_t$ be an even block in
$K(\pi)$. It contributes a term
\begin{equation*}
  \prod_{\nu = 1}^t \mathsf{s}(i_\nu') \mathsf{s}(i_\nu+1)
  = \mathsf{s}(i_t +1) \mathsf{s}(i_1 +1) \prod_{\nu=2}^t
    \mathsf{s}(i_{\nu-1}+1) \mathsf{s}(i_\nu+1),
\end{equation*}
where Lemma \ref{lem:even.odd}(c) provides the evaluation of
$i_\nu'$.
Using the tracial property of $\bar{\varphi}$,
\begin{align}
  \bar{\varphi} \Big( \prod_{\nu = 1}^t \mathsf{s}(i_\nu')
  \mathsf{s}(i_\nu+1) \Big)
   & =  \bar{\varphi} \Big( \prod_{\nu = 1}^t \mathsf{s}(i_\nu +1)^2 \Big)
     = \bar{\varphi} \Big( \prod_{v \in V} \mathsf{s}(v/2 +1)^2 \Big)
     \notag \\
   & = {\bar{\varphi} \Big( \prod_{v \in V} \Sigma_{r_{v/2+1}}
     \Big)}
     \label{eq:mean-4}
\end{align}
where in the last equality the tildes may be removed from the
$\tilde{\Sigma}_r$. 

An odd block $V = \{ 2i_1-1,\ldots,2i_t-1 \}$ in $K(\pi)$ contributes a
term
\begin{equation*}
  \prod_{\nu=1}^t \mathsf{f}(i_\nu, i_{\nu}')
  = \prod_{\nu=1}^{t-1} \mathsf{f}(i_\nu, i_{\nu+1}) \cdot
  \mathsf{f}(i_t, i_1)
  = \prod_{\nu=1}^{t-1} \mathsf{u}(i_\nu)^T B \mathsf{u}(i_{\nu+1})
  \cdot \mathsf{u}(i_t)^T B \mathsf{u}(i_1),
\end{equation*}
again using Lemma \ref{lem:even.odd}(c) to evaluate $i_\nu'$.
Again from the tracial property of $\bar{\varphi}$,
\begin{equation} \label{eq:mean-5}
  \bar{\varphi} \Big( \prod_{\nu = 1}^t \mathsf{f}(i_\nu, i_{\nu}') \Big)
  = \bar{\varphi} \Big( \prod_{\nu = 1}^t \mathsf{f}(i_\nu) \Big)
  = \bar{\varphi} \Big( \prod_{v \in V} \mathsf{f}((v+1)/2) \Big)
   = \bar{\varphi} \Big( \prod_{v \in V} F_{r_{(v+1)/2}} \Big)
\end{equation}

Combining  displays (\ref{eq:mean-1}) and
  (\ref{eq:mean-2})--~(\ref{eq:mean-5}),  and recalling
from Lemma \ref{lem:kreweras} that $|K(\pi)| = l+1$, and finally that
$\bar{\varphi} (a) = (p/m) {\varphi} (a)$ for deterministic $a$,
  we obtain
\begin{align}
  &\E \hat{\mu}_{nl}
   = p^{-1} \E \Tr S_n^l
    = (m/p)^{l+1} \bar{\varphi}( \Tilde{S}_n^l)  \label{eq:lem.kreweras.Sigma}\\
   =& \sum_{r_1, \ldots, r_l}^k
    \sum_{\pi\in\text{NC}_2(2l)}\prod_{V\in K(\pi)[1]}
    {\varphi}({\textstyle{\prod}}_{v\in V}F_{r_{(v+1)/2}})
    \prod_{V'\in K(\pi)[2]} {\varphi(\textstyle{\prod}_{v\in
    V'}\Sigma_{r_{v/2+1}})}+o_n(1). \notag
\end{align}
Together with (\ref{eq:concentrat}), this completes the proof of
Lemma~\ref{lem:map.asym}(a).
 \end{proof}

\subsection{Proof of Theorem \ref{thm:main}}\label{sec:evaluate.inverse}

\subsubsection{Existence and consistency}
First, we prove the existence of the method of moments estimator $\hat{\theta}_n$.  The argument relies on the following lemma, a quantitative version of the inverse function theorem adapted from Theorem~6.7.2 of~\cite{tao2006analysis}.
    \begin{lem}\label{lem:inverse}
      Let $E$ be an open subset of $\mathbb{R}^D$, and let
$x_0 \in E$ and $c>0, \rho>0$ be given, such that
  $B(x_0,\rho)\subseteq E$.    Let  $f:E\rightarrow \mathbb{R}^D$ be
continuously differentiable such that
        \begin{enumerate}
            \item[i.]  $Jf(x_0)\in\mathbb{R}^{D\times D}$ is invertible,
            \item[ii.] the minimum singular value $\sigma_{\min}(Jf(x_0))\geq c$,
            \item[iii.]  $\lVert Jf(x)-Jf(x_0) \rVert\leq {c/(2 D^2)}$ for all $x\in B(x_0,\rho)\subseteq E$.
        \end{enumerate}
        Then there exists an open set $U\subseteq B(x_0,\rho)$
        containing $x_0$, and an open set $V=B(f(x_0),c\rho/2)$ such
        that $f$ is a bijection from $U$ to $V$. In particular, there
        is an inverse map $f^{-1}:V\rightarrow U$. Furthermore, this
        inverse map is differentiable at every point of $V$
          and for each $x \in U$,
          \begin{equation*}
            (Jf^{-1})(f(x))=(Jf(x))^{-1}.
          \end{equation*}
      \end{lem}
We emphasize that while $U$ may depend on $f$, the radius $c
  \rho/2$ of $V$ does not: this is crucial for the consistency argument.
    \begin{proof}
        First, by following the proof of Theorem 6.7.2 in
        \cite{tao2006analysis}, we can without loss of generality assume that $x_0=0$, $f(x_0)=0$. Additionally, if $Jf(0)=\Id$
        and $\| Jf(x) - \Id \| \leq 1/(2D^2)$ for $x \in B(x_0,\rho)$,
        the desired result holds with $V=B(0,\rho/2)$. Here, we use the fact that for any matrix $A$, $\|A_{:,j}\|_2\leq \|A\|$.

        Next, define $\Tilde{f}(x):=(Jf(0))^{-1}f(x)$, so that
        $J\Tilde{f}(0)=\Id$. Condition iii. implies that
          \begin{equation*}
            \| J\tilde{f}(x) - \Id \|
            = \| (Jf(0))^{-1}[Jf(x) - Jf(0)] \|
            \leq \| (Jf(0))^{-1}\|c/(2D^2)
              \leq 1/(2D^2),
          \end{equation*}
so from the previous argument, there exists a neighborhood $U'$
containing $0$, and $V'=B(0,\rho/2)$ such that the result is
valid. Since $f(x)=Jf(0)\tilde{f}(x)$, it follows that $f$ is a
bijection from $U'$ to $Jf(0)V'$. Therefore, we conclude that
$V=B(0,c\rho/2)\subseteq Jf(0)V'$. The differentiability assertions
follow as in Theorem 6.7.2 and Exercise 6.7.4 of \cite{tao2006analysis}.
    \end{proof}
By assumption, there exist positive constants $c', \rho$ independent of $n$ such that for $n$ large enough, we have $\big|\text{det}(J\mathcal{F}_n(x))\big|\geq c'$ for each $x\in B(\theta,\rho)$. As a result, it suffices to prove that $\mathcal{F}_n$ satisfies conditions (ii)  and (iii) in the above lemma.

\begin{proof}[Proof of Theorem \ref{thm:main} (1)]
Recall the factorization $\mathcal{F}_n=\Psi_n\circ\Phi$, in which $\Psi_n$ is
a polynomial map and $\Phi$ is continuously differentiable. By
Assumption~\ref{ass:minimum}, $\|B\|$ and $\|U_r\|$ are bounded by a constant
$L>0$ independent of $n$. Setting $F_r=BU_rU_r^{T}$, then $ \|F_r\|\leq\|B\|\,\|U_r\|^{2}\leq L^{3}$. By Lemma~\ref{lem:map.asym}, every coefficient of $\Psi_n$ is of the form
$\varphi\!\left(\textstyle\prod_{r\in w}F_r\right)$ for some word $w$ with letters $r\in \{1,\ldots,k\}$, and length $\leq D$, and therefore
\[
  \Bigl\|\varphi\Bigl(\textstyle\prod_{r\in w}F_r\Bigr)\Bigr\|
  \leq\Bigl(\max_{r}\|F_r\|\Bigr)^{|w|}
  \leq L^{3D}.
\]
Hence the coefficients of $\Psi_n$ are uniformly bounded in $n$. As a result, we have for each $x\in B(\theta,\rho)$,
\begin{align}\label{eq:Lipschitz}
    \lVert J\Psi_n({\Phi(x)})\rVert,\lVert J\Phi(x)\rVert,\lVert
  J\mathcal{F}_n(x)\rVert \leq C
\end{align}
for a constant $C$ independent of $n$. Consequently, we have that  for
a constant ${c = c'/C^{D-1}}$, 
\begin{equation}  \label{eq:sigmin}
  \sigma_{\min}(J\mathcal{F}_n(x))
  \geq  |\text{det}(J\mathcal{F}_n(x))|/{\lVert
  J\mathcal{F}_n(x)\rVert^{D-1}
  \geq c.}
\end{equation}
From $J\mathcal{F}_n(x) = J\Psi_n(\Phi(x))J\Phi(x)$, we have that for any $x,x'\in B(\theta,\rho)$,
    \begin{align*}
        & \lVert J\mathcal{F}_n(x')-J\mathcal{F}_n(x)\rVert \\
        &\leq \lVert \left[J\Psi_n(\Phi(x'))-J\Psi_n(\Phi(x))\right]J\Phi(x')\rVert+\lVert J\Psi_n(\Phi(x))\left[J\Phi(x')-J\Phi(x)\right]\rVert.
    \end{align*}
    As a result, by the continuous differentiability of $\Phi$ and $\Psi_n$ and the uniform boundedness of the coefficients of $\Psi_n$, it follows that for $n\geq n_0$ with $n_0$ sufficiently large, there exists $\bar{\rho}>0$ that does not depend on $n$ such that for any $x,x'\in B(\theta,\bar{\rho})$,
    \begin{align*}
         \lVert J\mathcal{F}_n(x')-J\mathcal{F}_n(x)\rVert\leq \frac{c}{2D^2}.
    \end{align*}    
    With a slight abuse of notation, we subsequently take $\rho$
      to be the minimum of $\bar{\rho}$ and $\rho$.
Apply Lemma~\ref{lem:inverse} to the functions $\mathcal{F}_n, n \geq
n_0$ at $x_0 = \theta$ and set $\rho' = c \rho/2$ and $u_n =
\mathcal{F}_n(\theta)$. For $n \geq n_0$ the map $\mathcal{F}_n$ restricts to a bijection between an open
neighborhood $U_n\subseteq B(\theta,\rho)$ of $\theta$ and the open ball
$B(u_n,\rho')$; denote the
corresponding inverse by $\mathcal{F}_n^{-1}\colon B(u_n,\rho')\to U_n$.

Recall from equation (\ref{eq:Fn}) that $\hat{\mu}_n =
\mathcal{F}_n(\theta)+{o_{a.s.}(1)}= u_n+{o_{a.s.}(1)}$. Then we
conclude that for sample outcomes $\omega$ in a set of probability
one, for
{ $n \geq n_1(\omega) \geq n_0$,} we have $\hat{\mu}_n(\omega) \in
B(u_n, \rho')$ and 
the moment estimator
$\hat{\theta}_n(\omega)=\mathcal{F}_n^{-1}(\hat{\mu}_n(\omega))$
exists in $U_n$.  Moreover, for such $n$,
  \begin{equation*}
    \| \hat{\theta}_n - \theta \|
    = \| \mathcal{F}_n^{-1}(\hat{\mu}_n)-\mathcal{F}_n^{-1}(u_n) \|
    \leq M \| \hat{\mu}_n - u_n \|,
  \end{equation*}
using e.g. \cite[Theorem 9.19]{rudin1976principles}, where along with  $\|
A^{-1} \| \leq 1/\sigma_{\min}(A)$ 
and (\ref{eq:sigmin}), we have
\begin{align*}
  M & =  \sup_{y\in B(u_n,\rho')}\lVert
      J\mathcal{F}_n^{-1}(y)\rVert=\sup_{x\in U_n}\lVert
      (J\mathcal{F}_n(x))^{-1}\rVert \\
    & \leq \frac{1}{\inf_{x\in U_n}\sigma_{\min}(J\mathcal{F}_n(x))}
      \leq c^{-1},
\end{align*}
so that $ \| \hat{\theta}_n - \theta \| \convas 0.$

While this argument does not rule out the existence of another
solution $\hat{\theta}_n' \in \Theta \backslash U_n$, any a.s. consistent
sequence $\hat{\theta}_n'$ must eventually coincide with
$\hat{\theta}_n$.
Indeed, from (\ref{eq:Lipschitz}), it follows that $\mathcal{F}_n$ is
Lipschitz on $B(\theta,\rho)$ and so there exists $\rho_0 < \rho$ such
that $\mathcal{F}_n$ maps $B(\theta,\rho_0)$ into
$B(\mathcal{F}_n(\theta), \rho')$ for $n > n_0$. Thus $B(\theta,
\rho_0) \subset U_n$ and so if $\hat{\theta}_n'(\omega) \to \theta$ in
addition to $\hat{\theta}_n(\omega) \to \theta$ then for $n >
n_2(\omega)$ both $\hat{\theta}_n(\omega)$ and
$\hat{\theta}_n'(\omega) \in B(\theta, \rho_0) \subset U_n$ and since
$\mathcal{F}_n$ is a bijection there,
necessarily $\hat{\theta}_n'(\omega) = \hat{\theta}_n(\omega)$.
\end{proof}

\subsubsection{Central Limit Theorem}\label{subsec:proof.clt}
By Lemma \ref{lem:deteqv.mmts}, there exist deterministic equivalent moments $\deteqv{\mu}_n$, defined in Section \ref{subsec:deteqv}, such that 
\begin{align*}
    \hat{\mu}_n-\deteqv{\mu}_n=o_{a.s.}(1).
\end{align*}
Applying the central limit theorem for the empirical moments $\hat{\mu}_n$ stated in Theorem \ref{thm:clt}, we have that there exists $P_n\in\mathbb{R}^{D\times D}$, $R_n\in\mathbb{R}^{D}$, such that
\begin{align}\label{eq:clt.mu.n}
    P_n^{-1/2}\left(p(\hat{\mu}_n-\deteqv{\mu}_n) - R_n\right)\convd N(0,\Id _D).
\end{align}
Before proceeding with the proof, we introduce the following variant of the Delta method, which differs from the standard formulation in that we do not
  require any convergence of the mappings $H_{n}$.

A sequence of matrices $\{ A_n \}$ is said to be
  \textit{well-conditioned} if there exist positive constants $c < C$
  such that $c \leq \sigma_{\min}(A_n) \leq \sigma_{\max}(A_n) \leq C$.
\begin{lem}[Delta Method Variant]\label{lem:delta.var}
    Let $H_n:\mathbb{R}^D\rightarrow\mathbb{R}^D$ be twice
    differentiable at $\alpha_n$, with Jacobian denoted by
    $J{H_n}(\alpha_n)$ assumed to be well-conditioned. Denote
    $H_n=(h_{n1}(\cdot),\ldots,h_{nD}(\cdot))^T$. Assume that
    $\left|\frac{\partial^2}{\partial^{a_1}\ldots\partial^{a_D}}h_i(x)\right|\leq
    M$ with $|\sum a_i|=2$ for $x\in B(\alpha_n,r)$. If there exist a
    deterministic bias vector $\mu_n$ with bounded entries, and a well
    conditioned covariance matrix $\Lambda_n$, such that
    \begin{align*}
        \Lambda_n^{-1/2}\left[n(\hat{\alpha}_n-\alpha_n)-\mu_n\right]\convd N(0,\Id_k),
    \end{align*}
    then
    \begin{align*}
        \Lambda_n'^{-1/2}\left[n(H_n(\hat{\alpha}_n)-H_n(\alpha_n))-\mu'_n\right]\convd N(0,\Id_k),
    \end{align*}
    where $\Lambda_n'=J{H_n}(\alpha_n)\Lambda_nJ{H_n}(\alpha_n)^T,\quad \mu_n'=J{H_n}(\alpha_n)\mu_n$.
\end{lem}
 \begin{proof}
   Here we adapt the proof of Theorem 3.1 in \cite{vaart_1998}. From
   the differentiability of $H_n$, we have, when
     $\hat{\alpha}_n \in B(\alpha_n,r)$,
   \begin{align*}
      & \lVert H_n(\hat{\alpha}_n)-H_n(\alpha_n)-J{H_n}(\alpha_n)(\hat{\alpha}_n-\alpha_n)\rVert\leq \frac{DM}{2}\lVert \hat{\alpha}_n-\alpha_n\rVert^2.
   \end{align*}
   By Slutsky's theorem, we have that $n\lVert\hat
   \alpha_n-\alpha_n\rVert^2=o_p(1)$, and since $\Lambda_n'$ is well-conditioned,
   \begin{align*}
      &\Lambda_n'^{-1/2}\left[n(H_n(\hat{\alpha}_n)-H_n(\alpha_n))-\mu'_n\right]\\
      =&\Lambda_n'^{-1/2}\left[J{H_n}(\alpha_n)n(\hat{\alpha}_n-\alpha_n)-\mu'_n\right] +o_P(1)\\
      =&\Lambda_n'^{-1/2}J{H_n}\Lambda_n^{1/2}\Lambda_n^{-1/2}\left[n(\hat{\alpha}_n-\alpha_n)-\mu_n\right] +o_P(1)\\
      =&A_nZ_n+o_p(1),
    \end{align*}
where $A_n:=\Lambda_n'^{-1/2}J{H_n}(\alpha_n)\Lambda_n^{1/2}$, and $Z_n:=\Lambda_n^{-1/2}\left[n(\hat{\alpha}_n-\alpha_n)-\mu_n\right]\convd N(0,\Id_k)$. A direct computation gives
\begin{align*}
  A_nA_n^{T}
  =\Lambda_n'^{-1/2}J{H_n}(\alpha_n)\Lambda_n J{H_n}(\alpha_n)^{T}\Lambda_n'^{-1/2}
  =\Lambda_n'^{-1/2}\Lambda_n'\Lambda_n'^{-1/2}=\Id_D,
\end{align*}
so each $A_n$ is orthogonal. The set of $D\times D$ orthogonal matrices is compact, so every subsequence ${A_{n_j}}$ has a sub-subsequence $A_{n_{j_k}}\to A_\infty\in O(D)$. Along that sub-subsequence, by Slutsky's theorem, we have that $A_{n_{j_k}}Z_{n_{j_k}}\convd =N(0,\Id_D)$. Because every subsequence has a sub-subsequence converging to the same limit $N(0,\Id_D)$, the full sequence $A_nZ_n\convd N(0,\Id_D)$.
\end{proof}
\begin{proof}[Proof of Theorem \ref{thm:main} (2)]
First, we validate that when setting $H_n=\mathcal{F}_n^{-1}$ the conditions of Lemma \ref{lem:delta.var} are satisfied.
Using the chain rule and the formula for the derivative of the inverse of a matrix, we have
\begin{align*}
    \frac{\partial}{\partial y_i}J(\mathcal{F}_n^{-1})(y)
    &=\frac{\partial}{\partial y_i}\left[(J\mathcal{F}_n(\mathcal{F}_n^{-1}(y)))^{-1}\right]\\
    &=-J(\mathcal{F}_n^{-1})(y)\cdot \frac{\partial}{\partial y_i}J\mathcal{F}_n(\mathcal{F}_n^{-1}(y)) \cdot J(\mathcal{F}_n^{-1})(y)\\
    &=-J(\mathcal{F}_n^{-1})(y)\cdot \left(\sum_{r=1}^D\frac{\partial J\mathcal{F}_n}{\partial x_r}(\mathcal{F}_n^{-1}(y))\frac{\partial [\mathcal{F}_n^{-1}]_r}{\partial y_i}\right) \cdot J(\mathcal{F}_n^{-1})(y).
\end{align*}
As a result, by the boundedness of the first and second derivatives of
$\mathcal{F}_n$, and the boundedness of each entry of
$J(\mathcal{F}_n^{-1})$, we conclude that the second derivatives of
$\mathcal{F}_n^{-1}$ are also bounded.
Apply Lemma \ref{lem:delta.var} to the CLT established in
\eqref{eq:clt.mu.n} with $H_n = \mathcal{F}_n^{-1}, \alpha_n =
\tilde{\mu}_n, \hat{\alpha}_n = \hat{\mu}_n$ etc.
Setting $\tilde{\theta}_n = \mathcal{F}_n^{-1}(\tilde{\mu}_n)$ and
$\tilde{D}_n = J \mathcal{F}_n^{-1}(\tilde{\mu}_n)$, we get
 \begin{equation*}
        (\tilde{D}_nP_n
        \tilde{D}_n^T)^{-1/2}\left[p\left(\hat{\theta}_n
            -\tilde{\theta}_n\right)-\tilde{D}_nR_n\right]\convd N(0,\Id_D).
\end{equation*}
Together with Lemma \ref{lem:deteqv.mmts} and the consistency of MoM
estimators $\hat{\theta}$, we have that the difference $\tilde{\theta}_n-\theta =
\mathcal{F}_n^{-1}(\deteqv{\mu}_n)-\mathcal{F}_n^{-1}(\hat{\mu}_n)+\mathcal{F}_n^{-1}(\hat{\mu}_n)-\theta\rightarrow0$.
Since the functions $J \mathcal{F}_n$ are uniformly continuous and
satisfy (\ref{eq:sigmin}),
\begin{equation*}
  \tilde{D}_n - D_n = [J \mathcal{F}_n(\tilde{\theta}_n)]^{-1}
                        - [J \mathcal{F}_n(\theta)]^{-1}
                    \to 0,    
\end{equation*}
and hence $D_n^{-1} \tilde{D}_n \to \Id$. We may therefore replace
$\tilde{D}_n$ with $D_n$ in the penultimate display, recovering
Theorem \ref{thm:main}(2). \qedhere
\end{proof} 

\newpage
\bibliographystyle{plainnat}
\bibliography{reference}  

\newpage

\begin{appendix}

    \section{Definition of Matrix Normal}\label{sec:matrix-normal}
   
   \begin{definition}
   Given non-negative definite matrices $U$ ($n \times n$) and $V$ ($p \times p$),
   we say that an $n \times p$ matrix $Y \sim N(0, U \otimes V)$ if there exist
   matrices $A$, $B$ and $Z$ of conforming dimensions so that
   \begin{enumerate}
     \item[(i)] $Y = AZB$,
     \item[(ii)] $A$ and $B$ are deterministic with $AA^T = U$, $B^T B = V$,
     \item[(iii)] $Z$ is random with i.i.d.\ $N(0,1)$ entries.
   \end{enumerate}
   \end{definition}
   
   \begin{rem}\label{rem:non-uniqueness}
   The non-uniqueness of $A$, $B$ in (ii) does not invalidate the
   definition.
   Indeed, using the characteristic function $\varphi_Y(T) = \mathbb{E}\exp(i\,\Tr\,T^TY)$
   and its specific form $\varphi_Z(S) = \exp\!\bigl(-\tfrac{1}{2}\Tr(S^TS)\bigr)$,
   we have
   \[
     \varphi_Y(T) = \mathbb{E}\exp(i\,\Tr\,T^T\! AZB) = \mathbb{E}\exp(i\,\Tr\,S^TZ)
   \]
   with $S^T= BT^TA$ (from cyclicity of trace) and
   \[
     \Tr(S^TS) = \Tr(BT^TAA^TT B^T) = \Tr(T^TU T V).
   \]
   Thus $\varphi_Y(T)$ depends on $A$ and $B$ only through $U$ and $V$. 
   \end{rem}
   
   \begin{rem}\label{rem:different-dimensions}
   In consequence, we can verify that two representations are equal in distribution:
   $A_1 Z_1 B_1 =_d A_2 Z_2 B_2$ by verifying that $A_1 A_1^T= A_2 A_2^T$ and
   $B_1^TB_1 = B_2^TB_2$.
   This is especially useful when $A_i (n \times a_i), Z_i (a_i \times
   b_i)$ and $B_i (b_i \times p)$ have different dimensions.  
   \end{rem} 
   
   \begin{rem}\label{rem:generative-def}
   The most common definition (e.g., \cite{muirhead2009aspects}) requires that $U$
   and $V$ be invertible. This generative definition allows $U$ and $V$ also to be
   singular, as is needed for our applications. It is a particular case of the
   definition given in \cite[p.~69]{gupta2018matrix}.
   \end{rem}
   
   \section{Asymptotic bias and covariance}\label{sec:asym.bias.cov}
   In this section, we give explicit formulas for the asymptotic bias and covariance of the
   MoM estimators, in particular, the functions $\rho_n$ and $\sigma_n^{2}$ of
   Theorem~\ref{thm:clt}, defined in Lemma~2.4 of \cite{xie2024cltlinearspectralstatistics}. Define
   \begin{align*}
        A(z)=\Id+\sum_{r=1}^k\deteqv{g}_2^{(r)}(z)\Lambda_r,\quad B(z)=\Id+\sum_{r=1}^k\deteqv{g}_1^{(r)}(z)\Gamma_r,
   \end{align*}
   so that $\deteqv{R}(z)=-zB(z)$ for $\deteqv{R}$ defined in \cite{xie2024cltlinearspectralstatistics}. For each $a,b=0,1,\ldots,k$, further define 
   \begin{gather*}
           \bar{h}_{2}^{ab}=\frac{1}{p}\Tr \left[A^{-1}(z_1)\Lambda_aA^{-1}(z_2)\Lambda_b\right],~\Xi_2^{ab}=\frac{1}{z_1z_2}\frac{1}{p}\Tr \left[B^{-1}(z_1)\Gamma_aB^{-1}(z_2)\Gamma_b\right],
   \end{gather*}
   with the convention $\Lambda_{0} := \Id$ and $\Gamma_{0} := \Id$. Throughout, whenever
   a quantity is presented as a vector or matrix (rather than a single entry),
   its indices are understood to range over $\{1,\ldots,k\}$, excluding
   $0$; thus $\bar h_{2} := (\bar h_{2}^{ab})_{ab}$ and
   $\Xi_2 := (\Xi_{2}^{ab})_{ab}$ are $k\times k$ matrices. Let
   \begin{align*}
       S=\Id -\bar{h}_2\Xi_2, 
   \end{align*}
   then the asymptotic covariance function is
   \begin{align*}
       \sigma_n^2(z_1,{z}_2)&=-{2}\frac{\partial^2}{\partial z_2\partial z_1}\log\det(S).
   \end{align*}
   For each $a,b=0,1,\ldots,k$, define 
   \begin{gather*}
       \bar{h}^{abc}_3=\frac{1}{N}\Tr \left[A^{-1}\Lambda_aA^{-1}\Lambda_bA^{-1}\Lambda_c\right],\quad \Xi_3^{abc}=-\frac{1}{z^3}\frac{1}{N}\Tr \left[B^{-1}\Gamma_aB^{-1}\Gamma_bB^{-1}\Gamma_c\right],
   \end{gather*}
   and 
   \begin{gather*}
       \Xi_2^{\cdot 0}=(\Xi_2^{a0})_a\in\mathbb{R}^k,~\Xi_3^{\cdot 0\cdot}= (\Xi_3^{a0b})_{ab},~\Xi_3^{\cdot\cdot c}=(\Xi_3^{ab c})_{ab},~\bar{h}_3^{\cdot\cdot c}=(\bar{h}_3^{ab c})_{ab}\in\mathbb{R}^{k\times k}.
   \end{gather*}
   Further define
   \begin{align*}
       m^c=\Tr\left(S^{-1}\bar{h}_3^{\cdot\cdot c}\Xi_2\right),\quad \deteqv{d}_{0}&=\Tr\left(S^{-1}\bar{h}_2\Xi_3^{\cdot 0\cdot}\right)-m^T\Xi_2^{\cdot 0},\quad \deteqv{d}^c=\Tr\left(S^{-1}\bar{h}_2\Xi_3^{\cdot\cdot c}\right)-m^T\Xi_2^{\cdot c},
   \end{align*}
   and $m=(m^c)_c$, $\deteqv{d}=(\deteqv{d}^c)_c$.
   Then the asymptotic bias function is
   \begin{align*}
       \mu_n(z)&=\deteqv{d}_{0}+ \tilde{d}^TS^{-1}\bar{h}_2\Xi_2^{\cdot 0}.
   \end{align*}

   \section{Proof of Lemma \ref{lem:param.bdd.jacobian}}\label{app:proof.nested}
       The cornerstone of our argument is the Cauchy-Binet formula, stated below. 
       \begin{lem}[Cauchy-Binet Formula]\label{lem:Cachy-Binet}
           Let $A$ be an $m\times n$ matrix, and $B$ be an $n\times m$ matrix, with $n>m$. 
           Denote by $[n]$ the set $\{1,\ldots,n\}$, and by ${[n] \choose m}$ the set of $m$-combinations of $[n]$, that is, subsets of $[n]$ with size $m$. For $S\in {[n] \choose m}$, let $A_{[m]\times S}$ be the $m\times m$ minor of $A$ coming from the columns $S$, and similarly let $B_{ S\times [m]}$ be the $m\times m$ minor of $B$ coming from the rows $S$. Then,
           \begin{align}
               \text{det}(AB)=\sum_{S\in {[n] \choose m}}\text{det}(A_{[m]\times S})\text{det}(B_{ S\times [m]}).
           \end{align}
       \end{lem}
       Below, in Sections \ref{subsec:poly} to \ref{subsec:step}, we
       prove Lemma~\ref{lem:param.bdd.jacobian} for each of the
       eigenvalue decay models introduced in Section~\ref{sec:main}. We
       treat the step decay model last, as it involves the largest number of parameters. It suffices to establish the result for $\mathcal{F}_n^{(1)}(\theta_1;\hat{\theta}_2,\ldots,\hat{\theta}_k)$ with $\hat{\theta}_2,\ldots,\hat{\theta}_k$ already estimated. The goal is to show that there exists a parameter space $\Theta$ and a constant $c' > 0$ such that for all $\theta_1 \in \Theta$, $|\text{det}(J{\mathcal{F}}_n^{(1)}(\theta_1;\hat\theta_{2},\ldots,\hat\theta_{k}))|>c'$ for sufficiently large $n$. The remaining proof follows from induction. 
       
       \subsection{Polynomial eigenvalue decay}\label{subsec:poly}       
       Let $\theta_1=(\tau, \pi)$, with $\tau > 0$, $0 < \pi < 1$. The polynomial decay model is given by $g_{1}(x;\theta_1) = \tau(1-x/\pi)^s_+$ for a given integer $s>0$. 
       Here, the mapping ${\mathcal{F}}_n^{(1)}:\mathbb{R}^2\rightarrow\mathbb{R}^2$ is constructed as the composition of the mappings ${\mathcal{F}}_n^{(1)}=\Psi^{(1)}\circ\Phi^{(1)}$, where $\Phi^{(1)}:\mathbb{R}^2\rightarrow\mathbb{R}^{k+1}$ maps $\theta_1$ 
       to $(\phi({g}_1),\phi({g}_1\hat{g}_2),\ldots,\phi({g}_1\hat{g}_k),\phi({g}_1^2))$, 
        with $\hat{g}_r(x)=g_r(x;\hat{\theta}_r)$, and $\Psi^{(1)}:\mathbb{R}^{k+1}\rightarrow\mathbb{R}^2$ is a quadratic function of these moments. Taking the Jacobian of each, we have
       \begin{align*}
           J\Phi^{(1)}&=\begin{bmatrix}
               \frac{\partial}{\partial \tau}\phi(g_1)& \frac{\partial}{\partial \pi}\phi(g_1)\\
                \frac{\partial}{\partial \tau}\phi(g_1\hat{g}_2)&\frac{\partial}{\partial \pi}\phi(g_1\hat{g}_2)\\
                ...&...\\
                \frac{\partial}{\partial \tau}\phi(g_1\hat{g}_k)&\frac{\partial}{\partial \pi}\phi(g_1\hat{g}_k)\\
                \frac{\partial}{\partial \tau}\phi(g_1^2)& \frac{\partial}{\partial \pi}\phi(g_1^2)
           \end{bmatrix},
       \end{align*}
       and
       \begin{align*}
            J\Psi^{(1)} &= \begin{bmatrix}
               \varphi(F_{1}) &0&...&0&0\\
               J_{21}&2\varphi(F_{1}) \varphi(F_{2})&...&2\varphi(F_{1}) \varphi(F_{k})&\varphi(F_{1})^2,  
           \end{bmatrix}
       \end{align*}
       where $J_{21}=2\phi(g_1)\varphi(F_1^2)+2\sum_{r=2}^k\phi(\hat{g}_r)\varphi(F_1F_r).$ 
       For $r=2,...,k$, define
       \begin{align*}
           c_r&:=\frac{\partial}{\partial \tau}\phi(g_1)\frac{\partial}{\partial \pi}\phi(g_1\hat{g}_r)-\frac{\partial}{\partial \pi}\varphi(g_1)\frac{\partial}{\partial \tau}\phi(g_1\hat{g}_r).
       \end{align*}
       Applying the Cauchy-Binet formula in Lemma \ref{lem:Cachy-Binet}, we have that
       \begin{align*}
           |J{\mathcal{F}}_n^{(1)}|=\varphi({F}_1^2)\left(\varphi({F}_1)c_1+2\sum_{r=2}^k\varphi({F}_r)c_r\right).
       \end{align*}
       Note that by definition,  we have $\varphi({F}_r)=\varphi(B_1U_rU_r^T)>0$ for each $r$. As a result, it suffices to prove that 
       $c_1,\ldots,c_k$ satisfy
       \begin{align}\label{eq:cond}
           c_r\geq 0,\quad\sum_rc_r> 0.
       \end{align}
       
       Plugging in the decay model ${g}_1(x)={\tau}(1-x/{\pi})^s_+$, we have, 
       \begin{align*}
           c_1&=\frac{\partial}{\partial \pi}\int_0^1 g_1(x,\pi,\tau)dx\frac{\partial}{\partial \tau}\int_0^1g_1^2(x,\pi,\tau)dx-\frac{\partial}{\partial \tau}\int_0^1 g_1(x,\pi,\tau)dx\frac{\partial}{\partial \pi}\int_0^1g_1^2(x,\pi,\tau)dx\\
             &=\frac{\tau}{s+1}\frac{2\pi\tau}{2s+1}-\frac{\pi}{s+1}\frac{\tau^{2}}{2s+1}=\frac{\pi\tau^{2}}{(s+1)(2s+1)}>0\\
             c_r&=\frac{\partial}{\partial \pi}\int_0^1 g_1(x,\pi,\tau)dx\frac{\partial}{\partial \tau}\int_0^1 g_1(x,\pi,\tau)\hat{g}_r(x)dx
           -\frac{\partial}{\partial \tau}\int_0^1 g_1(x,\pi,\tau)dx\frac{\partial}{\partial \pi}\int_0^1 g_1(x,\pi,\tau)\hat{g}_r(x)dx\\
           &=\frac{\tau}{s+1}\int_0^{\pi}(1-x/\pi)^s\hat{g}_r(x)dx-\frac{\pi}{s+1}\frac{s\tau}{\pi^2}\int_0^{\pi}(1-x/\pi)^{s-1}x\hat{g}_r(x)dx\\
           &=\frac{\tau}{s+1}\int_0^{\pi}h(x)\hat{g}_r(x)dx,
       \end{align*}
       where
       \begin{align*}
           h(x) =(s+1)(1-x/\pi)^s-s(1-x/\pi)^{s-1}.
       \end{align*}
       Then it is easy to check that when $x<\pi/(s+1)$, $h(x)>0$, when $x=\pi/(s+1)$, $h(x)=0$, when $x>\pi/(s+1)$, $h(x)<0$, and $\int_0^{\pi}h(x)=0$. Since $\hat{g}_r(x)$ follows one of the decay models introduced in Section \ref{sec:main}, we have that the smaller $x$ is, the more (or equal) weight is given by $\hat{g}_r(x)$. As a result, we always have $\int_0^\pi h(x)\hat{g}_r(x)dx\geq 0$. In other words, the condition (\ref{eq:cond}) is satisfied. 
       \subsection{Exponential eigenvalue decay} \label{subsec:expo}
       Let $\theta_1=(\tau_1,\tau_2)$, where $\tau_1,\tau_2>0$. The decay function is given by
       $g_1(x) = \tau_1\exp(-\tau_2x)$. Since $\theta_1$ also has two degrees of freedom, the structures of $\Phi^{(1)}$, $\Psi^{(1)}$, $J\Phi^{(1)}$, and $J\Psi^{(1)}$ remain unchanged from the polynomial decay model. Consequently, the determinant formula for $J{\mathcal{F}}_n^{(1)}$ remains the same by an application of the Cauchy–Binet formula, yielding
       \begin{align*}
           |J{\mathcal{F}}_n^{(1)}|=\varphi({F}_1^2)\left(\varphi({F}_1)c_1+2\sum_{r=2}^k\varphi({F}_r)c_r\right).
       \end{align*}
       Therefore, it suffices to establish that $c_r \geq 0$ and $\sum_r c_r > 0$. Under the exponential decay model, we have
       \begin{align*}
           c_1&=\frac{\partial}{\partial \tau_1}\int_0^1 g_1(x,\tau_1,\tau_2)dx\frac{\partial}{\partial \tau_2}\int_0^1g_1^2(x,\tau_1,\tau_2)dx\\
           &\quad\quad\quad\quad\quad\quad\quad\quad\quad\quad-\frac{\partial}{\partial \tau_2}\int_0^1 g_1(x,\tau_1,\tau_2)dx\frac{\partial}{\partial \tau_1}\int_0^1g_1^2(x,\tau_1,\tau_2)dx\\
             &=\frac{1}{\tau_2}(1-e^{-\tau_2})(-2\tau_1^2)(1-e^{-2\tau_2})-(-\tau_1)(1-e^{-\tau_2})\frac{\tau_1}{\tau_2}(1-e^{-2\tau_2})\\
             &=-\frac{\tau_1^2}{\tau_2}(1-e^{-\tau_2})(1-e^{-2\tau_2}),
       \end{align*}
       and for $r=2,\ldots,k$,
       \begin{align*}
           c_r&=\frac{\partial}{\partial \tau_1}\int_0^1 g_1(x,\tau_1,\tau_2)dx\frac{\partial}{\partial \tau_2}\int_0^1 g_1(x,\tau_1,\tau_2)\hat{g}_r(x)dx\\
           &\quad\quad\quad\quad\quad\quad\quad\quad\quad\quad-\frac{\partial}{\partial \tau_2}\int_0^1 g_1(x,\tau_1,\tau_2)dx\frac{\partial}{\partial \tau_1}\int_0^1 g_1(x,\tau_1,\tau_2)\hat{g}_r(x)dx\\
           &=\frac{1}{\tau_2}(1-e^{-\tau_2})(-\tau_1\tau_2)\int_0^1e^{-\tau_2x}\hat{g}_r(x)dx-(-\tau_1)(1-e^{-\tau_2})\int_0^1e^{-\tau_2x}\hat{g}_r(x)dx=0.
       \end{align*}

       \subsection{Step eigenvalue decay}\label{subsec:step}
        Let $\theta_1=(t_{1},\ldots,t_{s-1},a_{1},\ldots,a_{s})$, where $t_j>0$, ${\textstyle\sum}_{j=1}^{s-1}t_j<1$, with $d_1=2s-1$. The decay function is 
        \begin{align*}
        g_1(x) = a_1\mathbbm{1}[0\leq x\leq t_1]+{\textstyle\sum}_{j=2}^{s} a_j\mathbbm{1}[{\textstyle\sum}_{i=1}^{j-1}t_i<x\leq {\textstyle\sum}_{i=1}^{j}t_i].
        \end{align*}
       Given $\hat{\theta}_2,\ldots,\hat{\theta}_k$, the mapping ${\mathcal{F}}_n^{(1)}:\mathbb{R}^{d_1}\rightarrow \mathbb{R}^{d_1}$ is the composition of ${\Psi}^{(1)}\in\mathbb{R}^{d}\rightarrow \mathbb{R}^{d_1}$ and ${\Phi}^{(1)}\in\mathbb{R}^{d_1}\rightarrow \mathbb{R}^{d}$, with $d=\frac{(d_1+2)(d_1+3)!}{12}$, where $\Phi^{(1)}$ maps ${\theta}_1$ to 
       \begin{align}\label{eq:step.moments}
           \{\phi({g}_1^{l_1}\hat{g}_2^{l_2}\ldots\hat{g}_k^{l_k}):l_1>0,l_2,\ldots,l_k\geq 0,{\textstyle\sum}_{r=1}^kl_r\leq d_1\},
       \end{align}
       and ${\Psi}_n^{(1)}$ is polynomial in these moments. 
       
       {To characterize the structure of $J{\Psi}_n^{(1)}$, we order the moments in (\ref{eq:step.moments}) as follows: first by the total degree $\sum_{r=1}^k l_r$, then respectively by $l_2\ldots l_k$ and finally by $l_1$. Under this ordering, the sequence begins with $\phi(g_1)$, followed by $\phi(g_1 \hat{g}_2), \dots, \phi(g_1 \hat{g}_k)$ and $\phi(g_1^2)$. For the $j$th element in the ordered list, denote the total degree as $d_j$. Under this ordering, the mapping ${\Psi}_{nl}^{(1)}$ depends only on elements of total degree at most $l$. Consequently, the Jacobian $J\Psi_{n}^{(1)} \in \mathbb{R}^{d_1 \times d}$ exhibits a staircase structure, satisfying the sparsity pattern
       \begin{align*}
           (J\Psi_{n}^{(1)})_{l,j} = 0 \quad \text{for all } j \text{ such that } d_j > l,\quad l=1,\dots,d_1,\quad j=1,\dots,d.
       \end{align*}
       Here, we partition the indices $1,\ldots,d$ into steps $S_l = \{j : d_j = l\}$ for $l=1,\dots,d_1$.}
        
       {Applying the Cauchy-Binet formula in Lemma \ref{lem:Cachy-Binet} to $J{\mathcal{F}}_n^{(1)}=J{\Psi}_n^{(1)}\cdot J{\Phi}^{(1)}$, we observe that the staircase structure of $J\Psi_{n,1}$ restricts the non-vanishing terms in the determinant expansion. In particular, a $d_1 \times d_1$ submatrix of $J\Psi_{n,1}$ can be non-singular only if it satisfies a sequential selection constraint. Let $c_l$ denote the number of columns selected from the $l$-th step $S_l$. For each $l = 1, \dots, d_1$, the submatrix must select at least one column from the corresponding step $S_l$, i.e., $c_l > 0$, unless the total column budget $d_1$ has already been exhausted by the preceding steps, with $\sum_{i=1}^{l-1} c_i = d_1$. Consequently, any submatrix contributing to the determinant must fill the steps in order, potentially selecting multiple columns from earlier steps but not bypassing any intermediate step.}
       
       Consider the submatrix made up of the last column in each step, which correspond to partial derivatives with respect to $\{\phi(g_1),\ldots,\phi(g_1^{d_1})\}$. The determinant of this submatrix is $\varphi(F_1)^{d_1(d_1+1)/2}$. Now consider the corresponding submatrix in $J{\Phi}^{(1)}$, which represents the partial derivatives of $\{\phi(g_1),\ldots,\phi(g_1^{d_1})\}$ with respect to $\theta_1=(t_1,\ldots,t_{s-1},a_1,\ldots,a_{s})$. {Without loss of generality, assume that $a_1>\ldots>a_s\geq 0$.} From the proof of Proposition 1 in \cite{baiyaochen2010}, this submatrix takes the form
       \begin{align*}
           \begin{bmatrix}
               a_1-a_{s}&\ldots&a_{s-1}-a_{s}&t_1&\ldots&t_{s}\\
               a_1^2-a_{s}^2&\ldots&a_{s-1}^2-a_{s}^2&2t_1a_1&\ldots&2t_{s}a_{s}\\
               \vdots\\
               a_1^{2s-1}-a_{s}^{2s-1}&\ldots&a_{s-1}^{2s-1}-a_{s}^{2s-1}&({2s-1})t_1a_1^{2s-2}&\ldots&({2s-1})t_{s}a_{s}^{2s-2}
           \end{bmatrix}
       \end{align*}
       with non-zero determinant. In the Leibniz formula of the determinant, it is easy to show that the term corresponding to the highest order in the largest eigenvalue $a_1$ is 
       \begin{align*}
           &(2s-1)t_1a_1^{3s-2}(s-1)t_{s}a_{s}^{s-2}\prod_{i=2}^{s-1}(a_i^{s-1+i}-a_{s}^{s-1+i})(i-1)t_ia_i^{i-2}\\
          -&st_1a_1^{3s-2}t_{s}\prod_{i=2}^{s-1}(a_i^{2s-i}-a_{s}^{2s-i})(s+1-i)t_ia_i^{s-i}\\
          =&s^{-1}(\prod_{i=1}^{s}it_i)a_1^{3s-2}\left[(2s-1)a_{s}^{s-2}\prod_{i=2}^{s-1}(a_i^{s-1+i}-a_{s}^{s-1+i})a_i^{i-2}-s\prod_{i=2}^{s-1}(a_{s+1-i}^{s-1+i}-a_{s}^{s-1+i})a_{s+1-i}^{i-2}\right].
       \end{align*}
       As a result, we have that the component contributing to the determinant of $J{\mathcal{F}}_n^{(1)}$ corresponding to $\{\varphi(g_1),\ldots,\varphi(g_1^{d_1})\}$ scales at the rate of $O(a_1^{3s-2})$ as $a_1$ increases.
       
       Recall that any submatrix of $J{\Psi}_n^{(1)}$contributing to the determinant must include at least one column from each step of the staircase when descending from the first to the last. Consequently, all other contributing terms can scale with $a_1$ at most at the rate $O(a_1^{3d_1-3})$. Thus, for sufficiently large $a_1$, $|J{\mathcal{F}}_n^{(1)}|$ remains bounded away from zero.
       
       \section{Additional proofs}
       \subsection{Proof of Lemma \ref{lem:free.equiv}}    \label{app:proof.free.equiv}
           First, we prove the existence of the sequence $(\mathcal{A}_m,\varphi_m)$. For each $m$, let $(\mathcal{A}_{m,1},\varphi_m^{(1)})$ be an NCP that contains $k$ free circular elements $c_1,\ldots,c_k$. Let  $(\mathcal{A}_{m,2},\varphi_m^{(2)})$ be the NCP $(\mathbb{C}^{m \times m},m^{-1}\Tr )$ that contains the original $\{D_m^{(\nu)}\}_{1\leq \nu\leq \ell }$. Then, the existence of $(\mathcal{A}_m,\varphi_m)$ follows directly from Theorem \ref{thm:free.prod}. 
       
           To establish \eqref{eq:conv.moments.deteqv}, we adapt the proof of Theorem 11.2.1 in \cite{bose2021random}. The principal differences are twofold. First, Theorem 11.2.1 derives an asymptotic limit under the assumption that $$(D_m^{(1)},\ldots, D_m^{(\ell)})\rightarrow (d_1,\ldots,d_{\ell})$$ for some $d_1,\ldots,d_{\ell}$, whereas we work with finite-$n$ equivalents $(d_m^{(1)},\ldots,d_m^{(\ell)})$. Second, \cite{bose2021random} consider Gaussian elliptic matrices with correlation $\rho$ and the associated free elliptic elements, denoted $\{E_m^{(\nu)}\}_{1\leq \nu\leq k}$ and $\{e_\nu\}_{1\leq \nu\leq m}$ \cite[Section 7.4,~ Definition 3.6.2]{bose2021random}, with normalized Gaussian elliptic matrices denoted by $\bar{E}_m:=m^{-1/2}E_m$. In contrast, we restrict attention to the case $\rho=0$ corresponding to Gaussian i.i.d. matrices with variance $m^{-1}$ and free circular elements, which we denote by $\{C_m^{(\nu)}\}_{1\leq \nu\leq m}$ and $\{c_\nu\}_{1\leq \nu\leq m}$ respectively. 
       
           We now describe the necessary modifications in the simplest case $k=1$, $\ell=1$. Following the notation in Theorem 11.2.1 in \cite{bose2021random}, let $a_m^{(l)}=\mu_{l}(d_m^{(1)},{d_m^{(1)}}^*)$ and $A_m^{(l)}=\mu_{l}(D_m^{(1)},{D_m^{(1)}}^*)$ denote monomials of degree $q_l$. Throughout the proof, our notation $m$, $k$, and $\mu_l$ correspond to $n$, $m$, and $m_l$, respectively, in \cite{bose2021random}. Steps 1--4 proceed exactly as in \cite{bose2021random}. The first modification arises in Step 5. In place of equation (11.13), we have by construction, for any even number $p=2l$, and any permutation $\sigma$ of $\{1,\ldots,2l\}$,
           \begin{align*}
               {(\varphi_m)}_{\sigma}(a_m^{(1)},\ldots, a_m^{(2l)})=\frac{1}{m}\Tr _{\sigma}(A_m^{(1)},\ldots,A_{m}^{(2l)}).
           \end{align*}
           Consequently, we obtain for $\epsilon_1,\ldots,\epsilon_{2l}\in \{1,*\}$,
           \begin{align*}
               \left|\frac{1}{m}\E\Tr ({C}_m^{\epsilon_1}A_m^{(1)}\ldots {C}_m^{\epsilon_{2l}}A_m^{(2l)})-\sum_{\pi\in\text{NC}_2(2l)}\kappa_{\pi}(c^{\epsilon_1},\ldots,c^{\epsilon_{2l}}){\varphi_{m}}_{\pi\gamma_{2l}}[a_m^{(1)},\ldots,a_m^{(2l)}]\right|\rightarrow 0.
           \end{align*}
           Finally, freeness of $c$ and $d_m^{(1)}$ follows from the same argument as in Step 6 of \cite{bose2021random}, and the extension to $k>1$ or $\ell>1$ follows as in Step 7. 
       
   \subsection{Proof of Lemma \ref{lem:deteqv.mmts}}\label{app:deteqv.mmts}
   
    Using $BX = 0$, then writing $B = \Gamma \Gamma^T$ for $\Gamma\in\mathbb{R}^{n\times n}$ and $\alpha_r = G_r \Sigma_r^{1/2}$ for $G_r \in \mathbb{R}^{I_s
     \times p}$ with i.i.d. standard normal entries, we have
   \begin{equation*}
     S_n = p^{-1} \check{\mathsf{Y}}^T  \check{\mathsf{Y}}, \qquad
      \check{\mathsf{Y}} = \sum_{r=1}^k \Gamma^T U_r G_r \Sigma_r^{1/2}.
   \end{equation*}
   In general both the rows and the columns of $\check{\mathsf{Y}}$ are
   correlated. If, however, we assume that $\Sigma_1, \ldots, \Sigma_k$
   are simultaneously diagonalizable, then we can transform
   $\check{\mathsf{Y}}$ to have independent columns. Indeed, there then
   exists an orthonormal matrix $V$, and non-negative diagonal matrices
   $L_1,\ldots,L_r$ such that  $\Sigma_r=VL_rV^T$ for each $r$.
   Now set $\check{\mathsf{Y}} = \mathsf{Y} V^T$, so that
   \begin{equation*}
     \mathsf{Y} =  \sum_r \Gamma^T U_r G_r V L_r^{1/2}
     \deq  \sum_r \Gamma_r^{1/2} X_r L_r^{1/2} 
   \end{equation*}
   has independent columns, 
   where $\Gamma_r = \Gamma^T U_r U_r^T \Gamma$ is non-negative definite
   and the independent matrices $X_r \in \mathbb{R}^{n \times p}$ have
   i.i.d. standard normal entries.
   Remark \ref{rem:different-dimensions} has extra detail on the equality in distribution.
   
   Now the ``companion'' matrix to $S_n$, namely
   \begin{equation*}
     B_n = p^{-1} \check{\mathsf{Y}} \check{\mathsf{Y}}^T
       = p^{-1} \mathsf{Y}  \mathsf{Y}^T
   \end{equation*}
   has precisely the form studied in \cite{xie2024cltlinearspectralstatistics}.
   
   Therefore, from Lemma 2.2 there, for each
   $z\in\mathbb{C}^+$, there exists unique $z$-dependent values
   $\deteqv{g}_i^{(r)}(z)$ with $z \deteqv{g}_1^{(r)}\in
   \overline{\mathbb{C}^+}$ and 
   $\deteqv{g}_2^{(r)} \in \mathbb{C}^+\cup\{0\}$,
   such that:
       \begin{align*}
       &\begin{cases}
           z\deteqv{g}_1^{(r)}(z)&=-\frac{1}{p}\Tr \left(({\sum_{s=1}^k}\deteqv{g}_2^{(s)}(z)L_s+\Id )^{-1}L_r\right),r=1,\ldots,k\\
           z\deteqv{g}_2^{(r)}(z)&=-\frac{1}{p}\Tr \left(({\sum_{s=1}^k}\deteqv{g}_1^{(s)}(z)\Gamma_s+\Id )^{-1}\Gamma_r\right),r=1,\ldots,k.
           \end{cases}
       \end{align*}
       The function
       $\tilde{m}_n^\circ:\mathbb{C}^+\rightarrow\mathbb{C}^+$
       given by\footnote{Note that we write $\tilde{m}_n^\circ$ and
     $\tilde{F}_n^\circ$ here for the quantities $\tilde{m}_n,
     \tilde{F}_n$ in \cite{xie2024cltlinearspectralstatistics}. In \cite{xie2024cltlinearspectralstatistics} the primary focus is on the spectrum
     of $B_n$, whereas here our interest is in the companion matrix $S_n$.
   }
       \begin{equation*}
         z\deteqv{m}_n^\circ(z)=-\frac{1}{n}\Tr \left({\textstyle\sum}_{r=1}^k\deteqv{g}_1^{(r)}(z)\Gamma_r+\Id  \right)^{-1}
       \end{equation*}
       defines the Stieltjes transform of a probability measure
       $\tilde{F}_n^\circ$ on $\mathbb{R}$ such that almost surely
       \begin{align*}
           F^{B_n}(x)-\deteqv{F}_n^\circ(x)\rightarrow 0
       \end{align*}
       at each $x \in \R$.
       Moreover,
           \begin{align*}
           &\iain{m_n^\circ(z)-\deteqv{m}_n^\circ(z)}\rightarrow 0
           \end{align*}
           pointwise almost surely for all $z \in \mathbb{C}^+$,
       where $m_n^\circ(z)$ is the Stieltjes transform of $F^{B_n}$. 
       
       Under the construction described above, we have
       \begin{align*}
         F^{S_n} = \left(1-\frac{n}{p}\right)I_{[0,\infty)}+
         \frac{n}{p} \iain{F^{B_n}.}
       \end{align*}
       Define the deterministic equivalent law $\deteqv{F}_n$ of $F^{S_n}$ as 
       \begin{align*}
         \deteqv{F}_n = \left(1-\frac{n}{p}\right)I_{[0,\infty)}+
         \frac{n}{p}\iain{\deteqv{F}_n^\circ}
       \end{align*}
       Then the Stieltjes transform of $\deteqv{F}_n$, denoted $\deteqv{m}_n$, is
       \begin{align*}
           \deteqv{m}_n =
         -\frac{1-n/p}{z}+\frac{n}{p}\iain{\deteqv{m}_n^\circ},
       \end{align*}
       and
       \begin{equation*}
         m_n(z) - \tilde{m}_n(z) = \frac{n}{p} [m_n^\circ(z) -
         \tilde{m}_n^\circ(z)].
       \end{equation*}
       Finally, let
       $\iain{\hat{\mu}_n^\circ}=(\frac{1}{p} \iain{ \Tr
         B_n,\ldots,\frac{1}{p}\Tr B_n^D})$, then we easily have
         that
       \begin{align*}
           \hat{\mu}_n= \iain{\hat{\mu}_n^\circ}.
       \end{align*} 
       As a result, the convergence of moments $\hat{\mu}_n-\deteqv{\mu}_n\rightarrow0$ is an intermediate result that appears in the proof of Corollary 3.10 in \cite{fan2017eigenvalue}, and follows directly from the definition of asymptotic equivalence in $\mathcal{D}$-law given in Definition 3.4 of \cite{fan2017eigenvalue}.
       \subsection{Proof of Lemma \ref{lem:nested.indep}}\label{subsec:proof.nested.indep}
       First, we introduce a useful lemma.
       \begin{lem}\label{lem:multi.normal.indep}
           Let $Z\in\mathbb{R}^{m\times n}$ have i.i.d. standard normal entries. For deterministic matrices $A_1,A_2,B_1,B_2$ of conformable sizes, define
           \begin{align*}
               X=A_1 Z B_1,\quad Y=A_2 Z B_2.
           \end{align*}
           Then $X$ and $Y$ are independent if and only if
           \begin{align*}
               A_1A_2^T=0 \quad\text{or}\quad B_1^T B_2=0.
           \end{align*}
       \end{lem}
       \begin{proof}
           By the vectorization identity,
           \[
           \operatorname{vec}(X)=(B_1^\top\otimes A_1)\operatorname{vec}(Z),
           \qquad
           \operatorname{vec}(Y)=(B_2^\top\otimes A_2)\operatorname{vec}(Z).
           \]
           Hence
           \[
           \begin{pmatrix}
           \operatorname{vec}(X)\\
           \operatorname{vec}(Y)
           \end{pmatrix}
           =
           \begin{pmatrix}
           B_1^\top\otimes A_1\\
           B_2^\top\otimes A_2
           \end{pmatrix}
           \operatorname{vec}(Z),
           \]
           so $(\operatorname{vec}(X),\operatorname{vec}(Y))$ is jointly Gaussian with mean zero.
           Its cross-covariance is
           \begin{align*}
           \operatorname{Cov}(\operatorname{vec}(X),\operatorname{vec}(Y))
           &=(B_1^\top\otimes A_1)\,
           \operatorname{Cov}(\operatorname{vec}(Z))\,
           (B_2^\top\otimes A_2)^\top\\
           &=(B_1^\top\otimes A_1)(B_2\otimes A_2^\top)\\
           &=(B_1^\top B_2)\otimes(A_1A_2^\top),
           \end{align*}
           using $\operatorname{Cov}(\operatorname{vec}(Z))=\Id$ and the mixed-product property
           $(P\otimes Q)(R\otimes S)=(PR)\otimes(QS)$.
           Since $(\operatorname{vec}(X),\operatorname{vec}(Y))$ is jointly Gaussian,
           $X$ and $Y$ are independent if and only if
           \[
           \operatorname{Cov}(\operatorname{vec}(X),\operatorname{vec}(Y))=0,
           \]
           equivalently,
           \[
           (B_1^\top B_2)\otimes(A_1A_2^\top)=0.
           \]
           \end{proof}
           
           \begin{proof}[Proof of  Lemma \ref{lem:nested.indep}]
            Write $B_r =  Q_r Q_r^T$ with $Q_r$ chosen so that $Q_r^T Q_r$ is
            invertible. Assumption \ref{assum:nested.indep}(ii) is then equivalent
            to
            \begin{equation}
              \label{eq:4ii-equiv}
              Q_r^T U_t U_t^T Q_s = 0 \qquad \text{for } \qquad 1 \leq r < s \leq
              t \leq k.
            \end{equation}
            Independence of the matrices
                   $\{S_{nr}\}_{r=1}^k$ can be reduced to independence of the projected data
                   $\{Q_r^TY^{(r)}\}_{r=1}^k$, where $Y^{(r)}=\sum_{t=r}^k
                   U_tZ_t\Sigma_t^{1/2}$, and the $Z_t$ are independent matrices
                   with i.i.d. standard normal entries.
           Suppose without loss of generality that $r < s$.
            It therefore suffices to show, for every $t \geq r$ and $s > r$ that
            $Q_r^TU_tZ_t\Sigma_t^{1/2} \ \perp\!\!\!\perp\ Q_s^TU_tZ_t\Sigma_t^{1/2}.$
                   By Lemma~\ref{lem:multi.normal.indep}, this holds whenever
                   $Q_r^TU_tU_t^TQ_s=0$.
            If $r \leq t < s$, then $B_s U_t = Q_s Q_s^T U_t = 0$ by Assumption
            \ref{assum:nested.indep}(i), and so $Q_s^TU_t = 0$, since $Q_s^T Q_s$
            is invertible.
            When $ t \geq s$, the required condition follows from
            (\ref{eq:4ii-equiv}).   
            \end{proof}
            
       \subsection{Proof of Lemma \ref{lem:nested}}\label{subsec:proof.nested}
   
   By the nested structure, we have
   \begin{align*}
       \left|\text{det}(\mathcal{F}_n)\right|=\left|\prod_r\text{det}\left({\mathcal{F}}_n^{(r)}\right)\right|
   \end{align*}
   is also lower bounded away from zero. Following similar arguments in the proof of Theorem \ref{thm:main}, we have that for large $n$, almost surely, the moment estimator $\hat{\theta}_n'$ exist and converges to $\theta$. 
   
   From \eqref{eq:nested.clt.mu}, we have
   \begin{align*}
       {P_{n}}^{-1/2}\left(p(\hat{\mu}_{n}-\deteqv{\mu}_{n}) - R_{n}\right)\convd N(0,\Id _{D}),
   \end{align*} 
   where
   \begin{align*}
       R_n = \begin{bmatrix} R_{n1}^T &\ldots & R_{nk}^T \end{bmatrix}^T \in \mathbb{R}^D, \quad
       P_n = \text{diag}(P_{n1}, \dots, P_{nk}) \in \mathbb{R}^{D \times D}.
   \end{align*}
   Further applying Lemma \ref{lem:delta.var}, we conclude the proof with $Q_n=D_nR_n+p(\mathcal{F}_n^{-1}(\deteqv{\mu}_n)-\theta)$.

       \subsection{Asymptotic normality of MoM estimators when oracle eigenvectors are known}\label{app:oracle.clt}
       Consider the class of nested models defined in Section~\ref{sec:sequential.nested} with no intercept, i.e., $X=0$. Then the design matrices satisfy
       \begin{align*}
           &I_1\leq ...\leq I_k\leq n,\quad\text{col}(U_1)\subset \ldots \subset\text{col}(U_k).
           \end{align*}
       Furthermore, assume that each $U_r$ is a deterministic, block-diagonal incidence matrix with nonnegative entries and full column rank. For example, $U_r$ may represent a group assignment matrix that assigns each individual to the corresponding group. That is, $U_r[i,j]=1$ if the $i$th individual is assigned to the $j$th group, and $U_r[i,j]=0$ otherwise. 
       Define the projection matrix
       \begin{align*}
           B_1 &=U_1(U_1^TU_1)^{-1}U_1^T,
       \end{align*}
       and consider the quadratic form matrix
       \begin{align*}
           S_n = \frac{1}{p}Y^TB_1Y.
       \end{align*}
       We solve for the method of moments estimators $\hat{\theta}_n$ following the joint estimation procedure outlined in \eqref{eq:emp.fix.oracle.main}, with the method of moments estimator $\hat{\theta}_n$ defined as the solution to the moments equation
       \begin{align*}
           \hat{\mu}_n = \mathcal{F}_n(\hat{\theta}_n).
       \end{align*}
       To establish the asymptotic normality of $\hat{\theta}_n$, we adopt a similar approach to the proof of Theorem \ref{thm:main}. Specifically, we first establish a central limit theorem for the moments $\hat{\mu}_n$, from which the asymptotic distribution of the method-of-moments estimator $\hat{\theta}_n'$ follows via the Delta method. The difference here is that the deterministic equivalent law of $S_n$ and the corresponding central limit theorem of $\hat{\mu}_n$ are derived leveraging the nested structure of the model and the particular form of the matrix $B_1$, and are valid without relying on the simultaneous diagonalizability assumption. Here, we derive the asymptotic normality of $\hat{\mu}_n$. The application of the Delta method is analogous to the proof given in Section \ref{subsec:proof.clt}. 
       
       By the nature of the incidence matrices, we have that $U_1^TU_1$ is diagonal. Under the nested model, we further have that $U_1^TU_r$ has full row rank and that $U_1^TU_rU_r^TU_1$ is diagonal and positive definite. Additionally, considering a matrix $A\in\mathbb{R}^{m_1\times m_2}$, $m_1<m_2$ and a standard Gaussian vector $x\in\mathbb{R}^{m_2}$, there exists a standard Gaussian vector $z\in\mathbb{R}^{m_1}$ such that  $Ax\deq (AA^T)^{1/2}z$.
       As a result, by Gaussianity of the random effects, 
       \begin{align*}
           (U_1^TU_1)^{-1/2}U_1^TU_r\alpha_r\deq L_r^{1/2}X_r\Sigma_r^{1/2},
       \end{align*}
       where $L_r^{1/2}=(U_1^TU_1)^{-1/2}(U_1^TU_rU_r^TU_1)^{1/2}$ is diagonal and positive definite, and $X_r\in\mathbb{R}^{I_1\times p}$ has i.i.d. standard Gaussian entries. Thus, we can express the sum of squares matrix as
       \begin{align*}
           S_n\deq\frac{1}{I_1}{\sum_{r,s=1}^k\Sigma_{r}^{1/2}}X_{r}^TL_r^{1/2}L_s^{1/2}X_{s}\Sigma_{s}^{1/2}.
       \end{align*}
       Let $T_{j}=\sum_{r=1}^kl_{j,r}\Sigma_r$ and denote the $j$th row of $X_r$ by $x_{j,r}$, then by the independence of the random effects there further exist i.i.d. standard Gaussian vectors 
       $x_j\in\mathbb{R}^p$ such that $\sum_{r=1}^kl_{j,r}\Sigma_r^{1/2}x_{j,r}\deq T_{j}^{1/2}x_j$. Therefore, we can alternatively express $S_n$ as
       \begin{align*}
           S_n&\deq\frac{1}{I_1}\sum_{j=1}^{I_1} T_{j}^{1/2}x_jx_j^TT_{j}^{1/2}, \quad T_{j}=\sum_{r=1}^kl_{j,r}\Sigma_r.
       \end{align*}
       In particular, $S_n$ coincides in form with the random matrix model analyzed in \cite{xie2024cltlinearspectralstatistics}. Theorem~2.3 of that reference therefore yields a central limit theorem for the moments $\hat{\mu}_n$ analogous to Theorem~\ref{thm:clt}, upon replacing $\Gamma_r$ by $\Sigma_r$ and $\lambda_{j,r}$ by $l_{j,r}$.

       \section{Additional details on software}
   
       The code is available at \url{https://github.com/ran-xie/mlmm-mom}. The implementation has three main components. First, to evaluate the map $\Psi$ constructed in Lemma~\ref{lem:map.asym}, we developed routines to enumerate noncrossing pair partitions and compute the Kreweras complement of a given noncrossing pair partition. These combinatorial routines underlie the evaluation of $\Psi_n$, and hence of the map $\mathcal{F}_n$.
   
       Second, the method-of-moments estimator $\hat{\theta}_n$ is obtained by solving the nonlinear system
       \begin{align*}
           \hat{\mu}_n=\mathcal{F}_n(\hat{\theta}_n).
       \end{align*}
       Because the raw moments $\hat{\mu}_n^{(\ell)}$ grow polynomially in $\|p^{-1}S_n\|$, direct solution of this system may be numerically unstable. To improve stability, we introduce the transformation $\mathcal{H}$ and define
       \begin{align*}
           \widehat{\widehat{\mu}}_n
           =\mathcal{H}(\hat{\mu}_n)
           =\left(\hat{\mu}_{n}^{(1)},
           \frac{\hat{\mu}_{n}^{(2)}}{\hat{\mu}_{n}^{(1)}},
           \ldots,
           \frac{\hat{\mu}_{n}^{(D)}}{\hat{\mu}_{n}^{(D-1)}}\right).
       \end{align*}
       We then solve the transformed system $\widehat{\widehat{\mu}}_n-\widehat{\mathcal{F}}_n(\hat{\theta}_n)=0$, 
       where $\widehat{\mathcal{F}}_n=\mathcal{H}\circ\mathcal{F}_n$. This transformation improves numerical stability and yields more reliable convergence of the root-finding procedure. In the simulation experiments, the algorithm is initialized at a uniform vector. In practical applications, when prior information on the underlying parameters is available, a more informed initialization may substantially improve computational efficiency.
   
       Finally, the asymptotic bias and covariance of the MoM estimators are evaluated using the Python package \textit{spectral-clt}, available at \url{https://github.com/ran-xie/spectral-clt}.

       \section{Additional details on experiments}\label{app:experiments}
       
       \subsection{Asymptotic properties of MoM estimators under Scenario 1} 
       We first establish the asymptotic boundedness of the functionals $\varphi(F_1^l F_2^s)$ for $l, s \geq 0$. Leveraging the specific structure of $F_1$ and $F_2$, and the boundedness of the $J_i$s, for $l,s\geq 1$, we can derive the following formulas for the coefficients:
       \begin{align}
           &\varphi(F_1^l)=\frac{1}{p}\sum_i(J_i-n^{-1}J_i^2)^l=\frac{1}{p}\sum_iJ_i^l+o_n(1),\quad\varphi(F_2^l)=\varphi(B_1)=\frac{I-1}{p},\nonumber\\
           &\varphi(F_2^lF_1^s)=\varphi(B_1^{l+1}U_1U_1^TF_1^{s-1})=\varphi(F_1^s)=\frac{1}{p}\sum_iJ_i^s+o_n(1).\label{eq:f1f2.main}
       \end{align} 
       Under the assumption that $\{J_i\}$ are independent and identically distributed with bounded moments, and letting $\tilde{c} := C/c$, the expressions derived in \eqref{eq:f1f2.main} and the strong law of large numbers imply the following almost sure bounds as $n \to \infty$:
       \begin{align*}
           \E[J_1^l]\tilde{c}^{-1}\leq \varphi(F_1^l)\leq \E[J_1^l]\tilde{c},\quad \tilde{c}^{-1}\leq\varphi(F_2^l)\leq \tilde{c},\quad \E[J_1^s]\tilde{c}^{-1}\leq\varphi(F_2^lF_1^s)\leq \E[J_1^s]\tilde{c}.
       \end{align*}
       Since $\Bar{\mathcal{F}}_n^{(1)}$ and $\Bar{\mathcal{F}}_n^{(2)}$ are polynomial functions of $\theta_1$ and $\theta_2$, respectively, it follows that the Jacobian determinants $\lvert \det(J \mathcal{F}_n^{(1)}) \rvert$, $\lvert \det(J \mathcal{F}_n^{(2)}) \rvert$ are bounded away from zero for sufficiently large $n$. Consequently, applying Lemma~\ref{lem:nested}, we obtain
       \begin{enumerate}
               \item for large $n$, almost surely, the moment estimator $\hat{\theta}_n'=(\hat{\pi},\hat{\tau}_1,\hat{\tau}_2)$ exist and converges to $\theta$,
           \item let $\mathcal{F}_n=(\mathcal{F}_n^{(1)},\mathcal{F}_n^{(2)})$, for $D_n:={(J\mathcal{F}_n(\theta))^{-1}}$, as $n\rightarrow \infty$,
           \begin{align}\label{eq:clt.fullsib.main}
               (D_nP_nD_n^T)^{-1/2}\left[p\left(\hat{\theta}_n'-\theta\right)-Q_n\right]\convd N(0,\Id_3)
           \end{align}
            for some $P_n\in\mathbb{R}^{3\times 3}$, $Q_n\in\mathbb{R}^3$.
           \end{enumerate}
       To characterize the structure of $P_n$ and $Q_n$, we define the vector of empirical moments
       $$
       \hat{\mu}_n'=(p^{-1}\Tr S_{n1},p^{-1}\Tr S_{n1}^2,p^{-1}\Tr S_{n2}),$$ 
       and recall that $\mathcal{F}_n=(\mathcal{F}_n^{(1)},\mathcal{F}_n^{(2)})$ denotes the mapping from $\hat{\theta}_n'$ to $\hat{\mu}_n'$. Our strategy involves first establishing a central limit theorem for the moments $\hat{\mu}_n'$, from which the asymptotic distribution of the method-of-moments estimator $\hat{\theta}_n'$ follows via the Delta method. As detailed in Section \ref{subsec:deteqv}, the empirical spectral distribution $F^{S_{n2}}$ admits a deterministic equivalent law $\tilde{F}_{n2}$, given by
       \begin{align*}
           \deteqv{F}_{n2} = \left(1-\frac{n}{p}\right)I_{[0,\infty)}+ \frac{n}{p}\deteqv{F}_{n2}',
       \end{align*}
       where $\deteqv{F}_{n2}'$ is a probability measure whose Stieltjes transform $\deteqv{m}_{n2}(z)$ is determined by the system of equations:
       \begin{align*}
       &\quad z\deteqv{m}_{n2}=-\frac{1}{n}\Tr \left(\deteqv{g}_1B_2+\Id  \right)^{-1},\\ 
       &\begin{cases}
           z\deteqv{g}_1&=-\frac{1}{N}\Tr \left((\deteqv{g}_2\Sigma_2+\Id )^{-1}\Sigma_2\right)\\
           z\deteqv{g}_2&=-\frac{1}{N}\Tr \left((\deteqv{g}_1B_2+\Id )^{-1}B_2\right)
           \end{cases}
       \end{align*}
       Similarly, the deterministic equivalent law $\deteqv{F}_{n1}$ for $F^{S_{n1}}$ is defined as
       \begin{align*}
           \deteqv{F}_{n1} = \left(1-\frac{n}{p}\right)I_{[0,\infty)}+ \frac{n}{p}\deteqv{F}_{n1}',
       \end{align*}
       where the Stieltjes transform $\deteqv{m}_{n1}(z)$ of $\deteqv{F}_{n1}'$ satisfies
       \begin{align}
       &\quad z\deteqv{m}_{n1}=-\frac{1}{n}\Tr \left(\deteqv{g}_1^{(1)}\Gamma_1+\deteqv{g}_1^{(2)}\Gamma_2+\Id  \right)^{-1},\\ 
       &\begin{cases}
           z\deteqv{g}_1^{(r)}&=-\frac{1}{N}\Tr \left((\deteqv{g}_2^{(1)}\Sigma_1+\deteqv{g}_2^{(2)}\Sigma_2+\Id )^{-1}\Sigma_r\right),r=1,2\\
           z\deteqv{g}_2^{(r)}&=-\frac{1}{N}\Tr \left((\deteqv{g}_1^{(1)}\Gamma_1+\deteqv{g}_1^{(2)}\Gamma_2+\Id )^{-1}\Gamma_r\right),r=1,2
           \end{cases},
       \end{align}
       with $\Gamma_1=B_1U_1U_1^TB_1$, $\Gamma_2=B_1U_2U_2^TB_1=B_1$. From $\tilde{F}_{n1}$ and $\tilde{F}_{n2}$, we define the deterministic equivalent moments:
       \begin{align*}
           \deteqv{\mu}_{n}^{(1)} &= \left(\int xd\deteqv{F}_{n1}(x),\int x^2d\deteqv{F}_{n1}(x)\right),\\
          \deteqv{\mu}_{n}^{(2)} &= \int xd\deteqv{F}_{n2}(x),
           \quad \deteqv{\mu}_n=(\deteqv{\mu}_{n}^{(1)},\deteqv{\mu}_{n}^{(2)}).
       \end{align*}
       By Theorem \ref{thm:clt}, the empirical moments satisfy a central limit theorem: there exist $R_{n1}\in\mathbb{R}^2$, $P_{n1}\in\mathbb{R}^{2\times 2}$, and scalars $R_{n2}$, $P_{n2}$, such that as $n\rightarrow\infty$,
           \begin{align*}
           &P_{n1}^{-1/2}\left(p(\hat{\mu}_{n}^{(1)}-\deteqv{\mu}_{n}^{(1)}) - R_{n1}\right)\xrightarrow{d}  N(0,\Id_2),\\
           &P_{n2}^{-1/2}\left(p(\hat{\mu}_{n}^{(2)}-\deteqv{\mu}_{n}^{(2)}) - R_{n2}\right)\xrightarrow{d}  N(0,1).
           \end{align*}
       Since $B_1B_2=0$, by Gaussianity we have that $S_{n1}$ and $S_{n2}$ are independent. Defining
       \begin{align}\label{eq:fullsib.rn.pn}
           R_n=\begin{bmatrix}
               R_{n1}\\
               R_{n2}
           \end{bmatrix},\quad P_n=\begin{bmatrix}
               P_{n1}&0\\
               0&P_{n2}
           \end{bmatrix},
           \end{align}
       it follows that  $P_{n}^{-1/2}\left(p(\hat{\mu}_{n}-\deteqv{\mu}_{n}) - R_{n}\right)\rightarrow N(0,\Id_3)$.
       Finally, by the Delta method, we conclude that (\ref{eq:clt.fullsib.main}) holds with
       $$
       Q_n=D_nR_n+p(\mathcal{F}_n^{-1}(\deteqv{\mu}_n)-\theta).
       $$

       \subsection{Asymptotic properties of MoM estimators under Scenarios 2, 3, 5}
       In Section~\ref{app:oracle.clt}, we derived a central limit theorem for the moments $\hat{\mu}_n$ in the nested model without intercept, without relying on simultaneous diagonalizability. Specializing that argument to the full-sibling model establishes the CLT for $\hat{\mu}_n$ under Scenarios~2, 3, and~5: there exist $P_n\in\mathbb{R}^{3\times 3}$ and $R_n\in\mathbb{R}^{3}$ such that, as $n\rightarrow\infty$, $ P_n^{-1/2}\left(p(\hat{\mu}_n-\deteqv{\mu}_n) - R_n\right)\rightarrow N(0,\Id_3).$

       In Scenario~2, simultaneous diagonalizability holds, so a CLT can alternatively be obtained from the general framework in Section~\ref{subsec:clt}.
       
       Under Scenario~3, the Delta method differs slightly from the general development in Section~\ref{subsec:proof.clt}. The model is misspecified in that $\mathcal{F}_n$ is constructed following the procedure in Section \ref{sec:main.results} under the assumption that $\Sigma_1$ and $\Sigma_2$ are simultaneously diagonalizable, but in reality they are not simultaneously diagonalizable. Hence
       \begin{align*}
           \hat{\mu}_n\neq_{a.s} \mathcal{F}_n(\theta).
       \end{align*}
       Consequently, $\hat{\theta}=\mathcal{F}_n^{-1}(\hat{\mu}_n)$ is inconsistent for $\theta$. By Lemma~\ref{lem:deteqv.mmts},
       \begin{align*}
           (\hat{\theta}-\theta)-(\mathcal{F}_n^{-1}(\deteqv{\mu}_n)-\theta) \rightarrow 0\quad \text{a.s.},
       \end{align*}
       where the bias correction $\mathcal{F}_n^{-1}(\deteqv{\mu}_n)-\theta$ need not vanish.
       For the same reason, $D_n':={J\mathcal{F}_n^{-1}(\deteqv{\mu}_n)}$ 
       need not converge to {$D_n=(J\mathcal{F}_n(\theta))^{-1}$}, as in Theorem~\ref{thm:main}. Consequently, under Scenario~3 the method-of-moments estimator $\hat{\theta}_n$ satisfies $(D_n'P_nD_n'^T)^{-1/2}\left(p(\hat{\theta}_n-\theta)-Q_n\right)\rightarrow N(0,\Id_3),$
       where $Q_n$ is given by $$Q_n=D_n'R_n+p(\mathcal{F}_n^{-1}(\deteqv{\mu}_n)-\theta).$$
    \subsection{Empirical and theoretical correlation matrix between the MoM estimators} \label{subsec:corr}
    In Table \ref{tab:corr}, we report the empirical and theoretical correlation matrices between the MoM estimators for Scenarios 1, 2, 3, and 5. We observe that both in the sequential estimation procedure and the nested estimation procedure, the MoM estimators of parameters in different levels of variance components are correlated.

       \begin{table}[ht]
\centering
\caption{Empirical and theoretical correlation matrices of the MoM estimators $(\hat\tau_1,\hat\rho,\hat\tau_2)$.}
\label{tab:corr}
\renewcommand{\arraystretch}{1.1}

\begin{tabular}{cc}
\multicolumn{2}{c}{\textbf{Scenario 1}} \\[0.3em]

Empirical &
Theoretical \\

$\begin{pmatrix}
1.000 & -0.636 & -0.066\\
 & 1.000 & -0.275\\
 &  & 1.000
\end{pmatrix}$
&
$\begin{pmatrix}
1.000 & -0.653 & -0.058\\
 & 1.000 & -0.303\\
 & & 1.000
\end{pmatrix}$
\\[2em]

\multicolumn{2}{c}{\textbf{Scenario 2}} \\[0.3em]

Empirical &
Theoretical \\

$\begin{pmatrix}
    1.0& -0.1243& -0.0066\\
    & 1.0& -0.9724\\
    & & 1.0    
\end{pmatrix}$
&
$\begin{pmatrix}
    1.0& -0.0818& -0.0639\\
    & 1.0& -0.9718\\
    & & 1.0
\end{pmatrix}$
\\[2em]

\multicolumn{2}{c}{\textbf{Scenario 3}} \\[0.3em]

Empirical &
Theoretical \\

$\begin{pmatrix}
    1.0& -0.6424& 0.4306\\
    & 1.0& -0.9291\\
    & & 1.0   
  \end{pmatrix}$
&
$\begin{pmatrix}
    1.0& -0.6283& 0.4051\\
    & 1.0& -0.9255\\
    & & 1.0
\end{pmatrix}$
\\[2em]

\multicolumn{2}{c}{\textbf{Scenario 5}} \\[0.3em]

Empirical &
Theoretical \\

$\begin{pmatrix}
    1.0& -0.6805& 0.5289\\
    & 1.0& -0.9549\\
    & & 1.0
    \end{pmatrix}$
&
$\begin{pmatrix}
    1.0& -0.6579& 0.4993\\
    & 1.0& -0.953\\ 
    & & 1.0
\end{pmatrix}$

\end{tabular}

\end{table}

       \section{Simulations under the two-way unbalanced model}\label{app:half.sib}
       Consider an experimental design in which each of the $I_1$ sires is mated with $M$ dams, so that the total number of dams is $I_2=I_1M$. The $m$th dam mated with the $i$th sire produces $J_{im}$ offspring, yielding a total sample size $n=\sum_{i=1}^{I_1}\sum_{m=1}^{M}J_{im}$. Let $J_i=\sum_{m=1}^{M}J_{im}$ denote the total number of offspring from sire $i$. For each offspring, we observe $p$ phenotypic traits. The observed trait matrix $Y\in\mathbb{R}^{n\times p}$ is assumed to follow the model
   \begin{align}\label{eq:Y}
       Y=X\beta+U_1\alpha_1+U_2\alpha_2+U_3\alpha_3,
       \qquad
       \alpha_r\sim N(0,\Id_{I_r}\otimes\Sigma_r).
   \end{align}
   Here $X=\mathbbm{1}_n$ is the fixed-effects design matrix, $\beta=\mu^T$ is the fixed effect,
   \[
   U_1=\diag(\mathbbm{1}_{J_1},\ldots,\mathbbm{1}_{J_{I_1}})\in\mathbb{R}^{n\times I_1}
   \]
   assigns each offspring to its sire,
   \[
   U_2=\diag(\mathbbm{1}_{J_{11}},\ldots,\mathbbm{1}_{J_{1M}},\ldots,\mathbbm{1}_{J_{I_11}},\ldots,\mathbbm{1}_{J_{I_1M}})\in\mathbb{R}^{n\times I_2}
   \]
   assigns each offspring to its dam, and $U_3=\Id_n$ corresponds to the individual-level error. The matrices $\alpha_r\in\mathbb{R}^{I_r\times p}$ represent the random effects at the three levels.
   
   In the simulation study, we take $I_1=100$ sires, each mated with $M=2$ dams, and observe $p=200$ phenotypic traits. The offspring counts $\{J_{im}\}$ are modeled as i.i.d. random variables uniformly distributed on $\{1,2\}$: $J_{im}\sim \Unif\{1,2\}$.

   We assume that the eigenvalues of $\Sigma_1$ follow an exponential decay parameterized by $\theta_1=(\tau_1,\rho_1)$, those of $\Sigma_2$ follow a linear decay parameterized by $\theta_2=(\tau_2,\rho_2)$, and those of $\Sigma_3$ follow a step-function model parameterized by $\theta_3=(\tau_3,\tau_4,\rho_3)$:
   \begin{align*}
       g_1(x)&=\tau_1\exp(-\rho_1x), \quad g_2(x)=\tau_2(1-\rho_2x)_+, \\
       g_3(x)&=\tau_3\mathbbm{1}[0\le x\le \rho_3]+\tau_4\mathbbm{1}[\rho_3<x\le 1].
   \end{align*}
   The true parameter values are set to
   \begin{align}\label{eq:halfsib.params}
       \tau_1=2,\quad \rho_1=8,\quad \tau_2=1,\quad \rho_2=3,\quad
       \tau_3=1,\quad \tau_4=0,\quad \rho_3=0.2.
   \end{align}
   Figure~\ref{fig:half.sib.sigma.eigs} displays the corresponding functions $g_1$, $g_2$, and $g_3$.
   
       \begin{figure}
           \centering
           \includegraphics[width=0.8\linewidth]{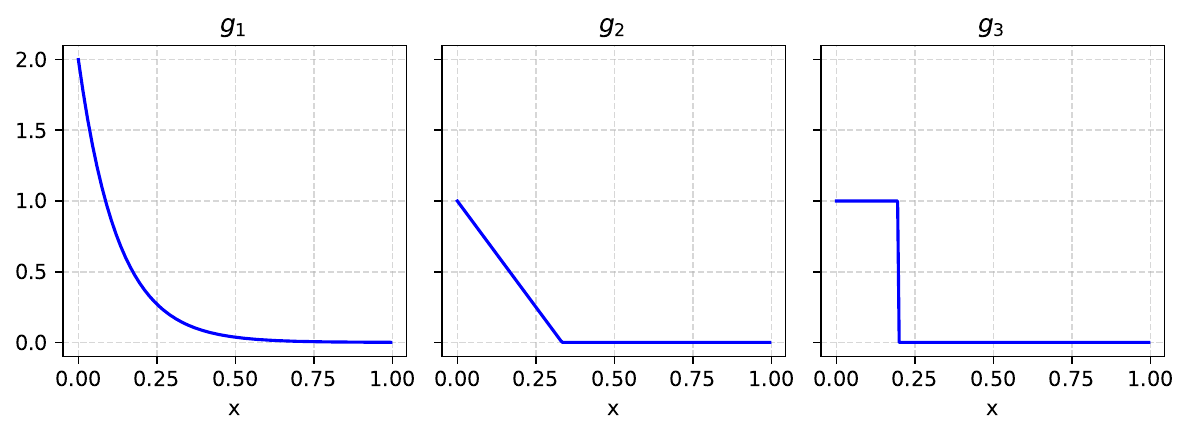}
           \caption{Exponential decay function $g_1$, linear decay function $g_2$ and step decay function $g_3$ with true parameters set in (\ref{eq:halfsib.params}).}
           \label{fig:half.sib.sigma.eigs}
       \end{figure}
   
       We estimate $\theta_3$, $\theta_2$, and $\theta_1$ sequentially following the procedure outlined in Section~\ref{sec:sequential.nested}. As in Section~\ref{subsec:simu.main}, we repeat the experiment 1000 times and divide the resulting estimators into 20 groups. Figure~\ref{fig:hist.fullsibhalfsib} shows histograms of the MoM estimators. Table~\ref{tab:table1} compares the empirical bias and standard deviation of the estimators with the corresponding theoretical values predicted by the central limit theorem. Figure~\ref{fig:normal.approx} presents $99\%$ normal-approximation confidence intervals for the empirical bias and standard deviation. 
       
       Among the parameters corresponding to the largest eigenvalue at each level, $\tau_3$ is estimated most accurately, followed by $\tau_2$ and then $\tau_1$.  This pattern is consistent with our earlier findings in the full sibling model simulations in Section \ref{subsec:simu.main} and reflects two key features of the parameter estimation procedure. First, estimation at inner levels requires separating the cumulative contributions of multiple random-effect components and is therefore more difficult. Second, because of the nested structure, inner levels involve fewer effective random variables, which leads to greater variability. Additionally, the estimator of $\rho_1$ exhibits substantially larger bias and standard deviation than the other estimators, likely in part because of the relatively large magnitude of $\rho_1$.

       Overall, the theoretical and empirical biases and standard deviations are in close agreement. Moreover, for each parameter, the theoretical bias and standard deviation lie within the corresponding $99\%$ confidence intervals constructed from the empirical estimates.

       \begin{figure}
           \centering       
           \includegraphics[width=\linewidth]{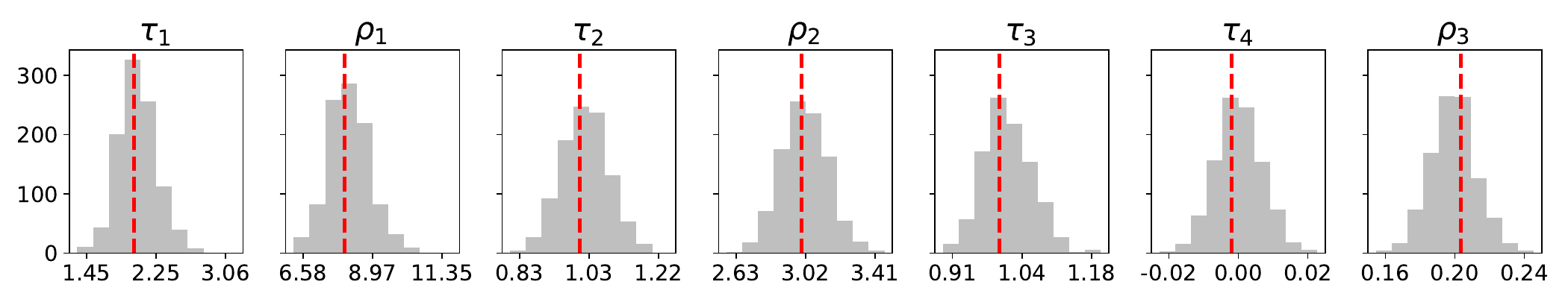}
           \caption{Histogram of MoM estimators, with true parameter value marked in red. }
           \label{fig:hist.fullsibhalfsib}
       \end{figure}

       \begin{table}
           \centering
            \caption{Theoretical and empirical bias and standard deviations of the MoM estimators. }
           \begin{tabular}{c c c c c c c c c}
               & & $\tau_1$ & $\rho_1$ & $\tau_2$ & $\rho_2$ & $\tau_3$ & $\tau_4$ & $\rho_3$\\
                \hline
               \multirow{2}{*}{bias}&empirical &0.041	&0.1627	&0.0148	&0.0178	&0.0172	&0.0018	&-0.0042 \\
               &theoretical&0.0572	&0.2373	&0.0174	&-0.0077	&0.0195	&0.0025	&-0.0095\\
               \multirow{2}{*}{sd}&empirical&0.2269	&0.7334	&0.0665	&0.1306	&0.0492	&0.0062	&0.0137 \\
               &theoretical&0.222	&0.7096	&0.0667	&0.13	&0.047	&0.0061	&0.0131

           \end{tabular}
         \label{tab:table1}
       \end{table}
       
       \begin{figure}
           \centering
           \includegraphics[width=0.57\linewidth]{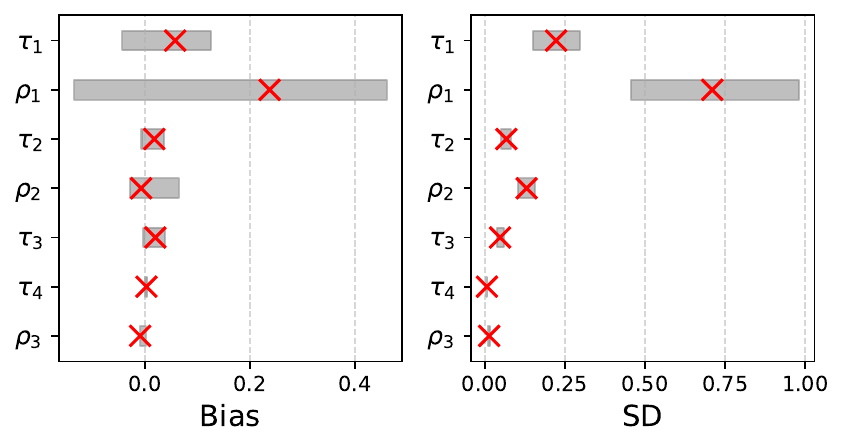}
           \caption{Normal approximation 99\% confidence intervals for the bias and standard deviation of method of moments estimators, computed using 20 empirical estimates, each based on 50 MoM estimators. The theoretical bias and standard deviations of the MoM estimators are marked in red.}
           \label{fig:normal.approx}
       \end{figure}

       \end{appendix}
       

\end{document}